\documentclass{article}
\usepackage{geometry}
\usepackage{titlesec}
\usepackage[T1]{fontenc}
\usepackage[utf8]{inputenc}
\usepackage{authblk}
\usepackage{graphics}
\providecommand{\keywords}[1]{\textbf{Keywords: } #1}
\titleformat{\section}{\centering\large\scshape}{\thesection}{1em}{}
\titleformat{\subsection}{\centering\normalsize\scshape}{\thesubsection}{1em}{}

\usepackage{hyperref}
\hypersetup{
	colorlinks=true,
	linkcolor=blue,
	filecolor=blue,      
	urlcolor=blue,
	citecolor=blue
}

\usepackage{natbib}
\usepackage{graphicx} %Loading the package
\graphicspath{{art/}}

\usepackage{setspace}
\usepackage{xr}
\usepackage{color}
\usepackage{amsmath}
\usepackage{amsfonts,dsfont,mathrsfs}
\usepackage{booktabs,multirow,array}
\usepackage{bm}
\usepackage[plain,noend]{algorithm2e}

\makeatletter
\renewcommand{\algocf@captiontext}[2]{#1\algocf@typo. \AlCapFnt{}#2} % text of caption
\def\@algocf@capt@plain{top}
\renewcommand{\algocf@makecaption}[2]{%
	\addtolength{\hsize}{\algomargin}%
	\sbox\@tempboxa{\algocf@captiontext{#1}{#2}}%
	\ifdim\wd\@tempboxa >\hsize%     % if caption is longer than a line
	\hskip .5\algomargin%
	\parbox[t]{\hsize}{\algocf@captiontext{#1}{#2}}% then caption is not centered
	\else%
	\global\@minipagefalse%
	\hbox to\hsize{\box\@tempboxa}% else caption is centered
	\fi%
	\addtolength{\hsize}{-\algomargin}%
}
\makeatother
\newcommand{\R}{\mathds{R}}
\newcommand{\E}{\mathds{E}}

\newcommand{\X}{\mathds{X}}

\renewcommand{\P}{\mathds{P}}

\newcommand{\Qcr}{\mathscr{Q}}
\newcommand{\Xcr}{\mathscr{X}}

\newcommand{\Ccr}{\mathscr{C}}

\newcommand{\D}{{\rm d}}
\newcommand{\kernel}{\mathcal K}
\newcommand{\Law}{\mathcal{L}}

\def\simind{\stackrel{\mbox{\scriptsize{ind}}}{\sim}}
\def\simiid{\stackrel{\mbox{\scriptsize{iid}}}{\sim}}
\newcommand{\indic}{\mathds{1}}

\allowdisplaybreaks[3]

\newtheorem{theorem}{Theorem}
\newtheorem{example}{Example}
\newtheorem{definition}{Definition}
\newtheorem{proposition}{Proposition}
\newtheorem{lemma}{Lemma}
\newtheorem{proof}{Proof}
\newtheorem{algo}{Algorithm}

\begin{document}

\title{\scshape\LARGE{Contaminated Gibbs-type priors}}

\author[1]{Camerlenghi F.\thanks{federico.camerlenghi@unimib.it}}
\author[1]{Corradin R.\thanks{riccardo.corradin@unimib.it}}
\author[1]{Ongaro A.\thanks{andrea.ongaro@unimib.it}}
\affil[1]{\normalsize{Department of Economics, Management and Statistics,University of Milano-Bicocca}}
\date{\today}

\maketitle 

\begin{abstract}
	Gibbs-type priors are widely used as key components in several Bayesian nonparametric models. By virtue of their flexibility and mathematical tractability, they turn out to be predominant priors in species sampling problems, clustering and mixture modelling.  We introduce a new family of processes which extend the Gibbs-type one, by including a contaminant component in the model to account for the presence of anomalies (outliers) or an excess of observations with frequency one. We first investigate the induced random partition, the associated  predictive distribution and we characterize the asymptotic behaviour of the number of clusters. All the results we obtain are in closed form and easily interpretable, as a noteworthy example we focus on the contaminated version of the Pitman-Yor process. Finally we pinpoint the advantage of our construction in different applied problems: we show how the contaminant component helps to perform outlier detection for an astronomical clustering problem and to improve predictive inference in a species-related dataset, exhibiting  a high number of species with frequency one. 
	
	\vspace{12pt}
	\noindent\keywords{Bayesian nonparametrics; Gibbs-type priors; mixture models; species sampling models; random partitions; outliers.}
\end{abstract}

\section{Introduction}

The great success of the  Dirichlet process within the Bayesian nonparametric framework has paved the way for the definition and investigation of a large variety of random probability measures. Indeed, since  the seminal paper by \cite{Fer73}, several discrete nonparametric priors have been proposed to accommodate for exchangeable observations, among these we mention: the Pitman-Yor process or two parameter Poisson-Dirichlet process \citep{Per92, Pit96}; species sampling processes \citep{Pit96}; priors based on normalization of completely random measures \citep{regz_03,Lij10}. Gibbs-type priors are another important class of Bayesian nonparametric models early introduced by \citep{Gne05} and recently investigated in \citep{Deb15}. The Gibbs-type family has the advantage to balance modelling flexibility and mathematical tractability. These  processes  have been successfully used in several frameworks, just to mention a few examples: to face prediction within species sampling framework \citep[e.g.][]{Lij07}, to define mixture models \citep[e.g.][]{Ish01,Lij07b}, for survival analysis \citep[e.g.][]{Jar10}, and for applications in linguistic and information retrieval \citep[e.g.][]{Teh06, Teh10}. \cite{Roy20} have recently discussed a class of feature allocation models  parametrized by Gibbs-type random probability measures.

%Motivated by several applications, we introduce a new Bayesian nonparametric model where a Gibbs-type prior is contaminated with an exogenous diffuse probability measure. Such a measure, called \textit{contaminant measure}, acts as a nuisance term on the sampling mechanism induced by the almost surely discrete Gibbs-type process. More precisely the
%random  probability measure we consider is a convex linear combination of a Gibbs-type prior $\tilde{q}$ and a diffuse probability $P_0$, i.e. we deal with the following random probability measure $\tilde{p}= \beta \tilde{q}+ (1-\beta) P_0$, where $\beta \in (0,1)$ is a weight.  

We introduce a new family of Bayesian nonparametric models where a Gibbs-type prior is contaminated with an exogenous diffuse probability measure, called \textit{contaminant measure}. More precisely, we define a new random probability measure as  a convex linear combination of a Gibbs-type prior $\tilde{q}$ and a diffuse probability $P_0$, i.e. we deal with $\tilde{p}= \beta \tilde{q}+ (1-\beta) P_0$, where $\beta \in [0,1]$ is a weight which tunes the impact of the contaminant measure. We refer to $\tilde{p}$ as a contaminated Gibbs-type prior (see Definition \ref{def:cGP}) and its distribution is then used  as a nonparametric prior in a Bayesian context. 
We will show that the advantage of this process in the Bayesian nonparametric setting is twofold: i)  $\tilde{p}$ is a tractable prior outside the Gibbs-type family which allows to enrich the predictive structure of exchangeable models, through the inclusion of the additional sampling information on the number of observations with frequency one out of the observed sample; ii) the contaminant measure $P_0$ accounts naturally for the presence of anomalies in the data (observations which are under some respects singular), thus resulting particularly suited for several applied problems. 
With regard to point ii), contaminated Gibbs-type prior can be exploited for modelling discrete data directly, when one needs to inflate the observations with frequency one. As an example, in Section \ref{sec:discrete_app}, we consider species detection data from the Global Biodiversity Information Facility project \citep{GBI21} with a high number of species detected only once. Within this framework we show the advantage of our model with respect to the traditional Pitman-Yor process to model the excess of ones. 
Besides, the proposed construction turns out to be useful in other contexts, such as in language modelling, disclosure risk assessment and operational taxonomic unit data analysis. See Section \ref{sec:discussion} for a thoughtful discussion on these applications. Furthermore, the new construction can be exploited as a mixing measure in mixture models to account for the presence of outliers in a dataset. In particular we are motivated by an astronomical dataset \citep{Iba11} composed by $n = 139$ stars, and we aim to understand which stars belong to a globular cluster and which stars are contaminants, i.e., outliers. Outlier detection is a crucial problem in Statistics and similar convex constructions are available also in classical setting, see, e.g., \cite{clustering_19} for an account. In the Bayesian framework, some contributions are available and rely on the use of traditional Dirichlet process. \cite{Qui_03,Qui_06} focus on product partition models, and they develop a decision-theoretic approach  that allows selecting a partition
with the purpose of outlier detection in regression problems.
\cite{Sho_11} identify an outlier detection criterion based on the Bayes factor, where they compare a
partition containing outliers against a partition with fewer or no outliers. As a remarkable addition with respect to the current literature, our prior process contemplates a specific component in the model, i.e. the contaminant measure $P_0$, to account for the presence of outliers, which thus follow a different generating process with respect to observations. This modeling strategy results in smoothed density estimates with respect to the ones obtained without the presence of the contaminant measure in the model. 
Motivated by all these applications, we introduce and deeply investigate the random partition structure induced by contaminated Gibbs-type priors. All the stated results are available in a closed form, they are simple and with a natural interpretation. The induced prediction rule can be easily explained in terms of a new Chinese restaurant with a social and a non-social room. As a concrete example, throughout the paper we focus on the contaminated version of the Pitman-Yor process, which exhibits more tractable expressions for all the quantity of interest and to face predictive inference. In particular, a simple extended P\'olya urn representation of the updating mechanism implied by this process can also be obtained.

%{\color{red} \st{Motivated by all these applications, in the first part of the paper we introduce and deeply investigate the random partition structure induced by contaminated Gibbs-type priors, in addition we derive  predictive distributions and asymptotic results for the number of clusters with a certain frequency. All the stated results  are available in closed-form, they are simple and with a natural interpretation.  More precisely, we will show that the random partition depends on a latent variable which represent the number of contaminated observations, and inference on this can be easily addressed. 
%As a concrete example, we focus on the contaminated version of the Pitman-Yor process, which exhibits more tractable expressions for all the quantities of interest; moreover all the algorithms have been implemented for such a noteworthy instance.}}

To the best of our knowledge, the Bayesian nonparametric literature has never focused on Gibbs type priors contaminated with the inclusion of a diffuse measure to model anomalies or the excess of singular observations. Experiments with simulated and real data will show their advantage to properly model the presence of such data with respect to standard Gibbs-type priors, which result in severe bias to estimate model parameters and in poor predictive performances.

\section{Contaminated Gibbs-type priors and their properties} \label{sec:model}

Let  $\{X_i\}_{i \geq 1}$  be a sequence of observations %defined on  a common probability space $(\Omega, \Acr, \P)$ and 
taking values in a Polish space $\X$, equipped with its Borel $\sigma$-field $\Xcr$. 
In the Bayesian nonparametric setting $X_1, X_2, \ldots$ are typically supposed to be exchangeable \citep{defin_37}, which is tantamount to saying that there exists a random probability measure 
$\tilde{p}$ such that $X_i | \tilde{p} \simiid \tilde{p}$, where the distribution of $\tilde{p}$ works as a prior in the Bayesian nonparametric framework.
The distribution of $\tilde{p}$, indicated by $\Qcr$, is called the de Finetti measure of the sequence $X_1, X_2, \ldots$, and several prior specifications 
$\Qcr$ are available in the Bayesian nonparametric literature. 
Among these we mention the remarkable class of species sampling models \citep{Pit96}. We recall that  an exchangeable sequence of observations $\{ X_i \}_{i \geq 1}$  is called a \textit{species sampling sequence}  if and only if it is governed by a distribution of the following type
\begin{equation}\label{eq:tilde_p_gen}
\tilde{p} = \sum_{j \geq 1} p_j \delta_{Z_j} +\Big( 1-\sum_{j \geq 1} p_j\Big) P_0,
\end{equation}
for a sequence of random weights $\{ p_j  \}_{j \geq 1}$ with $p_j \geq 0$ and $\sum_{j \geq 1}p_j \leq 1$ almost surely, and a sequence of random atoms $\{ Z_j \}_{j \geq 1}$ i.i.d. from $P_0$ independent of  $\{ p_j  \}_{j \geq 1}$, where  $P_0$ is assumed to be a  diffuse probability measure on  $(\X, \Xcr)$. A random distribution $\tilde{p} $ of the form in \eqref{eq:tilde_p_gen} is called a \textit{species sampling model}. Further, a species sampling model is termed \textit{proper} if and only if $\sum_{j \geq 1} p_j = 1$ almost surely, and most of the current Bayesian nonparametric literature focuses on the proper species sampling models. In this paper we discuss the case in which $\sum_{j \geq 1} p_j < 1$ with positive probability, and we show how non-proper models are particularly suited to take into account contaminated observations or more generally observations with frequency one.

Among the very general class of species sampling models we recover special subclasses of priors, which have been duly investigated in the literature, e.g., homogeneous normalized random measures with independent increments \citep{regz_03} and Gibbs-type priors \citep{Gne05,Deb15}. Here we focus on a contaminated version of Gibbs-type priors. For this reason, it is worth recalling that Gibbs-type random probability measures are typically characterized in terms of the exchangeable random partition \citep{Pit06} induced by the data.
More precisely, given a sample $X_{1:n}:= (X_1, \ldots , X_n)$ from a species sampling model governed by a random probability measure $\tilde{p}$, the $n$ observations are naturally partitioned into $K_n=k$ groups of distinct values, denoted here as 
$X_1^*, \ldots , X_k^*$, with corresponding frequencies $(N_{n,1}, \ldots , N_{n, K_n})= (n_1, \ldots , n_k)$. The exchangeable partition probability function (EPPF) corresponds to the probability of observing a specific partition of the data into clusters of distinct values, and it can be formalized as
\begin{equation}
\label{eq:EPPF_def}
\Pi_k^{(n) } (n_1, \ldots , n_k)  := \int_{\X^k } \E  \prod_{j=1}^k  \tilde{p}^{n_j} (\D  x_j^*).
\end{equation}
Gibbs-type priors are proper species sampling models  $\tilde{p}$ characterized by means of their sequence of EPPFs $\{ \Pi_k^{(n)} : \; n \geq 1, \, 1 \leq k \leq n  \}$, which can be expressed in the following form
\begin{equation}\label{eq:EPPF_def_gibbs}
\Pi_k^{(n) } (n_1, \ldots , n_k) = V_{n,k} \prod_{i=1}^k (1-\sigma)_{n_i-1},
\end{equation}
for all $n \geq 1$, $k \leq n$ and positive integers $n_1, \ldots , n_k$ with $\sum_{i=1}^k n_i =n$,  where $(a)_b= \Gamma (a+b)/\Gamma (a)$ in \eqref{eq:EPPF_def_gibbs}, for $a, b >0$, denotes the Pochhammer symbol. The discount parameter $\sigma <1$ and the non-negative weights $\{ V_{n,k} : \; n \geq 1, \, 1 \leq k \leq n \}$ must satisfy the recurrence  relation $V_{n,k} = (n -\sigma k) V_{n+1,k} +V_{n+1, k+1}$ for all $k =1, \ldots , n$, $n \geq 1$, with the proviso $V_{1,1}=1$ and $V_{0,0}=1$. The sequence of weights $V_{n,k}$s can be specified to recover prior processes commonly used in literature, such as the Dirichlet process \citep{Fer73}, the Pitman-Yor process \citep{Pit97}, the normalized inverse Gaussian process \citep{Lij05} and the normalized generalized gamma process \citep[see e.g.][and references therein]{Lij07b}. 
Building upon Gibbs-type priors, we now introduce a new family of prior processes which account for the possibility of contaminated observations.
\begin{definition} \label{def:cGP}
	Let $\tilde{q}$ be a Gibbs-type prior, specified by the sequence of weights $\{ V_{n,k} : \; n \geq 1, \, 1 \leq k \leq n \}$ and $\sigma<1$. 
	A \textit{contaminated Gibbs-type prior} is a random probability measure  on $(\X, \Xcr)$ defined as
	\begin{equation}\label{eq:HGT_def}
	\tilde p = \beta \tilde q + (1 - \beta) P_0, \qquad \beta \in (0, 1),
	\end{equation}
	where $Q_0(\, \cdot \,) = \E[\tilde{q}(\, \cdot \,)]$ is the base measure of $\tilde{q}$, and $Q_0,P_0$ are diffuse probability measures. 
\end{definition}
The prior $\tilde{p}$ in \eqref{eq:HGT_def} is a convex linear combination of two components: an almost surely discrete component $\tilde{q}$ which generates the data, and a diffuse probability measure $P_0$ which accounts for contaminated observations. In the sequel we refer to $P_0$ as the \textit{contaminant measure}.
Sampling from $\tilde p$ can be interpreted as sampling from a population formed by two parts: the first one, representing  a $\beta$ fraction of the entire population, is composed by a countable number of species each appearing with positive probability. The second part ($1-\beta$ fraction) can be thought of as composed by a continuum of individuals each belonging to a different species. Therefore any time we sample from this second part a new species is obtained that cannot be re-observed. As stated above, for simplicity we shall call  \textit{contaminant} this second part and \textit{contaminated} the relative observations. However, the diffuse part can be used more generally to account for any population which displays unique elements (see Section \ref{sec:discussion}) and/or to model a high number of generic singletons in the observations. Finally, note that in Definition \ref{def:cGP}, the contaminant measure $P_0$ may be different from the base measure $Q_0$, thus $\tilde{p}$ in \eqref{eq:HGT_def} may not be a species sampling model. This additional flexibility is introduced because it was found useful in some applied contexts to distinguish the distribution of contaminated observations from the others.

We first derive the expectation and the covariance structure of a contaminated Gibbs-type prior in order to understand  how 
the contaminant measure affects the distribution of $\tilde{p}$.
\begin{proposition}\label{thm:prior_quant}
	Let $\tilde p$ be a contaminated Gibbs-type prior as in Definition \ref{def:cGP}. Let $A, B \in \Xcr$, then 
	\[
	\begin{split}
	\E[\tilde p(A)] &= \beta Q_0(A) + (1-\beta)P_0(A),\\
	\mathrm{Cov}(\tilde p(A), \tilde p(B)) &= \beta^2 (1-\sigma) \frac{V_{2,1}}{V_{1,1}} [Q_0(A \cap B) - Q_0(A)Q_0(B)].
	\end{split}
	\]
	%	where $Q_0(\, \cdot \,) = \E[\tilde{q}(\, \cdot \,)]$ is the base measure of $\tilde{q}$, and $Q_0,P_0$ are diffuse probability measures on $(\X, \Xcr)$. 
\end{proposition}
As consequence of Proposition \ref{thm:prior_quant}, 
one has $\mathrm{Var}(\tilde p(A)) = \beta^2 \mathrm{Var}(\tilde q (A))$, therefore the diffuse probability measure $P_0$ in \eqref{eq:HGT_def} has the effect to shrink $\tilde q (A)$ towards its expected value. See Section \ref{sup:proof_prior} for a proof of Proposition \ref{thm:prior_quant}.
%\begin{equation}
%	\begin{aligned}\label{eq:hgibbs}
%		X_i | \tilde{p}  & \simiid \tilde{p}\\
%		\tilde{p} &= \beta \tilde{q} +(1-\beta) P_0. 
%	\end{aligned}
%\end{equation}
%\textcolor{red}{NO, o si specifica bene questa parte oppure direi di togliere}
%\textcolor{blue}{Let $A, B$ be arbitrary measurable subsets of $\X$. The resulting prior expectation of distribution as in \eqref{eq:hgibbs} is a convex combination of $P_0$ and $Q_0$ evaluated on the set $A$, $\E [\tilde p(A)] = \beta Q_0(A) + (1-\beta)P_0(A)$. Besides the inclusion of the diffuse term $P_0$ in \eqref{eq:hgibbs} has a shrinking effect on the diffusion of $\tilde p$, with $Cov(\tilde p(A), \tilde p(B) ) = \beta^2 \kappa(Q_0(A \cap B)-Q_0(A)Q_0(B))$ and  $\mathrm{Var} (\tilde p(A)) = \beta^2\kappa Q_0(A) (1 - Q_0(A))$, where $\kappa$ is a parameter depending on the specific choice for the distribution of the discrete term $\tilde q$. }
% In the next proposition, we first derive the expectation and the covariance structure of a contaminated Gibbs-type prior, which help us to understand  how 
% the contaminant measure affects the distribution of $\tilde{p}$.
%  As consequence of Proposition \ref{thm:prior_quant}, 
% one has $\mathrm{Var}(\tilde p(A)) = \beta^2 \mathrm{Var}(\tilde q (A))$, therefore the diffuse probability measure $P_0$ in \eqref{eq:HGT_def} has the effect to shrink $\tilde q (A)$ towards its expected value. See Section \ref{sup:proof_prior} for a proof of Proposition \ref{thm:prior_quant}.
\begin{example}[Contaminated Pitman-Yor process]  \label{ex:cPY1}
	Among the class of Gibbs-type priors, the Pitman-Yor process represents a noteworthy example, widely used in numerous applications. The contaminated Pitman-Yor process  can be constructed by selecting $\tilde q$ in \eqref{eq:HGT_def} to be a Pitman-Yor process. In such a case we recall that 
	\begin{equation}\label{eq:w_PY}
	V_{n,k} = \frac{\prod_{i=1}^{k-1}(\vartheta + i \sigma)}{(\vartheta + 1)_{n-1}}
	\end{equation}
	with $\sigma \in [0,1)$ and $\vartheta > -\sigma$. Moreover, if  $\sigma = 0$ we recover the Dirichlet process. 
\end{example}

\section{Random partition, prediction and asymptotic properties} \label{sec:inference}
Having introduced all the modeling assumptions in Section \ref{sec:model}, we now study the partition structure induced by a sample of observations from 
the random probability measure in \eqref{eq:HGT_def}, we further derive a closed form expression for predictive distributions and asymptotic properties for the number of clusters. 
We first focus on the random partition induced by a sequence of exchangeable observations governed by a contaminated Gibbs-type prior, deriving  the EPPF. %{\color{red}\st{If $X_1, \dots, X_n$ is a sample of size $n$, then the observations can be partitioned into $k$ clusters of distinct values $X_1^*, \ldots , X_k^*$, with frequencies $n_1, \ldots , n_k$.}} 
% Without loss of generality we can assume $X_1, \dots, X_{k_1}$ elements observed only once in the sample. We further denote by $J_1, \dots, J_{k_1}$ a suitable set of latent Bernoulli random variables, where the generic $J_i$ denotes if an observation is generated from the diffuse component, with $J_i = 0$, or to the discrete component, $J_i = 1$, and $\bar k_1 =\# \{ i : \; J_i=0 \}$ denotes the number of unique values associated to the diffuse component, with $k_1= \# \{ i : \; n_i=1 \}$ denotes the number of unique values out of the sample. 
%%$\bar k_1 = \sum_{i=1}^{k_1} \indic_{\{0\}}(J_i)$ denotes the number of unique values associated to the diffuse component, with $k_1= \sum_{i=1}^n \indic_{\{1\}} (n_i)$ denotes the number of unique values out of the sample. 
%The distribution of the generic $i$-th element in the sample is then
%\begin{equation}
%	\begin{aligned}\label{eq:hgibbs_aug}
%		X_i | \tilde{p}, J_i  & \simiid J_i \tilde{q} + (1-J_i) P_0 \\
%		J_i  &\simiid {\rm Bern} (\beta),
%	\end{aligned}
%\end{equation}
%and the model described in \eqref{eq:hgibbs} can be recover by marginalizing the model in \eqref{eq:hgibbs_aug} with respect to $J_i$. We can characterize the distribution of the latent partition in a sequence of observations $X_1, \dots, X_n$ through the EPPF of a contaminated Gibbs-type prior. 
\begin{theorem}\label{thm:EPPF_Gibbs}
	Let $\tilde p$ be a contaminated Gibbs-type prior  as in \eqref{eq:HGT_def}, with  $P_0$ and  $Q_0$ two diffuse probability measures 
	on $(\X, \Xcr)$. Suppose that  $X_i | \tilde{p} \simiid \tilde{p}$, as $i \geq 1$, then the probability that $n$ observations $X_{1:n}$ are partitioned into $K_n=k$ clusters of distinct values 
	$X_1^*, \ldots , X_{k}^*$  with corresponding frequencies $(N_{n,1}, \ldots , N_{n, K_n})= (n_1, \ldots , n_k)$ equals
	\begin{equation}\label{eq:EPPF_convex_comb}
	\Pi_k^{(n)} (n_1, \ldots , n_k) =  \E_{\bar{M}_{m_1}} [V_{n-\bar{M}_{m_1}, k- \bar{M}_{m_1}}] \beta^{n-m_{1}} \prod_{i=1}^{k} (1-\sigma)_{n_i-1}
	\end{equation}
	where $\bar{M}_{m_1} \sim {\mathrm{Binom}} (m_{1}, 1 - \beta)$  and $m_{1}= \# \{ i : \; n_i=1 \}$ denotes the number of singletons (i.e. observations with frequency one) out of the sample of size $n$.
\end{theorem}
%The EPPF described in Theorem~\ref{thm:EPPF_Gibbs} incorporates all the possible scenarios of allocating $X_1, \dots, X_{k_1}$ into the discrete or the diffuse component, by taking the expectation with respect to $\bar M_1$. 

See Section \ref{prof:th1} of the Appendix for a proof of Theorem \ref{thm:EPPF_Gibbs}.
From the expected value in \eqref{eq:EPPF_convex_comb}, it is apparent that the use of the contaminant measure $P_0$ in \eqref{eq:HGT_def} acts on observations with frequency one and, as expected, they play a central role in the expression of the EPPF. In order to fix the terminology we call \textit{singletons} the observations with frequency one, while the \textit{structural singletons} are those values generated from the contaminant measure $P_0$, whose number equals the latent quantity $\bar{M}_{m_1}$. Note that the term \textit{structural} refers to the fact that these values cannot be observed twice and this statistic could be of potential interest in certain applied problems, as it will be discussed in Section \ref{sec:discussion}.

We now get a glimpse of the probabilistic implications of the random partition \eqref{eq:EPPF_convex_comb} induced by contaminated Gibbs-type priors 
as compared to the pure Gibbs-type priors. 
In order to do this, we 
denote by  $(n_1, \ldots , n_k)$  and $(n_1', \ldots , n_k')$ two distinct compositions having the same number of distinct  values $k$ and corresponding to two samples with the same size $n$; the probability ratio between the EPPFs corresponding to the two compositions will be denoted by
$\mathrm{R} (n_1, \ldots  n_k ;n_1', \ldots  n_k '; n,k): =\Pi_k^{(n)} (n_1, \ldots , n_k)/ \Pi_k^{(n)} (n_1', \ldots , n_k')$.
%\begin{equation}
%\label{eq:ratio_def}
%\mathrm{R} (n_1, \ldots  n_k ;n_1', \ldots  n_k '; n,k) :=  \frac{\Pi_k^{(n)} (n_1, \ldots , n_k)}{\Pi_k^{(n)} (n_1', \ldots , n_k')}
%\end{equation}
% where $\Pi_k^{(n)}$ denotes a general EPPF. We compare the probability ratio when $\Pi_k^{(n)}$ is a Gibbs-type EPPF \eqref{eq:EPPF_def_gibbs} and when 
% it equals the EPPF of a contaminated Gibbs-type prior \eqref{eq:EPPF_convex_comb}. For the sake of simplifying notation, we denote by $R_{G}$ the ratio \eqref{eq:ratio_def} in the Gibbs-type case, and by $R_{cG}$ the same ratio between the EPPFs of a contaminated Gibbs-type prior.
In Proposition \ref{prop:compare_EPPF} (Section \ref{sec:proof_compare_EPPF}) we compare the probability ratio $\mathrm{R}$ when $\Pi_k^{(n)}$ is a Gibbs-type EPPF \eqref{eq:EPPF_def_gibbs} and when it equals the EPPF of a contaminated Gibbs-type prior \eqref{eq:EPPF_convex_comb}.
Proposition \ref{prop:compare_EPPF} of the Appendix clarifies that if the two compositions have the same number of singletons, the ratio is the same for the contaminated and non-contaminated model. On the other side, if the number of singletons out of the composition $(n_1, \ldots , n_k)$ is bigger w.r.t. the number of singletons out of $(n_1', \ldots , n_k')$, the relative ratio increases in the contaminated model. See Section \ref{sec:proof_compare_EPPF} for details.
Thus, in relative terms and given the number $k$ of distinct values, the contaminated Gibbs model modifies the probabilities of compositions only when a different number of singletons is involved, favouring compositions with a higher number of these elements. 
%	\textcolor{blue}{
%the contaminated Gibbs model assigns the same structure of probabilities to %compositions possessing the same number of singletons, favouring compositions %with higher such number otherwise.  	 }

For computational convenience, we can equivalently describe the EPPF \eqref{eq:EPPF_convex_comb} introducing a set of suitable latent variables on an augmented probability space. Indeed, we can denote by $J_1, \ldots, J_{n}$ Bernoulli random variables, where the generic $J_i$ indicates if the $i$th observation is generated from the contaminant measure $P_0$ ($J_i=0$), or from the a.s. discrete component $\tilde{q}$ ($J_i = 1$). %; in addition we set $\bar{M}_{m_1,1} =\# \{ i : \; J_i=0 \}$, which denotes the number of unique values generated by contaminant measure $P_0$. 
%$\bar k_1 = \sum_{i=1}^{k_1} \indic_{\{0\}}(J_i)$ denotes the number of unique values associated to the diffuse component, with $k_1= \sum_{i=1}^n \indic_{\{1\}} (n_i)$ denotes the number of unique values out of the sample. 
Thus, the introduction of latent elements $J_1, \ldots , J_n$ leads us to deal with the following augmented model
\begin{equation}
\begin{aligned}\label{eq:hgibbs_aug}
X_i | \tilde{p}, J_i  & \simiid J_i \tilde{q} + (1-J_i) P_0 \\
J_i  &\simiid \mathrm{Bern} (\beta),
\end{aligned}
\end{equation}
from which we may recover the marginal model by integrating  \eqref{eq:hgibbs_aug} with respect to $J_i$. Furthermore, $J_i =1 $ in \eqref{eq:hgibbs_aug} if the corresponding observation $X_i$ has been recorded at least twice in the sample: indeed 
if an observation $X_i$ is generated from $P_0$, it does not appear again in the sample with probability $1$.
Thus, the non-degenerate $J_i$s are those values referring to singletons out of the sample $X_{1:n}$. Without loss of generality 
we can assume that the elements appearing once in the sample are the first $m_{1}$ observations $X_1, \dots, X_{m_{1}}$. Based upon this augmentation, the random variable $\bar{M}_{m_1}$ in \eqref{eq:EPPF_convex_comb} equals $\sum_{i=1}^{m_1} (1-J_i)$ which represents the number of structural singletons among the observations recorded only once and it could be of potential interest in many application areas, as discussed in Section \ref{sec:discussion}. 
We now  describe the predictive distribution of the next observation $X_{n+1}$, conditionally given $X_{1:n}$  and the latent variables $J_{1:m_1}= (J_1, \dots, J_{m_{1}})$.
\begin{proposition}\label{prp:predictive}
	Let $\tilde p$ be a contaminated Gibbs-type prior  as in \eqref{eq:HGT_def}, with  $P_0$ and  $Q_0$ two diffuse probability measures 
	on $(\X, \Xcr)$.  Assume that  $X_i | \tilde{p} \simiid \tilde{p}$, as $i \geq 1$, and consider a  sample  $X_{1:n}$ which displays $K_n=k$ distinct values, denoted as $X_1^*, \ldots , X_k^*$, with respective frequencies $(N_{n,1}, \ldots , N_{n, K_n})= (n_1, \ldots , n_k)$, and the first $m_{1}$ values $X_1^*, \ldots , X_{m_{1}}^*$ are singletons. Then
	\begin{equation}
	\label{eq:pred_cond}
	\begin{split}
	&\P (X_{n+1} \in \D x^* | X_{1:n}, J_{1:m_1} ) = (1-\beta) P_0(\D x^*) + \beta \frac{V_{n-\bar{M}_{m_1}+1, k-\bar{M}_{m_1}+1}}{V_{n-\bar{M}_{m_1}, k-\bar{M}_{m_1}}} Q_0(\D x^*)\\
	&\qquad+ \beta\frac{V_{n-\bar{M}_{m_1}+1, k-\bar{M}_{m_1}}}{V_{n- \bar{M}_{m_1}, k- \bar{M}_{m_1}}} \left( \sum_{i=1}^{m_1} J_i (1-\sigma)  \delta_{X_i^*} (\D x^*)+ \sum_{i=m_1 + 1}^{k} (n_i-\sigma)  \delta_{X_i^*} (\D x^*) \right),
	\end{split}
	\end{equation}
	where $\bar{M}_{m_1} = \sum_{i=1}^{m_1} (1 - J_i)$ represents the latent number of structural singletons.
\end{proposition}
From the sampling mechanism dictated by the predictive distribution \eqref{eq:pred_cond}, it is apparent that those values sampled from the contaminant measure $P_0$ cannot be observed twice; moreover, at each sampling step, the probability of sampling a contaminated observation equals $1-\beta$ and does not depend on $n$. 
Note also that the sample without the $\bar{M}_{m_1}$ structural singletons is characterized by the usual predictive mechanism of  Gibbs-type priors. Finally, it is worth mentioning that the prediction rule has a nice interpretation in terms of a modified Chinese restaurant metaphor. Consider a restaurant with two rooms: a social and a non-social room.
The first customer arrives and she chooses a table either in the social room with probability $\beta$ or in the non-social room with probability $(1-\beta)$, she also chooses a dish which is shared by all the customers that will join the same table.
When the $n$th customer arrives, she first selects either the social room with probability $\beta$ or the non-social room with probability $1-\beta$. In the former case she can either sits at new table or at an occupied table according to the traditional Chinese restaurant metaphor in the social room, while in the latter case she sits alone at a new table eating a new dish.
%In the former case she sits alone at a new table eating a new dish, in the latter case she can either sits at new table or at an occupied table according to the traditional Chinese restaurant metaphor in the social room.

If we further assume that $P_0 = Q_0$, which corresponds to a proper species sampling model, we can derive an explicit form of the predictive distribution integrating over $J_{1:m_1}$ as shown in the following result.
\begin{proposition} \label{prp:predictive_mar}
	Let $\tilde p$ be a contaminated Gibbs-type prior  as in \eqref{eq:HGT_def}, with  $P_0=Q_0$ a diffuse probability measure
	on $(\X, \Xcr)$. Suppose that  $X_i | \tilde{p} \simiid \tilde{p}$, as $i \geq 1$, and consider a  sample  $X_{1:n}$ which displays $K_n=k$ distinct values, denoted as $X_1^*, \ldots , X_k^*$, with respective frequencies $(N_{n,1}, \ldots , N_{n, K_n})= (n_1, \ldots , n_k)$, and the first $m_{1}$ values $X_1^*, \ldots , X_{m_{1}}^*$ are singletons. Then
	\begin{align}
	\label{eq:predictiveG}
	&\P (X_{n+1} \in \D x^* | X_{1:n} ) =  \left\{ (1-\beta) +\beta \frac{\E_{\bar M_{m_1}}[V_{n-\bar M_{m_1} + 1, k-\bar M_{m_1} + 1}]}{
		\E_{\bar M_{m_1}}[V_{n-\bar M_{m_1}, k-\bar M_{m_1}}]} \right\} P_0 (\D x^*)\notag \\
	&\qquad+ \frac{1}{m_1}\sum_{i=1}^{m_1} \beta (1-\sigma) \frac{\E_{\bar M_{m_1}}[(m_1 - \bar M_{m_1}) V_{n-\bar M_{m_1} + 1, k-\bar M_{m_1}}]}{
		\E_{\bar M_{m_1}}[V_{n-\bar M_{m_1}, k-\bar M_{m_1}}]} \delta_{X_i^*} (\D x^*)\\
	&\qquad\qquad+ \sum_{i=m_1+1}^{k} \beta (n_i-\sigma) \frac{\E_{\bar M_{m_1}}[ V_{n-\bar M_{m_1} + 1, k-\bar M_{m_1}}]}{
		\E_{\bar M_{m_1}}[V_{n-\bar M_{m_1}, k-\bar M_{m_1}}]} \delta_{X_i^*} (\D x^*)\notag
	\end{align}
	where $\bar M_{m_1}\sim \mathrm{Binom}(m_1, 1 - \beta)$.
\end{proposition}
We refer to Section \ref{sup:predictive} of the Appendix for a proof of Proposition \ref{prp:predictive_mar}.  
The predictive distribution \eqref{eq:predictiveG} clearly shows that the probability that $X_{n+1}$ does not belong to $\{X_1^*, \ldots , X_k^*\}$ depends on the initial sample through the sample size $n$, the number of distinct values $k$ and the number of singletons  $m_1$. This is a remarkable addition w.r.t. the Gibbs-type family, in which such a probability does not depend on $m_1$ \citep{Bat_17}. Moreover  the probability that $X_{n+1}$ equals a previously observed value $X_i^*$, with $i=1,\ldots,k$, not only depends on $n$, $n_i$ and $k$, as in the  Gibbs-type framework, but also on $m_1$. As a consequence contaminated models allows to enrich the predictive structure of an exchangeable model, though the inclusion of the additional sampling information on the number of singletons out of the observable sample. On the other side analytical tractability is still preserved.
In Section \ref{sup:prob_new_comparison} we study the re-sampling mechanism induced by contaminated Gibbs-type prior in comparison with standard Gibbs-type priors. More precisely we show that the contaminant measure mainly acts on singletons by decreasing their re-sampling probabilities w.r.t. observations with higher frequencies. On the other side, for observations with frequency larger than one, we  preserve   the same reinforcement as the discrete term of the model, and the parameter $\sigma$ exhibits the same behavior as in the Gibbs-type case.   We now specialize all the results   for the contaminated Pitman-Yor process of Example \ref{ex:cPY1}.
\begin{example}[contaminated Pitman-Yor (continued)]\label{ex:PYp2}
	Consider the contaminated Pitman-Yor process of Example \ref{ex:cPY1}. We may recover an explicit expression  for  its EPPF starting from \eqref{eq:EPPF_convex_comb} and  by observing that  the weights $V_{n,k}$s equal \eqref{eq:w_PY}. Thus, we obtain
	%\begin{proposition} \label{prp:EPPF_PY}
	% Assume that $Q_0, P_0$ are diffuse measures on $(\X, \Xcr)$ absolutely continuous with respect to a diffuse measure $\mu$ on the same space. Let $\tilde p$ be an hybrid Gibbs-type prior with discrete component $\tilde{q}$ distributed as a Pitman-Yor process. Then the distribution of the random partition dictated by the data is described by the following EPPF
	%	Let $\tilde p$ be a contaminated Pitman-Yor process, and $X_1, \dots, X_n$ a sample from $\tilde p$. Then the EPPF o the partition dictated by the data is equal to
	\begin{equation}
	\label{eq:EPPF_convex_comb_PY}
	\Pi_k^{(n)} (n_1, \ldots , n_k) =  \prod_{i=1}^k  (1-\sigma)_{n_i-1}   \sum_{\bar m_1=0}^{m_1} \binom{m_1}{\bar m_1} \beta^{n - \bar m_1} (1-\beta)^{\bar m_1}  \frac{\sigma^{k - \bar m_1} (\vartheta/\sigma)_{k - \bar m_1}}{(\vartheta)_{n - \bar m_1}}.
	\end{equation} 
	%\end{proposition}
	See Section \ref{sup:eppf_PY} of the Appendix for the derivation of \eqref{eq:EPPF_convex_comb_PY}. The expression of the EPPF \eqref{eq:EPPF_convex_comb_PY} plays a central role to carry out posterior inference in our applications,  indeed all the algorithms we have developed (see Section \ref{sup:disc} of the Appendix) are based on this expression. 
	We conclude the example specializing the predictive distribution \eqref{eq:pred_cond} for the contaminated Pitman-Yor model:
	\begin{equation}    
	\label{eq:predictive}
	\begin{split}
	&\P (X_{n+1} \in \D x | X_{1:n}, J_{1:m_1}) =  (1 - \beta)P_0(\D x) + \beta \frac{\vartheta + (k - \bar{M}_{m_1}) \sigma}{\vartheta + n - \bar{M}_{m_1}} Q_0(\D x)\\
	&\qquad\qquad + \sum_{i=1}^{m_1}J_i \beta\frac{1 - \sigma}{\vartheta + n- \bar{M}_{m_1}} \delta_{X_i^*}(\D x)+ \sum_{i=m_1 + 1}^{k}\beta\frac{n_i - \sigma}{\vartheta + n - \bar{M}_{m_1}} \delta_{X_i^*}(\D x).
	\end{split}
	\end{equation}
	In Section \ref{sec:lemma} we show that the probability of sampling a new value is monotone as a function of the number of distinct values $m_1$, which  results in a richer predictive structure w.r.t. the Pitman-Yor case, where $m_1$ does not appear in the probability of sampling a new value. For example, the dependence on $m_1$ is always increasing in the Dirichlet process case ($\theta=0$), whereas it is always decreasing in the stable process one ($\sigma=0$). Some numerical experiments are presented in Section \ref{sup:prob_new}. Finally  the predictive distribution \eqref{eq:predictive} can be described in terms of an urn model, with solid and strip balls, when the prior for $\beta$ is a beta with parameters $\alpha$ and $\vartheta$, which correspond to the initial weight of strip colored balls and of black solid balls, respectively. At the first sampling step, if a strip colored ball is drawn from the urn, then we return the ball in the urn with an additional strip colored ball of a new color. On the other side if we draw a black solid ball, then we return a black ball in the urn with an additional black ball of weight $\sigma$ and a solid ball of a new color with weight $1-\sigma$. At the generic $i$th step, one can sample three different kinds of balls: a strip ball of an arbitrary color, a black solid ball or a  colored solid ball, where once that we draw a colored solid ball, we replace that ball in the urn with an additional one of the same color. See Section \ref{sup:urn} for a detailed description of the updating mechanism.
\end{example}

We conclude this section with some considerations on distributional properties of the number of clusters with a given frequency in a sample of size $n$:
this helps us to better understand the advantage of contaminated Gibbs-type priors. To fix the notation, we consider a sample $X_{1:n}$ from a  contaminated Gibbs-type prior, and we denote by  $M_{n,r}$ the random number of elements observed $r$ times out of the sample. In the sequel, if $V$ is a statistic depending on the sample $X_{1:n}$, we write $V(\beta)$ to make explicit the dependence on the parameter $\beta$ of the contaminated prior \eqref{eq:HGT_def}.
The following proposition clarifies the effect of the contaminant component with respect to the Gibbs-type model in terms of stochastic dominance and asymptotic properties.
\begin{proposition}\label{prp:asymptotic}
	If $\beta_1 < \beta_2$, then $K_n (\beta_1)$ (resp. $M_{n,1} (\beta_1)$) stochastically dominates $K_n (\beta_2)$ (resp. $M_{n,1} (\beta_2)$).
	Moreover, as $n \to +\infty$, we have 
	\[
	\frac{K_n}{n}  \stackrel{a.s.}{\to} 1-\beta, \quad  \frac{M_{n,1}}{n}  \stackrel{a.s.}{\to} 1-\beta \quad \text{and} \quad
	\frac{M_{n,r}}{n^\sigma} \stackrel{a.s.}{\to} \frac{\sigma (1-\sigma)_{r-1}}{r!}  S_{\sigma} \beta^\sigma,
	\]
	where $S_\sigma$ denotes the $\sigma$-diversity random variable \citep{Pit06}.
\end{proposition}
The first part of Proposition \ref{prp:asymptotic} is a result of first order stochastic dominance and it clarifies the effect of the contaminant measure in the model \eqref{eq:HGT_def}. As $\beta$ decreases, the number of distinct values and the number of singletons out of $X_{1:n}$ increases.   
%As a consequence  the expected values of $K_n(\beta)$ and $M_{n,1} (\beta)$ are decreasing in $\beta$. 
By noticing that the case $\beta=1$ corresponds to a Gibbs-type prior, it is now apparent that our model has the advantage to increase (in mean) the number of distinct values and the number of singletons: the smaller beta, the higher $\E [K_n (\beta)]$ and $\E [M_{n,1} (\beta)]$.
The second part of Proposition \ref{prp:asymptotic} tells us that the number of distinct values $K_n$ and the number of unique values scale linearly with $n$: this is a remarkable difference with respect to Gibbs-type priors. Indeed, as $n \to +\infty$, for Gibbs-type priors both $K_n$ and $M_{n,1}$ grows as $n^\sigma$ \citep{Pit06}. Also, the asymptotic behavior of $M_{n,r}$ remains unchanged with respect to Gibbs-type priors, apart for the presence of the factor $\beta^\sigma$. This asymptotic behavior clarifies the role of the contaminant measure $P_0$ in \eqref{eq:HGT_def}, which produces an inflation of the number of singletons, and consequentially of the number of unique elements, but it is not acting on higher frequencies values.
We now specialize the results for the contaminated Pitman-Yor process.

\begin{example}[contaminated Pitman-Yor (continued)]  \label{ex:PY3}
	As for contaminated Pitman-Yor priors of Example \ref{ex:cPY1}, it is possible to evaluate the expected value of $M_{n,r}$ and $K_n$, in particular we have obtained:
	\begin{align*}
	\E [M_{n,1}] &=  n (1-\beta) +  n \beta  \E [(\beta B_1 +(1-\beta))^{n-1}] \\
	\E [M_{n,r} ] &= \frac{(1-\sigma)_{r-1}}{(\vartheta+1)_{r-1}} \binom{n}{r} \beta^r \E[(B_r \beta +1-\beta)^{n-r}], \quad \text{if } r \geq 2 \\
	\E [K_n]& = \frac{\vartheta}{\sigma} \E [ (B_1\beta +1-\beta)^n ] +\frac{n \beta}{\sigma} \E [B_1  (B_1\beta+1-\beta)^{n-1}]-\frac{\vartheta}{\sigma} +n (1-\beta)
	\end{align*}
	where $B_r $ is a Beta random variable with parameters $(\vartheta+\sigma, r-\sigma)$, as $r \geq 1$.
	See Section \ref{sec:app_cPY} of the Appendix for further details.  In Figure \ref{fig:prior_plots}, we compare the behavior of the expected values of the statistics $K_n$, $M_{n,1}$ and $M_{n,2}$  in the Pitman-Yor case with the same quantities for the contaminated model. It is apparent that for the latter model the two curves of $\E[K_n]$ and $\E[M_{n,1}]$ grow faster as function of $n$, with respect to the Pitman-Yor model. 
	We finally underline that, resorting to the results by \cite{Fav13}, one may face prediction for a large number of statistics arising in species sampling models. Indeed, in Section \ref{sec:app_cPY} of the Appendix, we evaluated the posterior expected value of the following meaningful statistics: i)  $K_m^{(n)}$, which denotes
	the number of distinct species out a future sample  $X_{n+1:n+m}= (X_{n+1}, \ldots , X_{n+m})$ not yet observed in the initial sample $X_{1:n}$; ii) $N_{m,r}^{(n)}$, which denotes the number of new and distinct observations recorded with frequency $r$ out of the additional sample $X_{n+1:n+m}$, hitherto unobserved in the initial sample of size $n$.
	All these posterior expected values display closed form expressions (see Section \ref{sec:app_cPY} for details), which depend  not only on $n$ and $k$, as for all the class of Gibbs-type priors \citep{Bat_17}, but also on the number of singletons $m_1$, thus improving the flexibility of the Pitman-Yor process. Building upon the results of \cite{Fav13}, one may derive formulas also for  contaminated Gibbs-type random partitions.
	\begin{figure}
		\centering
		\includegraphics[width=0.9\textwidth]{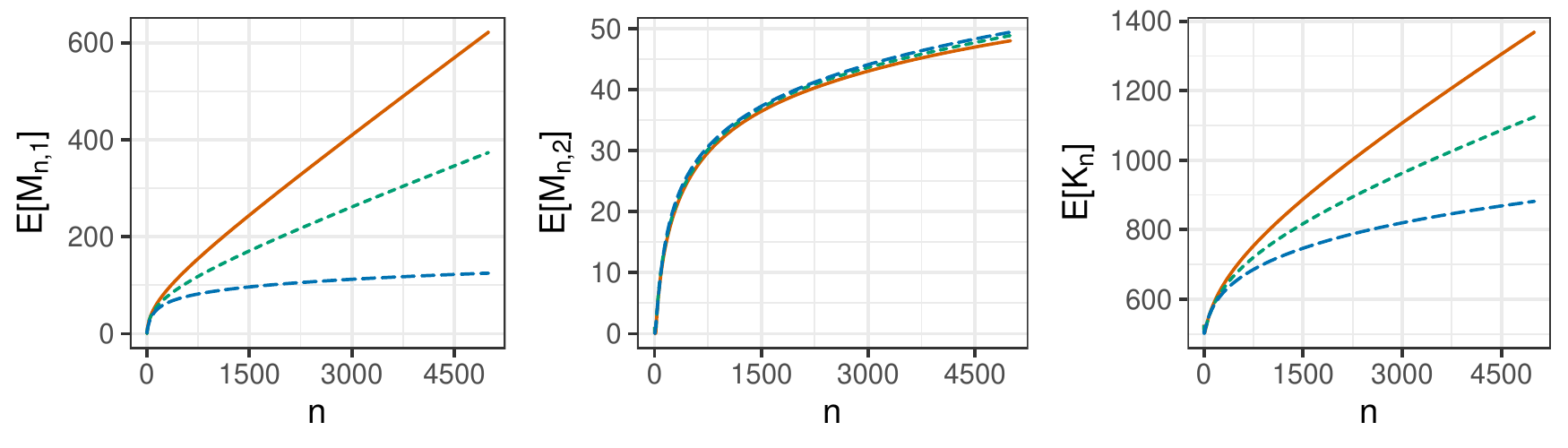}
		\caption{Curves of $\E[K_n]$, $\E[M_{n,1}]$  and $\E[M_{n,2}]$ as $n$ increases, for the contaminated Pitman-Yor model with $\beta=0.9$ (red), the contaminated Pitman-Yor model with $\beta=0.95$ (green), and the Pitman-Yor model (blue). The parameters are $\vartheta=50$ and $\sigma=0.2$.}\label{fig:prior_plots}
	\end{figure}
\end{example}

\section{Mixtures of contaminated Gibbs-type priors} \label{sec:mixture}

Contaminated Gibbs-type priors are not restricted only to species sampling models, but they can be used also as mixing measures to build contaminated mixture models.
Mixture models in Bayesian nonparametrics were early introduced by \citep{Lo84} for the Dirichlet process mixtures of univariate Gaussian distribution case, and later extended in several directions, by considering different kernel functions or mixing measures. %, e.g. to the multivariate Gaussian distribution \citep{Mul96}, Dirichlet process mixtures of Poisson distribution \citep{Krn08}, [ADD REFERENCES FOR DIFFERENT RPMs]. See also \citet{Fru19} for an extensive review on mixture models.  
%We now introduce a framework to face density estimation and model based %clustering in presence of outliers. 
%be a Polish space, endowed with its Borel $\sigma$-field $\Ycr$. 
%We assume the distribution of a generic random element $Y$ taking value in $\Y$, indexed by a latent parameter $\xi$ with support $(\Xi, \Tcr)$, where $\Xi$ is a Polish space and $\Tcr$ its Borel $\sigma$-field. We further assume the distribution of $\xi$ governed by a random probability measure $\tilde p$. Let $\kernel(y; \xi)$ be a kernel function, with $\kernel : \Y \times \Xi  \to \mathbb R^+$, $\kernel(y; \xi)$ density function for all $\xi \in \Xi$ and measurable in both its arguments. We set $\tilde p$ distributed as a contaminated Gibbs-type prior of the form $\tilde p = \beta \tilde q + (1-\beta)P_0$, with $\xi\mid \tilde p \simiid \tilde p$, $\tilde q$ denotes a Gibbs-type prior with base measure $Q_0$, and $Q_0, P_0$ diffuse probability measures taking values on $(\Xi, \Tcr)$. Then a mixture model with contaminated Gibbs-type mixing measure, and kernel function $\kernel(y; \xi)$, can be obtained by integrating the kernel function over the latent parameter $\xi$, with respect to the mixing measure $\tilde p$, as
%\begin{equation}\label{eq:model}
%	\tilde f(y) = \int_\Xi \kernel(y;\xi) \tilde p(\D \xi),
%\end{equation}
%where $\tilde f$ denotes a random distribution taking values on $\mathscr F_\Y$, space of all probability distributions with support $\Y$.
The standard general framework can be described as follows. It is assumed that  observations are $\R^d$-valued random elements generated from a random density $\tilde f(y) = \int_\Xi \kernel(y;\xi)  p(\D \xi)$, where $\kernel(y; \xi):  \R^d \times \Xi  \to \R^+$ is a kernel and $\Xi$ is  a general Polish space.
% We assume that $\tilde p$ is a contaminated Gibbs-type prior of the form $\tilde p = \beta \tilde q + (1-\beta)P_0$, according to Definition \ref{def:cGP}. 
Furthermore, the mixing measure $ p$ is usually assumed to belong to a specific class of discrete random probability measures.
If one denotes by $\xi_1, \ldots , \xi_n$ the latent variables corresponding to a sample of size $n$ from  $p$, the standard mixture model may be expressed in the following hierarchical form
\begin{equation}
\label{eq:mix_model}
\begin{split}
Y_{i} | \xi_i  \simind   \kernel(\, \cdot \, ;\xi_i) , \quad \xi_i | {p}  \simiid {p} \\
%\tilde p &\sim \Qcr 
\end{split}
\end{equation}
for any $i =1, \ldots , n$.
We remark that the model \eqref{eq:mix_model} describes a general formulation of a mixture models. Indeed a realization of $\tilde{f}$ can be a continuous distribution, a discrete distribution \citep[e.g.][]{Krn08}, or a distribution defined on more abstract spaces, depending on the kernel function $\kernel$. Nowadays it is an established opinion in the applied statistics framework that mixture models are flexible tools for density estimation and model-based clustering analysis \citep{Fru19}. 

Here we propose to extend such framework by choosing as mixing measure $p$ the contaminated Gibbs-type prior $\tilde p$. 
Thanks to the definition of $\tilde p$, which is a linear convex combination of two elements, we can decompose the mixture in two terms, a first term corresponding to the discrete part of $\tilde p$ and a second term which corresponds to the diffuse component,
\begin{equation}\label{eq:model2}
\tilde f(y) = \beta \sum_{j=1}^{\infty} p_j \kernel(y;Z_j) + (1 - \beta) \int_\Xi \kernel(y;\xi) P_0(\D \xi)
\end{equation}
where the last equality holds in force of the almost sure discreteness of $\tilde q= \sum_{j \geq 1} p_j \delta_{Z_j}$. The first term on the r.h.s. of equation \eqref{eq:model2} describes the standard random mixture components of the model, while the second term corresponds to a different probabilistic mechanism contaminating the mixture. A noteworthy application of this model is to the cases where outliers are possibly present in the data. Indeed, according to well developed classical theory \citep{Fru19} they can be interpreted as generated by a different random process with respect to the other observations.

If one considers $Q_0 \neq P_0$, she can specify the contaminant measure $P_0$ depending on a specific scenario of interest: if our prior opinion is translated into contaminant observations on a particular subset of $\R^d$, we can force $P_0$ to shrink its mass on such subset. On the other hand, if we aim to model possible contaminant observations spreading over the entire support, we can specify $P_0$ over-disperse with respect to $Q_0$.

%Notice that the distribution of a sequence of exchangeable data $Y_1, \dots, Y_n$ from $\tilde f(y)$ can be equivalently expressed in terms of hierarchical representation, as 
%\[
%\begin{split}
%Y_i \mid \theta_i &\simind K(Y_i; \theta_i), \quad i = 1, \dots, n,\\
%\theta_i \mid \tilde p &\simiid \tilde p,\\
%\tilde p &\sim \Qcr,
%\end{split}
%\]
%where it is apparent the role of 
From the hierarchical formulation \eqref{eq:mix_model}, it is apparent that the random probability measure $\tilde p$ governs the distribution of the latent parameters $\xi_i$s. Thus, posterior inference for mixture models may  be performed by exploiting the results described in the previous sections deriving a marginal sampling strategy in the spirit of the seminal works of \citet{Esc88} and \citet{Esc95}.
%using posterior quantities of the distribution of the latent parameters. We can for example exploit the predictive distribution of the generic $i$-th element given the rest of the sample to define a marginal sampling strategy, in the spirit of the seminal works of \citet{Esc88} and \citet{Esc95}. 
See Section \ref{sup:mix} of the Appendix for a description of a possible sampling strategy to perform posterior inference with mixtures of contaminated Pitman-Yor processes. 

\section{Illustrations}  \label{sec:exp}
\subsection{Simulation studies}  \label{sec:simulations}

In the Appendix we carried out some simulation studies to illustrate the use of the contaminated Pitman-Yor process. We first tested the proposed model in discrete scenarios by simulating  observations from the contaminated Pitman-Yor process of Example \ref{ex:cPY1}, with different values of the parameters $\beta$, $\vartheta$ and $\sigma$. See Section \ref{sup:sim_disc}.
We faced posterior inference on the main parameters of the model, on $\beta$ and the number of structural singletons $\bar{M}_{m_1}$ by exploiting Algorithm \ref{algo:discrete}. Our strategy provides good results in terms of parameters' estimation, also in comparison with the Pitman-Yor process, which, for example, overestimates the discount parameter of the model in presence of contamination of the data. With the proposed model, we also  obtain reliable estimates of the weight $\beta$ and the number of structural singletons.

We then moved to a simulation study within the framework of mixture models in Section \ref{sec:cont_simu}. We tested the model on different simulated scenarios, where observations are generated from a mixture of Gaussian distributions, with the inclusion of some outliers in the sample. See Section \ref{sup:mix} of the Appendix for further details on the data generating process. We compare posterior inference faced with three different Gaussian mixture models, where the random mixing measure is specified as: i) a contaminated Pitman-Yor with $P_0 = Q_0$; ii) a contaminated Pitman-Yor with $P_0 \neq Q_0$, forcing an over-dispersion of the contaminant measure; iii) a Pitman-Yor.
Posterior inference is carried out on the basis of a marginal sampling scheme with the goal of outlier detection (see Algorithm \ref{algo:mixture} in Section \ref{sup:mix} of the Appendix). Note that  we consider an observation to be an outlier iff it is clustered as a singleton in the posterior point estimate of the latent partition dictated by the data. From Table \ref{tab:sing} in the Appendix, one may realize that the Pitman-Yor mixture model is not appropriate to perform outlier detection, indeed only few clusters with frequency one are detected. On the other hand, the two specifications of the contaminated Pitman-Yor mixture models produce reasonable estimates of the number of contaminants in the data, and the model with  $P_0 \neq Q_0$ displays appreciable superior performances.

\subsection{The North America Ranidae dataset}  \label{sec:discrete_app}

We consider a set of species detection data from the Global Biodiversity Information Facility project \citep{GBI21}. The project is an extensive database consisting in record of species found across the world, where for each individual is reported the taxonomy, location and possibly other relevant information. Our sample consists of $n = 131\,204$ observations belonging to $k = 619$ distinct species of the Ranidae family observed in North America, and identified by their scientific name. 
Among the $k=619$ species, $m_1 = 296$ species were observed only once in the sample, creating a possible inflation of the number of elements with frequency one. Such inflation might be caused by miss reported scientific name of the observed animals. We aim to investigate the benefit of including a contaminant measure in the prior model specification by comparing posterior inference when we use a contaminated and a standard Pitman-Yor process.  We choose non-informative prior specifications  for the parameters, namely $\vartheta \sim \mathrm{Gamma}(2, 0.02)$ and $\sigma , \beta \sim \mathrm{Unif}(0,1)$. We carried out posterior inference by  exploiting Algorithm \ref{algo:discrete} described in Section \ref{sup:disc} of the Appendix, and similarly for the standard model. Refer to Section \ref{sup:disc_app} for diagnostic summaries and algorithmic details.

\begin{figure}
	\centering
	\includegraphics[width = 0.9 \textwidth]{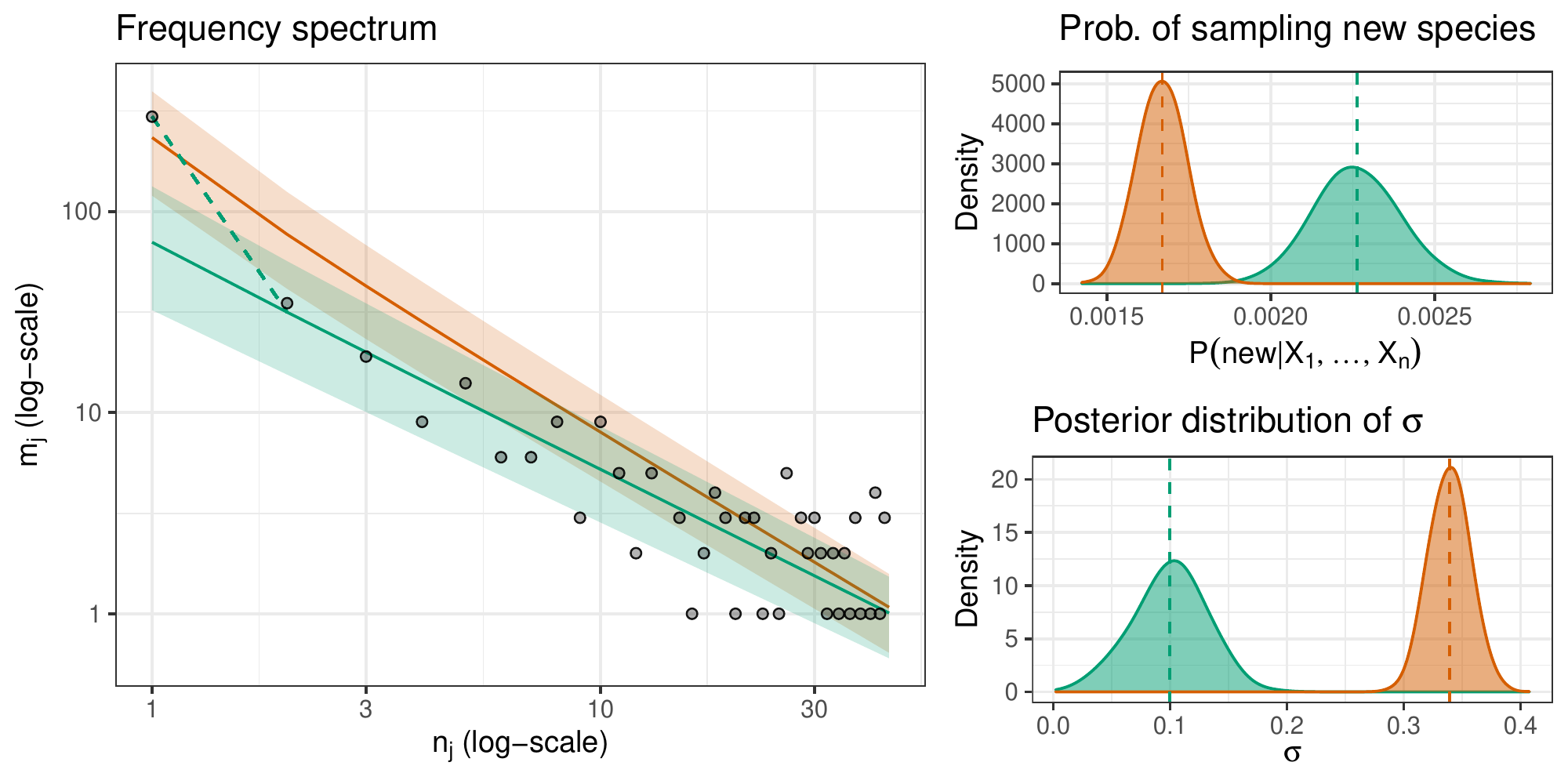}
	\caption{Posterior summaries for contaminated Pitman-Yor model (green) and Pitman-Yor model (orange). Left panel: frequency spectrum of the first non-empty frequencies, with the posterior expectation of the discrete part of the two models, shaded bands represent $90\%$ posterior credible intervals; the dashed line corresponds to the inflation of the diffuse component;. Right-top panel: posterior probability of sampling a new species. Right-bottom: posterior distribution of $\sigma$.}	
	\label{fig:disc_app}
\end{figure}

Figure \ref{fig:disc_app} clarifies how the presence of a large number of species observed only once leverages the estimation of the parameters in the Pitman-Yor model, while the use of a  contamination component helps to obtain a much more suitable modeling of the data. Indeed, in the latter case, some of the observations with frequency $1$ are assigned to the diffuse component. As consequence of the excessive number of singletons, the estimated posterior distribution of the frequency spectrum is remarkably different on small values of the support, as emphasized in the left panel of Figure \ref{fig:disc_app}. Furthermore, both the probability of sampling a new species and the posterior distribution of $\sigma$ in  the Pitman-Yor case are translated with respect to the contaminated model. Additional posteriors summaries are reported in the Appendix: the posterior distributions of $\vartheta, \beta$ and $\bar{M}_{m_1}$.
%: Figure \ref{fig:post_plot3} shows the comparison between the contaminated Pitman-Yor model and the Pitman-Yor model for the posterior distribution of $\vartheta$. Figure \ref{fig:post_plot4} shows the posterior distribution of $\beta$ and $\bar{M}_{m_1}$ for the contaminated Pitman-Yor model. 

We finally consider the task  of predicting the number of new species and the number of new species observed with a given frequency in a follow-up sample,  given an initial training sample. We have retained the $80\%$ of the $n$ data for purposes of training, and the remaining $m$ data points are used as a test set. We focused on estimation of: i) $K_{m}^{(n-m)}$, the distinct number of new species in a follow-up sample hitherto unobserved in the initial training sample of size $n-m$; ii) $N_{m,1}^{(n-m)}$, the number of new species observed with frequency one in an additional sample of size $m$, hitherto unobserved in the training dataset. The posterior expectations of $K_{m}^{(n-m)}$  and $N_{m,1}^{(n-m)}$ are evaluated using the corresponding closed-form expressions, reported in Equations \eqref{eq:Kmn_posterior} and \eqref{eq:Nmr1} respectively, for the contaminated Pitman-Yor model.
The predicted values are compared with the true ones, obtained by extrapolating to the remaining
$m$ data. We repeated the experiment  $1\,000$ times in order to asses variability.
Figure \ref{fig:add_sample} shows the cross-validated distributions of $K_{m}^{(n-m)}$ and $N_{m,1}^{(n-m)}$ when we exploit the contaminated Pitman-Yor model in comparison with the predicted values obtained by using the Pitman-Yor process. The average true value is represented with a dashed black line.

\begin{figure}
	\centering
	\includegraphics[width = 0.45 \textwidth]{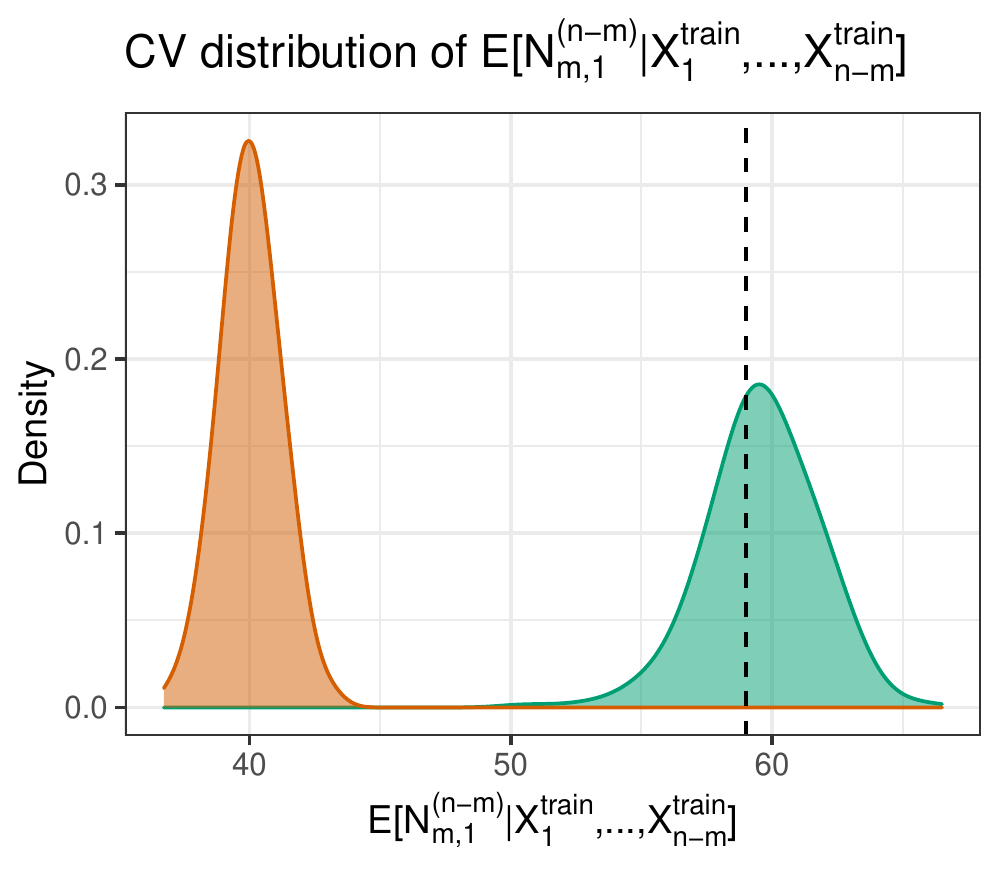}\centering
	\includegraphics[width = 0.45 \textwidth]{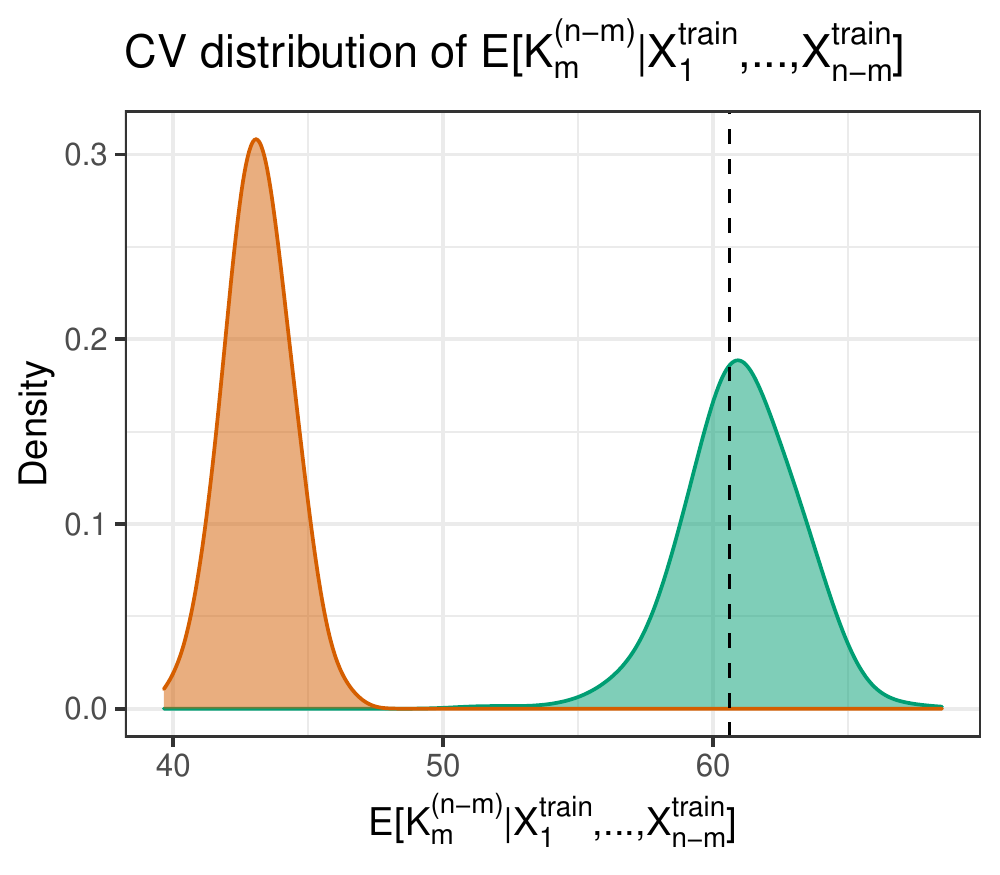}
	\caption{Cross-validated distributions of $K_{m}^{(n-m)}$  and $N_{m,1}^{(n-m)}$, for the contaminated Pitman-Yor model (green) and the Pitman-Yor model (orange). The black dashed lines correspond to the true values. 
		%denote the average of the observed number of new singletons and the observed number of new distinct species in the additional sample, respectively.
	}\label{fig:add_sample}
\end{figure}

From Figure \ref{fig:add_sample}, it is apparent how the contaminant measure in the model specification can be crucial also for its predictive properties. Indeed the cross-validated distributions of $N_{m,1}^{(n-m)}$ and $K_{m}^{(n-m)}$ for the contaminated model, conditionally on an observed sample, shrink to the corresponding average of the observed values (black dashed line), while the distributions for the model without a contaminant term provides a systematic error in prediction. Such behavior is also confirmed by the mean squared error ($\mathrm{MSE}$) of the predictions, which is bigger for the Pitman-Yor (PY) model with respect to the contaminated Pitman-Yor (CPY), indeed: $\mathrm{MSE}(N_{m,1}^{(n-m)}) = 74.12$ for CPY and $\mathrm{MSE}(N_{m,1}^{(n-m)}) = 423.36$ for PY; $\mathrm{MSE}(K_{m}^{(n-m)}) = 75.48$ under CPY and $\mathrm{MSE}(K_{m}^{(n-m)}) = 370.67$ under PY. See Section \ref{sec:inf_new_samp} of the Appendix for further details on the cross-validation study. 
Finally we stress that this example highlights the strong lack of robustness of the pure Pitman-Yor model. Indeed, drastically erroneous inferential conclusions are caused by relatively few singletons ($m_1=296 $) compared to the total number of observations ($n = 131\,204 $).

\subsection{Analysis of the NGC 2419 data}\label{sec:cont_app}
We consider a set of data composed by $n = 139$ stars, possibly belonging to the globular cluster NG 2419 and sharing the same galactic center. The data were early introduced and studied by \citet{Iba11}. 
For each observation we have measurements of $d = 4$ different variables: the two-dimensional projection on the plane of the position of the star $(D_1, D_2)$, the line of sight velocity $V$ on a logarithmic scale, and the metallicity of the star $[Fe/H]$ on a logarithmic scale, which is a measure of the abundance of iron relative to hydrogen. We denote by $Y_i = (D_{1,i}, D_{2,i}, V_i, [Fe/H]_i)$ the
$i$th observed vector. 

A crucial problem for the astronomical community is to identify which stars belong to the globular cluster, and which star are contaminants (or outliers), to properly study the dynamic of a group of stars. To this aim, we consider a contaminated Pitman-Yor mixture model, specified with a multivariate Gaussian kernel function $\kernel(\, \cdot \, ; (\mu, \Sigma))$, with expectation $\mu$ and covariance matrix $\Sigma$. We further assume a base measure conjugate to the kernel function, i.e., $Q_0 \sim \mathrm{NIW}(\mu_0, \kappa_0, \nu_0, S_0)$ is a Normal-Inverse-Wishart distribution. We consider the same distributional law also for the diffuse component:  $P_0 \sim \mathrm{NIW}(\mu_1, \kappa_1, \nu_1, S_1)$. 
%Before estimating the models, all the observed variables were marginally standardized, preserving the dependency across different dimensions, by applying for each dimension a robust transformation, i.e. subtracting the median and dividing by the ratio of the interquartile range and the interquartile range of a standard Gaussian distribution. 
We specify the base measure of the discrete component $Q_0$ by setting $\mu_0$ equal to the sample mean of the data, $\kappa_0 = 1$, $\nu_0 = d + 3 = 7$, and $S_0$ equals to the diagonal of the sample variance of the data. We further specify the parameters of the contaminant measure as follows: $\mu_1$ equals the sample mean of the data, $\kappa_1 = 0.1$, $\nu_1 = d + 3 = 7$, and $S_1$ matches the sample variance of the data, in order to force an over-disperse contaminant measure with respect to the base measure. We complete the model specification assuming vague priors for the parameters of the mixing measure, $\vartheta \sim \mathrm{Gamma}(2,0.02)$ and  $\sigma , \beta\sim \mathrm{Unif}(0,1)$. Posterior inference is carried out by exploiting the sampling scheme described in Section \ref{sup:mix} of the Appendix, see also Section \ref{sup:mixt_app} for diagnostics. We exploit the decisional  approach based on the variation of information loss function \citep{Wad18,Ras18}
to provide an optimal posterior point estimate of the latent partition induced by the data.

The results are summarized in Table \ref{tab:vsIBA} in comparison with the previous clusters identified by \cite{Iba11}. Within the $16$ stars identified as contaminants, $4$ belongs to the globular cluster identified by \citet{Iba11}, $5$ stars to the likely globular cluster group, and $7$ to the contaminants. Most of the stars of the main estimated cluster, denoted by $A$ in Table \ref{tab:vsIBA}, belong to the globular cluster of \citet{Iba11}. We have also recovered two additional clusters: a cluster of stars mainly belonging to the globular cluster in \citet{Iba11}, and a cluster with two contaminant stars. Our findings are coherent with the previous literature, but providing a more conservative detection of the contaminant stars. Figure \ref{fig:application_plot} of the Appendix shows the optimal partition with the mean of the estimated posterior random density. From the contour lines in Figure \ref{fig:application_plot}, we can see how the inclusion of a diffuse component is producing a smoothed estimate of the expectation of the posterior random density, while for the standard Pitman-Yor mixture model the expectation of the posterior random density shows some peaks in correspondence of the contaminants, as shown in Figure \ref{fig:application_plot_PY}. We further compare our findings with the latent partition obtained using a Pitman-Yor mixture model, as described in Section \ref{sec:mixt_comparison} of the Appendix: the optimal latent partition recovered with the contaminated Pitman-Yor prior is characterized by a larger main globular cluster and a higher number of singletons. 

\begin{table}
	\centering
	\begin{tabular}{llc*{5}{c}}
		\\
		&       && \multicolumn{5}{c}{CPY partition} \\ 
		&       && \textit{Singletons}    & \textit{A}   & \textit{B}    & \textit{C}  \vspace{5pt}\\
		&&\multicolumn{1}{c}{\textit{total}}&\textit{16}&\textit{115}&\textit{4}&\textit{4}\\ 
		\multirow{3}{*}{\citet{Iba11}} & \textit{globular cluster}   &\multicolumn{1}{c}{\textit{118}}& 4  & 109 & 3 &2   \\
		& \textit{likely globular cluster} &\multicolumn{1}{c}{\textit{12}}& 5 & 6 & 1 &0 \\
		& \textit{contaminants} & \multicolumn{1}{c}{\textit{9}}& 7 &0 &0&2    \\
	\end{tabular}
	\caption{Comparison between the partition described in \citet{Iba11} and optimal partition estimated using a contaminated Pitman-Yor mixture model.}
	\label{tab:vsIBA}
\end{table}

\section{Discussion}\label{sec:discussion}

%We introduced a new family of prior processes where a Gibbs-type prior is contaminated by an exogenous diffuse probability measure, inflating the number of distinct elements observed only once. 
We introduced a new family of priors outside the Gibbs-type one which are still tractable from an analytical viewpoint.  According to the characterization by \cite{Bat_17}, the predictive probability weights of Gibbs-type priors cannot depend on the number of observations recorder with frequency one  $M_{n,1}$ in  the initial sample.
With the inclusion of a contaminant component, we have enriched the predictive structure  of Gibbs-type priors by including the additional sampling information on $M_{n,1}$.
Moreover we discussed the usage of contaminated Gibbs-type priors in two situations: i) for discrete data in presence of and excess of ones; ii) in mixture models to account for outliers. Nevertheless, the use of contaminated Gibbs-type priors is not restricted to the scenarios presented in this manuscript, but they can be relevant in other applications, where the presence of elements with frequency one is a key inferential interest. 

Firstly, contaminated Gibbs-type priors could be of potential interest in the analysis of disclosure risk for microdata. Microdata files typically contains two types of categorical information about individuals: identifying and sensitive information. Before releasing a dataset, statistical agencies estimate different measures of disclosure risk, which are typically based on the number of sample records which have a unique combination of the categorical variables and that are not shared with any element of the entire population. See, e.g., \cite{Bet_90,Ski_02,Ski_94} for possible definitions and estimators of disclosure indexes. In the disclosure risk framework, the random variable $\bar{M}_{m_1}$ appearing in our model represents a measure of disclosure, i.e., the number of records that contain a unique element both in the sample and in the whole population.

Language modeling constitutes another application area when one is interested to estimate the number of \textit{hapax legomena} in a corpus of documents. An hapax legomena is indeed a word that occurs only once in the entire production of an author. 
These unique words are particularly important since they have been recognized as peculiar
usage of words by specific authors, and they represent an interesting problem to study
from a statistical perspective. See e.g. \citet{Baa01} for further details on word frequency distributions. In such a framework, one may use a contaminated Gibbs-type prior to estimate the number of hapax legomena on the basis of the latent quantity $\bar{M}_{m_1}$. 

%{\color{red} \st{Another possible field where contaminated priors can be an useful model is represented by the analysis of operational taxonomic unit data. Such data consist in measurements of group of organisms showing based on similarity of traits. Remarkable examples of this type of observations are microbial species data among others. The observed data are known to shown often sparsity, and they are usually composed by few dominant species and a large amount of rare species} \citep[see e.g.][]{Jia20, Sog15}. \st{A model which allows an inflation of the species with frequencies equal to one can be a suitable prior choice also in these scenarios. 
%
%Extending the current framework, and considering a proper species sampling priors as discrete component of the model instead of a Gibbs-type prior, is an interesting direction for future studies. Further, the specification of a contaminated prior can be investigated in dependent scenarios by including a contaminant term also in the case of partially exchangeable observations, with the contaminant measure possibly shared across different groups of data. }}
Finally a contaminated Gibbs-type model may be exploited to test the presence of contaminant observations in a set of data by selecting a spike and slab prior \citep{Mit88} for the parameter $\beta$ in \eqref{eq:HGT_def}. More precisely one may specify  a prior for $\beta$ which assigns positive mass 
to the point $\beta =1$. A prior specification of this type may  be exploited to perform posterior inference on the presence of a contaminant term in the model, looking at the posterior probability of $\{ \beta =1 \}$. 
%Such procedure is a natural strategy to infer possible contamination of the %data.
%, by simply testing if the posterior model differs significantly from an usual Gibbs-type prior. 
Work on these points is ongoing.

\section*{Acknowledgement}
The authors gratefully acknowledge the financial support from the Italian Ministry of Education, University and Research (MIUR), ``Dipartimenti di Eccellenza" grant 2018-2022, and the DEMS Data Science Lab for supporting this work through computational resources. Federico Camerlenghi received funding from the European Research Council (ERC) under the European Union's Horizon 2020 research and innovation programme under grant agreement No 817257. 

%\printbibliography
%\bibliographystyle{}
\bibliography{bibliography}

\appendix

\section{Proofs}

\subsection{Proof of Proposition \ref{thm:prior_quant}}\label{sup:proof_prior}

The first assertion of Proposition \ref{thm:prior_quant} is immediate: indeed $\E[\tilde p(A)] = \beta \E[\tilde q(A)] + (1 - \beta) P_0(A)$, and $\E[\tilde q(A)]= Q_0(Q)$ since $\tilde{q}$ is a Gibbs-type prior. To prove the second assertion, note that
\[
\mathrm{cov}(\tilde p(A), \tilde p(B)) = \mathrm{cov}(\beta\tilde q(A) + (1-\beta)P_0(A), \beta\tilde q(B) + (1-\beta) P_0(B)) = \beta^2 \mathrm{cov}(\tilde q(A), \tilde q(B)),
\]
with 
\begin{equation}\label{eq:cov_sup}
\mathrm{cov}(\tilde q(A), \tilde q(B)) = \E[\tilde q(A)\tilde q(B)] - \E[\tilde q(A)]\E[\tilde q(B)].
\end{equation}
We focus on the evaluation of the first term in \eqref{eq:cov_sup} 
\[
\E[\tilde q(A)\tilde q(B)] = \P (X_1 \in A, X_2 \in B)
\]
where $ X_1, X_2 \mid \tilde q \simiid \tilde q$, thus, we get:
\[
\begin{split}
\P (X_1 \in A, X_2 \in B) &= \E [\P (X_2 \in B | X_1)\cdot \indic_A (X_1)]\\
&= \E\left[ \left((1 - \sigma)\frac{V_{2,1}}{V_{1,1}} \indic_B(X_1) + \frac{V_{2,2}}{V_{1,1}} Q_0(B)\right)\indic_A(X_1) \right] ,
\end{split} 
\]
where we used the predictive distribution of a Gibbs-type prior.
%and $X_1$ is marginally distributed according to $Q_0$. 
Then, the previous expression equals 
\[
\E[\tilde q(A)\tilde q(B)]=(1 - \sigma)\frac{V_{2,1}}{V_{1,1}} Q_0(A\cap B) + \frac{V_{2,2}}{V_{1,1}} Q_0(A)Q_0(B).
\]
We substitute the previous term in \eqref{eq:cov_sup}, and we obtain
\[
\begin{split}
\mathrm{cov}(\tilde q(A), \tilde q(B)) &= (1 - \sigma)\frac{V_{2,1}}{V_{1,1}} Q_0(A\cap B) + \frac{V_{2,2} - V_{1,1}}{V_{1,1}} Q_0(A)Q_0(B) \\
&= (1 - \sigma)\frac{V_{2,1}}{V_{1,1}} [Q_0(A\cap B) - Q_0(A)Q_0(B)]
\end{split}
\]
where the last equality holds true in force of the recurrence relation of the weights $\{ V_{n,k} : \; n \geq 1, \, 1 \leq k \leq n \}$, which implies $V_{2,2} - V_{1,1} = - (1-\sigma) V_{2,1}$.

\subsection{Proof of Theorem \ref{thm:EPPF_Gibbs}}\label{prof:th1}

Let us denote by $\mu := 1/2P_0+1/2 Q_0$ the diffuse probability measure on $(\X, \Xcr)$, with respect to which both $P_0$ and $Q_0$ are absolutely continuous measures.
We would like to evaluate the EPPF using the definition \eqref{eq:EPPF_def} and focusing on the contaminated Gibbs-type prior case
\begin{align}
\Pi_k^{(n)} (n_1, \ldots , n_k) & = \int_{\X^k} \E \left[  \prod_{i=1}^k (\beta \tilde{q} (\D x_i^*) + (1-\beta) P_0(\D x_i^*))^{n_i} \right] \nonumber\\
& = \int_{\X^k} \E \left[ \prod_{i=1}^k \sum_{j_i=0}^{n_i} \binom{n_i}{j_i} \beta^{j_i} \tilde{q}^{j_i} (\D x_i^*)   (1-\beta)^{n_i-j_i} P_0^{n_i-j_i}(\D x_i^*)\right]. \label{eq:EPPF_1}
\end{align}
We now concentrate on the evaluation of the expected value in \eqref{eq:EPPF_1}, that can be computed as follows
\begin{equation} \label{eq:EXP_EPPF}
\begin{split}
&\E \left[ \prod_{i=1}^k \sum_{j_i=0}^{n_i} \binom{n_i}{j_i} \beta^{j_i} \tilde{q}^{j_i} (\D x_i^*)   (1-\beta)^{n_i-j_i} P_0^{n_i-j_i}(\D x_i^*)\right] \\
& \qquad =  \sum_{j_1=0}^{n_1} \ldots  \sum_{j_k=0}^{n_k} \prod_{i=1}^k  \binom{n_i}{j_i} \E \left[ \prod_{i=1}^k  \beta^{j_i} \tilde{q}^{j_i} (\D x_i^*) (1-\beta)^{n_i-j_i} P_0^{n_i-j_i}(\D x_i^*) \right].
\end{split}
\end{equation}
We now recall that $m_1= \# \{ i : \; n_i=1 \}$  is the number of observations recorded only once out of the sample of size $n$. Without loss of generality we can assume that these observations are the first $m_1$ values $X_1^*, \ldots , X_{m_1}^*$, which is tantamount to saying that $n_i =1$ for any $i=1, \ldots , m_1$ and $n_i >1$ if $i \geq m_1+1$. Neglecting the superior order terms in \eqref{eq:EXP_EPPF}, we obtain
\begin{equation*}
\begin{split}
&\E \left[ \prod_{i=1}^k \sum_{j_i=0}^{n_i} \binom{n_i}{j_i} \beta^{j_i} \tilde{q}^{j_i} (\D x_i^*)   (1-\beta)^{n_i-j_i} P_0^{n_i-j_i}(\D x_i^*)\right] \\
& \qquad =  \sum_{(j_1, \ldots , j_{m_1}) \in \{ 0,1 \}^{m_1}} \E \left[ \prod_{i=1}^{m_1}  \beta^{j_i} \tilde{q}^{j_i} (\D x_i^*) (1-\beta)^{1-j_i} P_0^{1-j_i}(\D x_i^*)  \prod_{i=m_1+1}^{k} \beta^{n_i} \tilde{q}^{n_i} (\D x_i^*)\right]\\
& \qquad\qquad\qquad\qquad
+ o \left( \prod_{i=1}^k \mu (\D x_i^*) \right).
\end{split}
\end{equation*}
We now define $\bar m_{1} =\# \{ i : \; j_i=0 \}$
%	\[
%	\bar{k}_1:= \sum_{i=1}^{k_1} j_i 
%	\]
which represents the number of observations generated from $P_0$, having frequency $1$. Thus,  by noticing that $\sum_{i=m_{1}+1}^k n_i = n-m_1$, we get
\begin{equation*}
\begin{split}
&\E \left[ \prod_{i=1}^k \sum_{j_i=0}^{n_i} \binom{n_i}{j_i} \beta^{j_i} \tilde{q}^{j_i} (\D x_i^*)   (1-\beta)^{n_i-j_i} P_0^{n_i-j_i}(\D x_i^*)\right] \\
& \qquad =  \sum_{(j_1, \ldots , j_{m_1}) \in \{ 0,1 \}^{m_1}}  \beta^{n-\bar m_{1}} (1-\beta)^{\bar m_1}  
\E \left[ \prod_{i=1}^{m_1} \tilde{q}^{j_i} (\D x_i^*) \prod_{i=m_1+1}^{k} \tilde{q}^{n_i} (\D x_i^*) \right] \prod_{i=1}^{m_1} P_0^{1-j_i}(\D x_i^*)
\\
& \qquad\qquad\qquad\qquad
+ o \Big( \prod_{i=1}^k \mu (\D x_i^*) \Big).
\end{split}
\end{equation*}
By integrating the previous expression over $\X^k$, we get the expression of the EPPF
\begin{equation} \label{eq:EPPF_2}
\begin{split}
\Pi_k^{(n)} (n_1, \ldots , n_k)  & =   \sum_{(j_1, \ldots , j_{m_1}) \in \{ 0,1 \}^{m_1}}  \beta^{n-\bar m_1} (1-\beta)^{\bar m_1}  \\
& \qquad\qquad\times
\int_{\X^{k-\bar m_1}} \E \left[ \prod_{i=1}^{m_1} \tilde{q}^{j_i} (\D x_i^*) \prod_{i=m_1+1}^{k} \tilde{q}^{n_i} (\D x_i^*) \right].
\end{split}
\end{equation}
We recognize that the integral in \eqref{eq:EPPF_2} is the EPPF of a Gibbs-type prior, therefore
\begin{equation} \label{eq:EPPF_3}
\begin{split}
&\Pi_k^{(n)} (n_1, \ldots , n_k)  \\
& \quad=   \sum_{(j_1, \ldots , j_{m_1}) \in \{ 0,1 \}^{m_1}}  \beta^{n-\bar m_1} (1-\beta)^{\bar m_1}  
V_{n-\bar m_1, k-\bar m_1}  \prod_{i=m_1+1}^k (1-\sigma)_{n_i-1},
\end{split}
\end{equation}
where we have now to solve the summation over the $j_i$'s. We observe that each summand in \eqref{eq:EPPF_3} depends on the vector 
$(j_1, \ldots , j_{m_1})$ only through $\bar m_1$. Moreover note  that, fixed the value of $\bar m_1$, there are $\binom{m_1}{\bar m_1}$ possible ways to choose $(j_1, \ldots , j_{m_1})$ so that $\sum_{i=1}^{m_1} (1 - j_i) =\bar m_1$, thus
\begin{equation*} 
\begin{split}
\Pi_k^{(n)} (n_1, \ldots , n_k)  =  \sum_{\bar m_1=0}^{m_1}  \binom{m_1}{\bar m_1} \beta^{n-\bar m_1} (1-\beta)^{\bar m_1}  
V_{n-\bar m_1, k-\bar m_1} \prod_{i=m_1+1}^k (1-\sigma)_{n_i-1},
\end{split}
\end{equation*}
and the results easily follows having realized that the sum in the previous expression is an expected value w.r.t. the distribution of the Binomial random variable $\bar M_{m_1}\sim \mathrm{Binom}(m_1, 1 - \beta)$.

\subsection{Relative ratio between EPPFs: comments} \label{sec:proof_compare_EPPF}

Recall the probability ratio defined in the paper:
\begin{equation}
\label{eq:ratio_def_app}
\mathrm{R} (n_1, \ldots  n_k ;n_1', \ldots  n_k '; n,k) :=  \frac{\Pi_k^{(n)} (n_1, \ldots , n_k)}{\Pi_k^{(n)} (n_1', \ldots , n_k')}
\end{equation}
In this section, we want to compare the probability ratio when $\Pi_k^{(n)}$ is a Gibbs-type EPPF \eqref{eq:EPPF_def_gibbs}, denoted by $\mathrm{R_{G}}$, and when it equals the EPPF of a contaminated Gibbs-type prior \eqref{eq:EPPF_convex_comb}, denoted by $\mathrm{R_{cG}}$. 
\begin{proposition} \label{prop:compare_EPPF}
	Consider two compositions  $(n_1, \ldots , n_k)$  and $(n_1', \ldots , n_k')$ deriving from two samples having the same size $n$ and the same number of distinct values $k$, and denote by $m_1 = \# \{ i : \; n_i=1 \}$ (resp. $m_1'= \# \{ i : \; n_i'=1 \}$) the number of singletons in the first (resp. second) composition. If $m_1 =m_1' $ then %of unique element for the first (resp. second) configuration. If $m_1 =m_1' $ then
	\[
	\mathrm{R_{cG}} (n_1, \ldots  n_k ;n_1', \ldots  n_k '; n,k) = \mathrm{R_{G}} (n_1, \ldots  n_k ;n_1', \ldots  n_k '; n,k),
	\]
	whereas if $m_1 > m_1'$
	\[
	\mathrm{R_{cG}} (n_1, \ldots  n_k ;n_1', \ldots  n_k '; n,k) \geq \mathrm{R_{G}} (n_1, \ldots  n_k ;n_1', \ldots  n_k '; n,k).
	\]
\end{proposition}
\begin{proof}
	If $m_1=m_1'$, we have that
	\begin{align*}
	&\mathrm{R_{cG}} (n_1, \ldots  n_k ;n_1', \ldots  n_k '; n,k) = \frac{\E_{\bar{M}_{m_1}} [V_{n-\bar{M}_{m_1}, k- \bar{M}_{m_1}}] \beta^{n-m_{1}} \prod_{i=1}^{k} (1-\sigma)_{n_i-1}}{\E_{\bar{M}_{m_1}} [V_{n-\bar{M}_{m_1}, k- \bar{M}_{m_1}}] \beta^{n-m_{1}} \prod_{i=1}^{k} (1-\sigma)_{n_i'-1}}\\
	& \qquad \qquad= \frac{V_{n,k} \prod_{i=1}^{k} (1-\sigma)_{n_i-1}}{V_{n,k}\prod_{i=1}^{k} (1-\sigma)_{n_i'-1}}=  \mathrm{R_{G}} (n_1, \ldots  n_k ;n_1', \ldots  n_k '; n,k)
	\end{align*}
	as stated.\\
	For the second part of the proposition consider $m_1 > m_1'$. Firstly, observe that $V_{n-\bar{m},k-\bar{m}}$ is an increasing function of $\bar{m}$, as a simple consequence of the recurrence relation which entails that
	\[
	\frac{V_{n-\bar{m}+1,k-\bar{m}+1}}{V_{n-\bar{m},k-\bar{m}}} \leq 1 .
	\]
	Moreover $\bar{M}_{m_1}$ stochastically dominates $\bar{M}_{m_1'}$, since they are two binomial random variables with $m_1> m_1'$, thus one has
	\begin{equation}
	\label{eq:dom_EV}
	\E_{\bar{M}_{m_1}} [V_{n-\bar{M}_{m_1}, k- \bar{M}_{m_1}}] \geq \E_{\bar{M}_{m_1'}} [V_{n-\bar{M}_{m_1'}, k- \bar{M}_{m_1'}}].
	\end{equation}
	One can exploit \eqref{eq:dom_EV} to conclude the proof:
	\begin{align*}
	\mathrm{R_{cG}} (n_1, \ldots  n_k ;n_1', \ldots  n_k '; n,k)& =  \frac{\E_{\bar{M}_{m_1}} [V_{n-\bar{M}_{m_1}, k- \bar{M}_{m_1}}] \beta^{n-m_{1}} \prod_{i=1}^{k} (1-\sigma)_{n_i-1}}{\E_{\bar{M}_{m_1'}} [V_{n-\bar{M}_{m_1'}, k- \bar{M}_{m_1'}}] \beta^{n-m_{1}'} \prod_{i=1}^{k} (1-\sigma)_{n_i'-1}}\\
	& \geq  \beta^{m_1'-m_{1}}\frac{ \prod_{i=1}^{k} (1-\sigma)_{n_i-1}}{ \prod_{i=1}^{k} (1-\sigma)_{n_i'-1}}
	\geq \frac{ \prod_{i=1}^{k} (1-\sigma)_{n_i-1}}{ \prod_{i=1}^{k} (1-\sigma)_{n_i'-1}}\\
	& =   \mathrm{R_{G}} (n_1, \ldots  n_k ;n_1', \ldots  n_k '; n,k)
	\end{align*}
	where we used the fact that $\beta^{m_1'-m_{1}} \geq 1$.
\end{proof}
The previous proposition tells us that if the two configurations have the same number of singletons, the ratio is the same for the contaminated and non-contaminated model. On the other side, the ratio increases in the contaminated model if we increase the number of singletons in the composition at the numerator term. See Section \ref{sec:proof_compare_EPPF} for a proof of Proposition \ref{prop:compare_EPPF}.

\subsection{Example \ref{ex:PYp2}: details}  \label{sec:lemma}

In this section we provide all the details to show that the probability of sampling a new value is monotone as a function of the number of distinct values 
$m_1$ for the contaminated Pitman-Yor process.\\
First of all we prove the following.
\begin{lemma} \label{lem:Mdom}
	Under the contaminated Pitman-Yor prior, the posterior distribution of $\bar{M}_{m_1}$, counting the number of values assigned to the diffuse component, satisfies the monotone likelihood ratio property, i.e., 
	if $m_1 < m_1'$, then 
	\[
	\frac{\P (\bar{M}_{m_1'} = \bar{m}| X_1, \ldots , X_n)}{\P (\bar{M}_{m_1} = \bar{m}| X_1, \ldots , X_n)}
	\]
	increases as $\bar{m}$ increases.
\end{lemma}
\begin{proof}
	Recall that the posterior distribution of $\bar{M}_{m_1}$ equals
	\[
	\P (\bar{M}_{m_1} = \bar{m}| X_1, \ldots , X_n) \propto \binom{m_1}{\bar{m}} \beta^{n-\bar{m}} (1-\beta)^{\bar{m}} V_{n-\bar{m}, k-\bar{m}}
	\] 
	where the normalizing factor does not depend on $\bar{m}$, but only on $m_1$ and $\beta$.
	To prove the monotone likelihood ratio property, we consider $m_1< m_1 '$ and we focus on the ratio
	\begin{align*}
	\frac{\P (\bar{M}_{m_1'} = \bar{m}| X_1, \ldots , X_n)}{\P (\bar{M}_{m_1} = \bar{m}| X_1, \ldots , X_n)} & 
	\propto \binom{m_1'}{\bar{m}}\cdot \binom{m_1}{\bar{m}}^{-1}  = \frac{m_1'!}{m_1!} \cdot \frac{(m_1-\bar{m})!}{(m_1'-\bar{m})!}\\
	& \propto \frac{(m_1-\bar{m})!}{(m_1'-\bar{m})!} = \frac{1}{(m_1'-\bar{m})\cdots (m_1-\bar{m}+1)}.
	\end{align*}
	The previous ratio increases as $\bar{m}$ increases, as long as $m_1'> m_1$. Thus, the result follows.
\end{proof}

Thanks to Lemma \ref{lem:Mdom} we are able to show the monotone property stated at the beginning of the section.
From the predictive distribution \eqref{eq:predictive}, we observe that, conditionally on $\bar{M}_{m_1}= \bar{m}_1$, the probability of sampling a new value at the $(n+1)$th stage equals 
\[
p_{new} (\bar{m}_1) :=(1 - \beta) + \beta \frac{\vartheta + (k - \bar{m}_1) \sigma}{\vartheta + n - \bar m_1}
\] 
and it is immediate to see that such a probability is monotone in $\bar{m}_1$, when $n$ and $k$ are fixed values. Moreover if $\sigma\to 0$, i.e. the Dirichlet case, $p_{new} (\bar{m}_1)$ is increasing, whereas if $\theta \to 0$ (stable process) the probability $p_{new} (\bar{m}_1)$ decreases in
$\bar{m}_1$. Note that the unconditional probability of sampling a new value may be recovered integrating $p_{new} (\bar{m}_1)$ with respect to the posterior distribution of $\bar{M}_{m_1}$. As a consequence of Lemma \ref{lem:Mdom} such a distribution satisfies the monotone likelihood ratio property in $\bar{m}_1$, thus  the unconditional probability of sampling a new vale is monotone as a function of the number of distinct values $m_1$. This results in a richer predictive structure w.r.t. the Pitman-Yor case, where $m_1$ does not appear in the probability of sampling a new value.

\subsection{Proof of Proposition \ref{prp:predictive}}

The conditional predictive distribution \eqref{eq:pred_cond} follows from the EPPF augmented with the introduction of the latent variables $J_1, \ldots , J_n$,  by considering all the possible scenarios: $X_{n+1}$ is new from $P_0$, $X_{n+1}$ is new from $Q_0$, $X_{n+1}$ coincides with a previously observed value appearing with frequency one in the initial sample, and $X_{n+1}$ coincides with a previously observed value having frequency $n_i \geq 2$ in the initial sample.

\subsection{Proof of Proposition \ref{prp:predictive_mar}}\label{sup:predictive}

Suppose now that $P_0 = Q_0$. We can integrate the distribution of the latent variables $(J_1, \dots, J_n)$ in \eqref{eq:pred_cond}. We start from the law of the random partition dictated by the data, augmented with the introduction of the latent variables $(J_1, \ldots, J_{m_1})$. This can be recovered from the proof of Theorem \ref{thm:EPPF_Gibbs}, and assuming $P_0 = Q_0$ it amounts to be
\begin{equation} \label{eq:joint_aug}
\begin{split}
&\Law (X_1, \ldots , X_n, J_1, \ldots, J_{m_1}) \\
& \qquad = \beta^{n-\bar m_{1}} (1-\beta)^{\bar m_{1}} V_{n-\bar m_{1}, k-\bar m_{1}} \prod_{i=m_1+1}^{k} (1-\sigma)_{n_i-1}
\prod_{i=1}^k P_0(\D X_i^*)
\end{split}
\end{equation}
where  $(j_1, \ldots, j_{m_1})$ are the observed values of $(J_1, \ldots , J_{m_1})$ and $\bar m_{1} =\# \{ i : \; j_i=0 \}$ is the observed number of uniques generated from the diffuse component. From Equation \eqref{eq:joint_aug} we may apply the Bayes theorem to recover the conditional distribution of $(J_1, \ldots, J_{m_1} | X_{1}, \ldots , X_n)$ and this is proportional to
\begin{equation}
\label{eq:cond_j|X_no_norm}
\P (J_1=j_1, \ldots, J_{m_1}= j_{m_1}| X_1, \ldots , X_n)  \propto \beta^{n-\bar m_{1}}(1-\beta)^{\bar m_{1}} V_{n-\bar m_{1}, k-\bar m_{1}}
\end{equation}
where the normalizing constant can be determined summing over all the vectors $(J_1,\ldots , J_{m_1})$ belonging to the set $\{ 0,1\}^{m_1}$. Indeed 
one can easily verify that
\begin{equation}
\label{eq:cond_j|X}
\P (J_1=j_1, \ldots, J_{m_1}= j_{m_1}| X_1, \ldots , X_n) = \frac{ \beta^{m_1 - \bar m_{1}}(1-\beta)^{\bar m_{1}} V_{n-\bar m_{1}, k-\bar m_{1}}}{\E_{\bar M_{m_1}} [V_{n-\bar M_{m_1}, k-\bar M_{m_1}}]}
\end{equation}
with $\bar M_{m_1} \sim {\mathrm{Binom}} ( m_1, 1 - \beta)$. We can now integrate the expression in \eqref{eq:pred_cond} w.r.t. the law \eqref{eq:cond_j|X} to get the result. It is straightforward to integrate the first and the last term on the r.h.s. of \eqref{eq:pred_cond}, the second one is more subtle. Fixing $i \in \{1, \ldots , m_1 \}$, we have to evaluate the following sum
\begin{equation} \label{eq:sum_star}
\begin{split}
& \sum_{(j_1, \ldots , j_{m_1}) \in \{0,1\}^{m_1}}  j_i  (1-\sigma)\frac{\beta^{m_1 - \bar m_1+1}(1-\beta)^{\bar m_1}
	V_{n-\bar m_1+1, k-\bar m_1}}{\E_{\bar M_{m_1}}[V_{n-\bar M_{m_1}, k-\bar M_{m_1}}]}  \\
& \qquad\qquad = \sum_{\bar m_1=0}^{m_1} \sum_{(\star)}
j_i  (1-\sigma)\frac{\beta^{m_1 - \bar m_1+1}(1-\beta)^{\bar m_1}
	V_{n-\bar m_1+1, k-\bar m_1}}{\E_{\bar M_{m_1}}[V_{n-\bar M_{m_1}, k-\bar M_{m_1}}]}
\end{split} 
\end{equation}
where the sum $(\star)$ is extended over all the possible vectors $(j_1, \ldots , j_{m_1})$ such that $\sum_{h=1}^{m_1} j_h= m_1 - \bar m_1$.
We note that if $j_i=0$ the summand on the r.h.s. of \eqref{eq:sum_star} is equal to $0$, hence we can equivalently sum over all the vectors  $(j_1, \ldots , j_{m_1})$
such that $j_i=1$. We further observe that, apart of $j_i$, the summand depends on $(j_1, \ldots , j_{m_1})$ only through $\bar m_1$. Thanks to these 
remarks, one has
\begin{align*}
& \sum_{(j_1, \ldots , j_{m_1}) \in \{0,1\}^{m_1}}  j_i (1-\sigma)\frac{\beta^{m_1 - \bar m_1+1}(1-\beta)^{\bar m_1}
	V_{n-\bar m_1+1, k-\bar m_1}}{\E_{\bar M_{m_1}}[V_{n-\bar M_{m_1}, k-\bar M_{m_1}}]}  \\
& \qquad\qquad = \sum_{\bar m_1=0}^{m_1-1} \# \{(j_1, \ldots , j_{m_1}): \; j_1+\cdots + j_{m_1}= m_1-\bar m_1, \; j_i=1\}\\
& \qquad\qquad\qquad\qquad \times
(1-\sigma)\frac{\beta^{m_1 - \bar m_1+1}(1-\beta)^{\bar m_1}
	V_{n-\bar m_1+1, k-\bar m_1}}{\E_{\bar M_{m_1}}[V_{n-\bar M_{m_1}, k-\bar M_{m_1}}]} \\
& \qquad\qquad = \sum_{\bar m_1=0}^{m_1-1}  \binom{m_1-1}{m_1 -\bar m_1-1}   (1-\sigma)\frac{\beta^{m_1 - \bar m_1+1}(1-\beta)^{\bar m_1}
	V_{n-\bar m_1+1, k-\bar m_1}}{\E_{\bar M_{m_1}}[V_{n-\bar M_{m_1}, k-\bar M_{m_1}}]}.
\end{align*}
By the fact that
\[
\binom{m_1-1}{m_1-\bar{m}_1-1} = \frac{(m_1 -\bar{m}_1)}{m_1} \binom{m_1}{\bar{m}_1}
\]
the previous expression reduces to
\begin{align*}
& \sum_{(j_1, \ldots , j_{m_1}) \in \{0,1\}^{m_1}}  j_i (1-\sigma)\frac{\beta^{m_1 -  \bar{m}_1+1}(1-\beta)^{\bar{m}_1}
	V_{n-\bar{m}_1+1, k-\bar{m}_1}}{\E_{\bar{M}_{m_1}}[V_{n-\bar{M}_{m_1}, k-\bar{M}_{m_1}}]}  \\
& \qquad = \sum_{\bar{m}_1=0}^{m_1}   \frac{(m_1 - \bar{m}_1)}{m_1} \binom{m_1}{\bar{m}_1}  (1-\sigma)\frac{\beta^{m_1 - \bar{m}_1+1}(1-\beta)^{\bar{m}_1}
	V_{n-\bar{m}_1+1, k-\bar{m}_1}}{\E_{\bar{M}_{m_1}}[V_{n-\bar{M}_{m_1}, k-\bar{M}_{m_1}}]}\\
& \qquad =  \frac{\beta (1-\sigma) \E_{\bar{M}_{m_1}}[(m_1 - \bar{M}_{m_1}) V_{n-\bar{M}_{m_1}+1, k-\bar{M}_{m_1}}]}{m_1\E_{\bar{M}_{m_1}}[V_{n-\bar{M}_{m_1}, k-\bar{M}_{m_1}}]}
\end{align*}
and this provides the second term on the r.h.s. of \eqref{eq:predictive}, after summing over $i =1, \ldots , m_1$.

\subsection{Re-sampling mechanism: details}\label{sup:prob_new_comparison}

Assume that $(X_1, \ldots , X_n)$ is a sample from an exchangeable sequence of observations $X_i | \tilde{p} \simiid \tilde{p}$, where $\tilde{p}$ is a random probability measure.
In this section we study the ratio between the probability of re-sampling a distinct value observed $n_1$ times and the probability of re-sampling a value observed $n_2$ times out of the initial sample in two distinct cases: i) $\tilde{p}$ is  a Gibbs-type prior; ii) $\tilde{p}$ is a contaminated Gibbs-type prior when $P_0=Q_0$.
From the predictive distribution  \eqref{eq:predictiveG} if $n_1, n_2 >1$ the ratio in the two cases is the same. Indeed, under the contaminated Gibbs-type prior the ratio between the two probabilities is equal to
\[
\frac{\beta (n_1-\sigma) \E_{\bar M_{m_1}}[ V_{n-\bar M_{m_1} + 1, k-\bar M_{m_1}}]\E_{\bar M_{m_1}}[V_{n-\bar M_{m_1}, k-\bar M_{m_1}}]}{
	\beta (n_2-\sigma) \E_{\bar M_{m_1}}[ V_{n-\bar M_{m_1} + 1, k-\bar M_{m_1}}]\E_{\bar M_{m_1}}[V_{n-\bar M_{m_1}, k-\bar M_{m_1}}]} 
= \frac{(n_1 - \sigma)}{(n_2 - \sigma)},
\]
which is exactly the same for Gibbs-type prior. On the other hand, if $n_1=1$ and $n_2>1$, the ratio in the contaminated model decreases w.r.t. the same quantity under a  Gibbs-type prior specification. Indeed, the probability ratio when $\tilde{p}$ is a contaminated Gibbs-type prior equals
\[
\frac{\beta (1-\sigma) \E_{\bar M_{m_1}}	\left[\frac{(m_1 - \bar M_{m_1})}{m_1} V_{n-\bar M_{m_1} + 1, k-\bar M_{m_1}}\right]\E_{\bar M_{m_1}}[V_{n-\bar M_{m_1}, k-\bar M_{m_1}}]}{
	\beta (n_2-\sigma) \E_{\bar M_{m_1}}[ V_{n-\bar M_{m_1} + 1, k-\bar M_{m_1}}]\E_{\bar M_{m_1}}[V_{n-\bar M_{m_1}, k-\bar M_{m_1}}]} 
%&\leq \frac{\beta (1-\sigma) \E_{\bar M_{m_1}}[ V_{n-\bar M_{m_1} + 1, k-\bar M_{m_1}}]\E_{\bar M_{m_1}}[V_{n-\bar M_{m_1}, k-\bar M_{m_1}}]}{
%	\beta (n_2-\sigma) \E_{\bar M_{m_1}}[ V_{n-\bar M_{m_1} + 1, k-\bar M_{m_1}}]\E_{\bar M_{m_1}}[V_{n-\bar M_{m_1}, k-\bar M_{m_1}}]} \\
\leq \frac{(1 - \sigma)}{(n_2 - \sigma)},
\]
where the last quantity corresponds to the probability ratio in the case of a Gibbs-type prior specification.	
The previous relations clarify that, conditionally on a sample $X_1, \dots, X_n$, the inclusion of a contaminant measure is preserving the same reinforcement as the discrete term of the model for observations with frequency larger than one, 
it mainly acts on singletons by decreasing the re-sampling probabilities w.r.t. observations with higher frequencies.

\subsection{Details for the determination of \eqref{eq:EPPF_convex_comb_PY}}\label{sup:eppf_PY}

We resort to Theorem \ref{thm:EPPF_Gibbs} and we specialize the expression of the EPPF in Equation \eqref{eq:EPPF_convex_comb} to the Pitman-Yor case:
\begin{align*}
\Pi_k^{(n)} (n_1, \ldots , n_k) &=  \sum_{\bar{m}_1=0}^{n_1} \binom{m_1}{\bar{m}_1} \beta^{n-\bar{m}_1} (1-\beta)^{\bar{m}_1}   \prod_{i=m_1+1}^k  (1-\sigma)_{n_i-1}   \frac{\prod_{i=1}^{k-\bar{m}_1-1} (\vartheta+i \sigma)}{(\vartheta+1)_{n-\bar{m}_1-1}}
\end{align*}
where we used in \eqref{eq:EPPF_convex_comb} the specification of the $V_{n,k}$'s of the Pitman-Yor process \eqref{eq:w_PY}. It can be easily seen that
\begin{align*}
\Pi_k^{(n)} (n_1, \ldots , n_k) &=  \prod_{i=m_1+1}^k  (1-\sigma)_{n_i-1}  
\sum_{\bar{m}_1=0}^{m_1} \binom{m_1}{\bar{m}_1} \beta^{n-\bar{m}_1} (1-\beta)^{\bar{m}_1}  \\
& \qquad\qquad\qquad \times
\frac{\sigma^{k-\bar{m}_1-1} \Gamma (\vartheta/\sigma+k-\bar{m}_1) \Gamma (\vartheta+1)}{\Gamma (\vartheta/\sigma+1) \Gamma (\vartheta+n-\bar{m}_1)}\\
&=  \prod_{i=m_1+1}^k  (1-\sigma)_{n_i-1} \beta^{n-m_1}
\sum_{\bar{m}_1=0}^{m_1} \binom{m_1}{\bar{m}_1} \beta^{m_1 - \bar{m}_1} (1-\beta)^{\bar{m}_1}  \frac{\sigma^{k - \bar{m}_1} (\vartheta/\sigma)_{k - \bar{m}_1}}{(\vartheta)_{n - \bar{m}_1}}.
\end{align*}

\subsection{Proof of Proposition \ref{prp:asymptotic}}

We start from the augmented model 
\begin{equation} \label{eq:model_prp}
\begin{split}
X_i \mid \tilde{p}, J_i  & \simiid J_i \tilde{q} + (1 - J_i) P_0 \\
J_i  &\simiid {\rm Bern} (\beta).
\end{split}
\end{equation}
We now use a slightly different notation to underline the dependence w.r.t. $\beta$ of the cluster frequencies, in particular if $V$ is a random variable depending on $(X_1 , \ldots , X_n)$ in \eqref{eq:model_prp}, we write $V(\beta)$ to make explicit the dependence on $\beta$.  If $\beta_1 < \beta_2$ we want to prove that $K_n (\beta_1)$ stochastically dominates $K_n (\beta_2)$, and the same for $M_{n,1}$.
We recall that $\bar M_{n} (\beta)  = \sum_{i=1}^n (1-J_i (\beta))$, which identifies the number of components coming from $P_0$, therefore the following equality holds true a.s.
\begin{equation} \label{eq:repK_n}
K_n (\beta)= \tilde{K}_{n-\bar M_{n}(\beta)} + \bar M_{n} (\beta)
\end{equation}
where $\tilde{K}_{n-\bar M_{n} (\beta)}|\bar M_{n} (\beta)$ is the number of distinct observations among the ones generated by the process $\tilde{q}$.
Now let $\beta_1 < \beta_2$, under the contaminated model with parameter $\beta_2$, $\sum_{i=1}^n J_i (\beta_2)$ represents the number of observations assigned to the discrete component $\tilde{q}$.  Now we consider $n$ independent Bernoulli random variables $U_1, \ldots , U_n$ with mean $\beta_1/\beta_2$ and we stochastically assign some of the $\sum_{i=1}^n J_i (\beta_2)$ observations to the diffuse component. The updated number of observations assigned to the diffuse component equals $\sum_{i=1}^n J_i (\beta_2) U_i$, which is a Binomial with parameters $n$ and $\beta_1$. As a consequence, the following chain holds true:
\begin{align*}
K_n (\beta_2) & = \tilde{K}_{\sum_{i=1}^n J_i (\beta_2)} + n-\sum_{i=1}^n J_i (\beta_2) \\
& \leq 
\tilde{K}_{\sum_{i=1}^n U_i J_i (\beta_2)} +\sum_{i=1}^n J_i (\beta_2)-\sum_{i=1}^n J_i (\beta_2) U_i + n-\sum_{i=1}^n J_i (\beta_2)\\
& 
= \tilde{K}_{\sum_{i=1}^n U_i J_i (\beta_2)}  + n -\sum_{i=1}^n J_i (\beta_2) U_i
\end{align*}
where we used the fact that the distinct number of observations $\tilde{K}_{\sum_{i=1}^n J_i (\beta_2)} $ is always less than  the number of observations still assigned to the discrete components plus the stochastic number of observations now assigned to $P_0$. Note that the distribution of the random variable on the right hand side of the previous equation equals the distribution of 
\[
\tilde{K}_{n-\bar M_{n}(\beta_1)} + \bar M_{n} (\beta_1)= K_n (\beta_1)
\]
as a consequence $\P (K_n (\beta_2) \geq x) \leq \P (K_n (\beta_1) \geq x) $, or in other words if $\beta_1 < \beta_2$, then $K_n (\beta_1)$ stochastically dominates $K_n (\beta_2)$.
As for the observations with frequency $1$, one can observe that
\[
M_{n,1} (\beta)= \tilde{M}_{n-\bar M_{n},1 (\beta)}+\bar M_{n} (\beta)
\]
where $\tilde{M}_{n-\bar M_{n} (\beta),1}$ denotes the number of observations with frequency $1$ among those with $J_i (\beta)=1$. Along similar lines as  
before, it is not difficult to see that if $\beta_1 < \beta_2$, then $M_{n,1} (\beta_1)$ stochastically dominates $M_{n,2} (\beta_2)$.

We now move to prove the asymptotic results, and we drop the explicit dependence on $\beta$ of the sample statistics.
We first focus on $K_n$, by exploiting the representation in \eqref{eq:repK_n} one observe that
$\tilde{K}_{n-\bar M_{n}}|\bar M_{n}$ is the number of distinct values in a sample of size $n-\bar M_{n}$ generated by a Gibbs-type prior $\tilde{q}$, then
\[
\P(\tilde{K}_{m}=k ) = \frac{V_{m,k}}{\sigma^k} \Ccr (m,k; \sigma), \quad \text{with }  \Ccr (m,k; \sigma) = \frac{1}{k!} \sum_{i=0}^k (-1)^i
\binom{k}{i} (-i \sigma)_m
\]
see, e.g., \citep{Gne05} and \citep{Lij07}. 
Then, the distribution of $K_n$ can be recovered by marginalizing out the distribution of $\bar M_{n}$ as follows
\begin{align*}
\P  (K_n=k) &= \sum_{\ell=0}^k \P (K_n=k| \bar M_{n}=\ell) \P (\bar M_{n}=\ell) \\
&= \sum_{\ell=0}^k \P (K_n=k| \bar M_{n}=\ell) \binom{n}{\ell} \beta^{n-\ell} (1-\beta)^{\ell} \\
&= \sum_{\ell=0}^k \P (\tilde{K}_{n-\ell}=k-\ell) \binom{n}{\ell} \beta^{n-\ell} (1-\beta)^{\ell}\\
& = \sum_{\ell=0}^k  \binom{n}{\ell} \beta^{n-\ell} (1-\beta)^{\ell} \frac{V_{n-\ell,k-\ell}}{\sigma^{k-\ell}} \Ccr (n-\ell,k-\ell; \sigma).
\end{align*}
We can easily study the asymptotic distribution of $K_n$, by observing that 
\[
\frac{\bar M_{n}}{n} = \sum_{i=1}^n \frac{1-J_i}{n} \stackrel{a.s.}{\to} (1 -\beta), \qquad \text{as  } n \to +\infty,
\]
by the strong law of large numbers. Thanks to the previous equation we have that $n - \bar M_{n} = n (1 - \frac{\bar M_{n}}{n}) \stackrel{a.s.}{\to} +\infty$, since $1 - \frac{\bar M_{n}}{n}$ is almost surely positive and bounded as $n \to +\infty$. We further notice that 
\[
\frac{\tilde{K}_{n-\bar M_{n}}}{(n-\bar M_{n})^\sigma}\to S_\sigma
\]
where $S_\sigma$ denotes a finite random variable termed $\sigma$-diversity \citep[see, e.g.,][]{Pit06, Deb15}. We can conclude that
\[
\frac{K_n}{n} = \frac{\tilde{K}_{n-\bar M_{n}} + \bar M_{n}}{n} = \frac{\tilde{K}_{n-\bar M_{n}}}{(n-\bar M_{n})^\sigma} \times \frac{(n-\bar M_{n})^\sigma}{n} + \frac{\bar M_{n}}{n} \to 1-\beta
\]
almost surely as $n \to +\infty$.\\

We now study the convergence of $M_{n,r}$, the number of types having frequency $r$ in the sample. We first recall that if $\tilde{M}_{n,r}$ is the number of unique elements with frequency $r$ in a sample of size $n$ generated from a Gibbs-type priors, then, thanks to \citep[Lemma 3.11]{Pit06}, one has
\begin{equation}
\label{eq:PY_freq}
\frac{\tilde{M}_{n,r}}{n^\sigma} \to \frac{\sigma (1-\sigma)_{r-1}}{r!}  S_{\sigma}, \qquad \text{as  } n \to +\infty.
\end{equation}
See also \citep{Fav13} for an explicit expression of the distribution of $\tilde{M}_{n,r}$.
For the case $r=1$, some observations are generated from $P_0$ and others from $\tilde{q}$, in formulas
\[
M_{n,1}= \tilde{M}_{n-\bar M_{n},1}+\bar M_{n}
\]
where $\tilde{M}_{n-\bar M_{n},1}$ denotes the number of observations with frequency $1$ among those with $J_i=1$. Thanks to \eqref{eq:PY_freq} we have that 
\[
\frac{\tilde{M}_{n-\bar M_{n},1}}{(n-\bar M_{n})^\sigma}  \to  \sigma S_{\sigma},
\]
since the random quantity $n-\bar M_{n} $ diverges as $n$ grows to infinity. As a consequence we obtain
\[
\frac{M_{n,1}}{n} =  \frac{ \tilde{M}_{n-\bar M_{n},1}+\bar M_{n}}{n} = \frac{ \tilde{M}_{n-\bar M_{n},1}}{(n-\bar M_{n})^\sigma}\times \frac{(n-\bar M_{n})^\sigma}{n} +\frac{\bar M_{n}}{n}  \to
1-\beta, \qquad \text{as  } n \to +\infty.
\]

If we now concentrate on the case $r \geq 2$, all the observations are generated from $\tilde{q}$, and we observe that
\[
\frac{M_{n,r}}{n^\sigma} = \frac{\tilde{M}_{n-\bar M_{n},r}}{n^\sigma}= \frac{\tilde{M}_{n-\bar M_{n},r}}{(n-\bar M_{n})^\sigma}\frac{(n-\bar M_{n})^\sigma}{n^\sigma}  \to
\frac{\sigma (1-\sigma)_{r-1}}{r!}  S_{\sigma} \beta^\sigma, \qquad \text{as  } n \to +\infty,
\]
thanks to \eqref{eq:PY_freq} and the fact that $\frac{\bar M_{n}}{n} \to 1-\beta$.

\section{Posterior inference for contaminated Pitman-Yor processes}  \label{sec:app_cPY}

In the present section we face posterior inference for the contaminated Pitman-Yor process of Example \ref{ex:PY3}. In particular, conditionally on a sample $(X_1, \ldots , X_n)$ of size $n$ we derive closed-form expressions for the posterior expected value of the following statistics: i)  $K_m^{(n)}$, i.e., 
the number of distinct values out a future sample  $X_{n+1}, \ldots , X_{n+m} $ not yet observed in the initial sample; ii) $N_{m,r}^{(n)}$, which denotes the number of new and distinct observations with frequency $r$ out of the additional sample, hitherto unobserved in the initial sample. By virtue of the posterior results, we also get formulas for the expected value of: i) $K_n$, the number of distinct values out of $(X_1, \ldots , X_n)$; ii) $M_{n,r}$, the number of clusters with frequency $r$ out of $(X_1, \ldots , X_n)$. Note that, in our framework, it is important to focus separately on the case $r=1$ and $r \geq 2$, since the contaminated model acts in a different way on the number of observations with frequency one.
Our results are based on the expressions derived by \cite{Fav_09,Fav13}, who have faced posterior inference for the number of blocks with a certain frequency generated by Gibbs-type random partitions.
It is possible to extend the results by \cite{Fav_09,Fav13} to contaminated Gibbs-type priors, here, for the easy of exposition, we discuss the contaminated Pitman-Yor case.

\subsection{Posterior expected value of $N_{m,1}^{(n)}$} \label{sec:posterior_new1}

Conditionally on a sample $X_1, \ldots , X_n$, we focus on predicting the number of new and distinct observations out of the additional sample $X_{n+1}, \ldots , X_{n+m} $ observed with frequency $1$, denoted here as $N_{m,1}^{(n)}$. This is an important quantity in our framework, indeed the contaminated model mainly acts on the number of observations with frequency $1$. 
We introduce the following random variables: i) $\bar{M}_{n}= \sum_{i=1}^n (1-J_i)$ the number of observations generated by the diffuse component out of the sample of size $n$; ii) $\bar{M}_{m}^{(n)}:= \sum_{i=1}^m (1-J_{i+n})$  the number of observations generated by the diffuse component out of the additional sample of size $m$. Note that these random variables have Binomial distributions and are independent.\\
We now focus on the evaluation of $\E [N_{m,1}^{(n)}| X_1, \ldots , X_n, \bar{M}_{n} ]$, we first observe  that
\[
N_{m,1}^{(n)} =  \tilde{N}_{m-\bar{M}_{m}^{(n)},1}^{(n-\bar{M}_{n})}+\bar{M}_{m}^{(n)}
\]
where $\tilde{N}_{m-\bar{M}_{m}^{(n)},1}^{(n-\bar{M}_{n})}$ denotes the number of distinct values out of the additional sample observed with frequency $1$ and coming from the discrete component, whose posterior expectation has been derived by \cite[Equation (30)]{Fav13}. As a consequence we get:
\begin{align*}
&\E [N_{m,1}^{(n)}| X_1, \ldots , X_n, \bar{M}_{n} ]  =  \E [\tilde{N}_{m-\bar{M}_{m}^{(n)},1}^{(n-\bar{M}_{n})}+\bar{M}_{m}^{(n)} | X_1, \ldots , X_n, 
\bar{M}_{n} ]\\
&\qquad = \E [\E [\tilde{N}_{m-\bar{M}_{m}^{(n)},1}^{(n-\bar{M}_{n})}+\bar{M}_{m}^{(n)} | X_1, \ldots , X_n, \bar{M}_{n}, \bar{M}_{m}^{(n)} ]| X_1, \ldots , X_n, \bar{M}_{n} ]\\
&\qquad = \E \Big[ (m - \bar{M}_{m}^{(n)}) (\vartheta+(k-\bar{M}_{n})\sigma)\cdot \frac{(\vartheta+n-\bar{M}_{n}+\sigma)_{m-\bar{M}_{m}^{(n)}-1}}{(\vartheta+n-\bar{M}_{n})_{m-\bar{M}_{m}^{(n)}}} \\
& \qquad\qquad\qquad\qquad\qquad\qquad\qquad\qquad\qquad\qquad\qquad+\bar{M}_{m}^{(n)}\Big| X_1, \ldots , X_n, \bar{M}_{n}\Big]
\end{align*}
where we used \cite[Equation (30)]{Fav13}. By observing that $\bar{M}_{m}^{(n)}$ has a Binomial distribution with parameters $(m, 1-\beta)$ and it is independent of $X_1, \ldots ,X_n$, we obtain
\begin{equation}  \label{eq:N_posterior_noint}
\begin{split}
&\E [N_{m,1}^{(n)}| X_1, \ldots , X_n, \bar{M}_{n} ]\\
& \quad= \sum_{\ell=0}^m \Big[ (m - \ell) (\vartheta+(k-\bar{M}_{n})\sigma)\cdot \frac{(\vartheta+n-\bar{M}_{n}+\sigma)_{m-\ell-1}}{(\vartheta+n-\bar{M}_{n})_{m-\ell}} +
\ell\Big]  \binom{m}{\ell} (1-\beta)^\ell \beta^{m-\ell}.
\end{split}
\end{equation}
In order to find $\E [N_{m,1}^{(n)}| X_1, \ldots , X_n]$, we need to marginalize the previous expectation with respect to the conditional distribution of $\bar{M}_{n}$ given $X_1, \ldots , X_n$, which can be derived from the EPPF:
\begin{equation}
\label{eq:condMbar}
\P (\bar{M}_{n}= \bar{m}_1 | X_1, \ldots , X_n )  = \frac{1}{C}  \cdot \binom{m_1}{\bar{m}_1}  \beta^{n-\bar{m}_1} (1-\beta)^{\bar{m}_1} \sigma^{k-\bar{m}_1}
\frac{\Gamma (\vartheta/\sigma+k - \bar{m}_1)}{\Gamma (\vartheta+n - \bar{m}_1)}
\end{equation}
where $C$ is the normalizing factor, i.e., 
\[
C =  \sum_{s=0}^{m_1} \binom{m_1}{s}  \beta^{n-s} (1-\beta)^{s} \sigma^{k-s}
\frac{\Gamma (\vartheta/\sigma+k - s)}{\Gamma (\vartheta+n - s)} .
\]
Thus integrating \eqref{eq:N_posterior_noint} with respect to the distribution \eqref{eq:condMbar} we get
\begin{equation}  \label{eq:N_posterior}
\begin{split}
&\E [N_{m,1}^{(n)}| X_1, \ldots , X_n]= \sum_{\bar{m}_1=0}^{m_1} \sum_{\ell=0}^m \Big[ (m - \ell) (\vartheta+(k-\bar{m}_1)\sigma)\cdot \frac{(\vartheta+n-\bar{m}_1+\sigma)_{m-\ell-1}}{(\vartheta+n-\bar{m}_1)_{m-\ell}} \\
& \qquad+
\ell\Big]  \binom{m}{\ell} (1-\beta)^{\ell+\bar{m}_1} \beta^{m-\ell+n-\bar{m}_1}   \frac{1}{C}  \cdot \binom{m_1}{\bar{m}_1}   \sigma^{k-\bar{m}_1}
\frac{\Gamma (\vartheta/\sigma+k - \bar{m}_1)}{\Gamma (\vartheta+n - \bar{m}_1)}.
\end{split}
\end{equation}
%which can also be written as
%\[
%\begin{split}
%&\E [N_{m,1}^{(n)}| X_1, \ldots , X_n] \\
%& \quad=  \E_{\bar{M}_{n}| X_1, \ldots , X_n}  \Big[ \E_{L_m}  \Big[
%(m-L_m) (\vartheta+(k -\bar{M}_{n})\sigma) \cdot  \frac{(\vartheta+n-\bar{M}_{n}+\sigma)_{m-L_m-1}}{
%(\vartheta+n -\bar{M}_{n})_{m-L_m}}  +L_m
%\Big]  \Big]
%\end{split}
%\]
%where $L_m \sim {\rm Binom} (m, 1-\beta)$ and the distribution of $\bar{M}_{n}| X_1, \ldots , X_n$ is given in \eqref{eq:condMbar}.\\
By a close inspection of the conditional expected value in \eqref{eq:N_posterior}, it is immediate to realize that it depends on the initial sample $X_1, \ldots , X_n$ through the sample size $n$, the number of distinct species $k$ and also by $m_1$, i.e., the number of species having frequency $1$ in the initial sample. This is a remarkable difference with respect to Gibbs-type priors, in which the posterior expected value depends only on $n$ and $k$ \citep{Bat_17}, thus improving the flexibility of the model.
Finally one may use \eqref{eq:N_posterior} to determine the distribution of $M_{n,1}$ for contaminated model by simply setting $n=0$ and replacing $m$ with $n$, thus obtaining:
\begin{equation}  \label{eq:M_prior}
\begin{split}
&\E [M_{n,1}]= n(1-\beta)+ \sum_{\ell=0}^n   (n - \ell)\frac{(\vartheta+\sigma)_{n-\ell-1}}{(\vartheta+1)_{n-\ell-1}}  \binom{n}{\ell} \beta^{n-\ell} (1-\beta)^{\ell}.
\end{split}
\end{equation}
We finally  provide the reader with another representation of the summation appearing in the previous equation \eqref{eq:M_prior}, by a change of variable $r=n-\ell$ we get
\begin{align*}
& \sum_{\ell=0}^n   (n - \ell)\frac{(\vartheta+\sigma)_{n-\ell-1}}{(\vartheta+1)_{n-\ell-1}}  \binom{n}{\ell} \beta^{n-\ell} (1-\beta)^{\ell}  = \sum_{r=1}^n r  \frac{(\theta+\sigma)_{r-1}}{(\theta+1)_{r-1}}  \binom{n}{r}   \beta^{r} (1-\beta)^{n-r}\\
& \qquad = \sum_{r=0}^{n-1} (r+1)  \frac{(\theta+\sigma)_{r}}{(\theta+1)_{r}}  \binom{n}{r+1}   \beta^{r+1} (1-\beta)^{n-1-r} \\
& \qquad = n \beta \sum_{r=0}^{n-1}  \frac{(\theta+\sigma)_{r}}{(\theta+1)_{r}}  \binom{n-1}{r}   \beta^{r} (1-\beta)^{n-1-r}.
\end{align*}
We now observe that the ratio between Pochhammer symbols is the expected value of $B_1^r$, where $B_1$ is a random variable having a Beta distribution with parameters $(\vartheta+\sigma, 1-\sigma)$, then:
\begin{align*}
& \sum_{\ell=0}^n   (n - \ell)\frac{(\vartheta+\sigma)_{n-\ell-1}}{(\vartheta+1)_{n-\ell-1}}  \binom{n}{\ell} \beta^{n-\ell} (1-\beta)^{\ell}  = n \beta  \sum_{r=0}^{n-1}  \E [B_1^r]  \binom{n-1}{r} \beta^{r} (1-\beta)^{n-1-r} \\
& \qquad =  n \beta  \E [(\beta B_1 +(1-\beta))^{n-1}].
\end{align*}
As a consequence the expected value of $M_{n,1}$ boils down to
\begin{equation} \label{eq:EM_beta}
\E [M_{n,1}] =  n (1-\beta) +  n \beta  \E [(\beta B_1 +(1-\beta))^{n-1}].
\end{equation}

\subsection{Posterior expected value of $N_{m,r}^{(n)}$, with $r \geq 2$} \label{sec:posterior_newr}

Conditionally on a sample $X_1, \ldots , X_n$, we focus on predicting the number of new and  distinct observations out of the additional sample $X_{n+1}, \ldots , X_{n+m} $ observed with frequency $r\geq 2$, denoted here as $N_{m,r}^{(n)}$. We exploit the same notation introduced in Section \ref{sec:posterior_new1}\\
We now focus on the evaluation of $\E [N_{m,r}^{(n)}| X_1, \ldots , X_n, \bar{M}_{n} ]$, we first observe  that
\[
N_{m,r}^{(n)} =  \tilde{N}_{m-\bar{M}_{m}^{(n)},r}^{(n-\bar{M}_{n})}
\]
where $\tilde{N}_{m-\bar{M}_{m}^{(n)},r}^{(n-\bar{M}_{n})}$ takes into account the contribution of the discrete component $\tilde{q}$ and it denotes the number of clusters containing  $r$ observations out of the additional sample
and not observed in $X_1, \ldots , X_n$. The posterior expected value of $\tilde{N}_{m-\bar{M}_{m}^{(n)},r}^{(n-\bar{M}_{n})}$ has been found in \cite[Equation (30)]{Fav13}, thus:
\begin{align*}
&\E [N_{m,r}^{(n)}| X_1, \ldots , X_n, \bar{M}_{n} ]  = \E [ \E [\tilde{N}_{m-\bar{M}_{m}^{(n)},r}^{(n-\bar{M}_{n})} |X_1, \ldots , X_n, \bar{M}_{n}, \bar{M}_{m}^{(n)}] | X_1, \ldots , X_n, \bar{M}_{n} ]\\
&\qquad\qquad = \E \Big[   \binom{m-\bar{M}_{m}^{(n)}}{r}  (1-\sigma)_{r-1}  (\vartheta+ (k-\bar{M}_{n})\sigma)\\
& \qquad\qquad\qquad\qquad \times
\frac{(\vartheta+n-\bar{M}_{n}+\sigma)_{m-\bar{M}_{m}^{(n)}-r}}{(\vartheta+n -\bar{M}_{n})_{m-\bar{M}_{m}^{(n)}}} \Big|   X_1, \ldots , X_n, \bar{M}_{n} \Big] .
\end{align*}
By marginalizing the previous expression over $\bar{M}_{m}^{(n)}$, which is distributed as a Binomial with parameters $m$ and $1-\beta$, we get
\begin{equation} \label{eq:Nmr1}
\begin{split}
&\E [N_{m,r}^{(n)}| X_1, \ldots , X_n, \bar{M}_{n} ] =
\sum_{\ell=0}^{m-r}  \binom{m-\ell}{r}  (1-\sigma)_{r-1} (\vartheta+\sigma (k-\bar{M}_{n}))\\
& \qquad\qquad\qquad\qquad\qquad\qquad \times \frac{(\vartheta+n-\bar{M}_{n}+\sigma)_{m-r-\ell}}{(\vartheta+n-\bar{M}_{n})_{m-\ell}}  \binom{m}{\ell} (1-\beta)^\ell \beta^{m-\ell}.
\end{split}
\end{equation}
The posterior expected value $\E [N_{m,r}^{(n)}| X_1, \ldots , X_n]$ can be easily evaluated by marginalizing \eqref{eq:Nmr1} with respect to the posterior  distribution of $ \bar{M}_{n}$, which appears in \eqref{eq:condMbar}. Thus, we obtain:
\begin{equation} \label{eq:posterior_newr}
\begin{split}
&\E [N_{m,r}^{(n)}| X_1, \ldots , X_n] \\
& \qquad=
\sum_{\bar{m}_1=0}^{m_1}\sum_{\ell=0}^{m-r}  \binom{m-\ell}{r}  (1-\sigma)_{r-1} (\vartheta+\sigma (k-\bar{m}_1))\frac{(\vartheta+n-\bar{m}_1+\sigma)_{m-r-\ell}}{(\vartheta+n-\bar{m}_{1})_{m-\ell}} \\
& \qquad\qquad\qquad \times  \binom{m}{\ell} \beta^{m+n-\ell-\bar{m}_1}
(1-\beta)^{\ell+\bar{m}_1}  \frac{\sigma^{k-\bar{m}_1}}{C} \binom{m_1}{\bar{m}_1} \frac{\Gamma (\vartheta/\sigma+k-\bar{m}_1)}{\Gamma (\vartheta+n-\bar{m}_1)}.
\end{split}
\end{equation}
As a consequence of \eqref{eq:posterior_newr}, one may derive an expression for $\E[M_{n,r}]$, considering $n=0$ and substituting $n$ in place of $m$:
\begin{equation}
\label{eq:EMr_prior}
\E [M_{n,r} ] =  \sum_{\ell=0}^{n-r} \binom{m-\ell}{r} (1-\sigma)_{r-1} \vartheta \frac{(\vartheta+\sigma)_{n-r-\ell}}{(\vartheta)_{n-\ell}}
\binom{n}{\ell}  \beta^{n-\ell} (1-\beta)^{\ell}
\end{equation}
and, proceeding along similar lines as in Section \ref{sec:posterior_new1}, one may easily see that
\begin{equation} \label{eq:Mnr_beta}
\E [M_{n,r} ] = \frac{(1-\sigma)_{r-1}}{(\vartheta+1)_{r-1}} \binom{n}{r} \beta^r \E[(B_r \beta +1-\beta)^{n-r}]
\end{equation}
where $B_r$ is a Beta distribution with parameters $(\vartheta+\sigma, r-\sigma)$.

\subsection{Posterior expected value of $K_m^{(n)}$} \label{sec:posterior_new}

Here we focus on $K_m^{(n)}$, which represents the number of new and distinct observations out of an additional sample $X_{n+1}, \ldots , X_{n+m}$, hitherto unobserved in $X_1, \ldots , X_n$. More specifically,  we are interested in the evaluation of its posterior expected value. Note that
\[
K_m^{(n)} = \tilde{K}_{m-\bar{M}_{m}^{(n)}}^{(n-\bar{M}_{n})} +\bar{M}_{m}^{(n)},
\]
in which  we have decomposed the new clusters generated by the discrete component ($\tilde{K}_{m-\bar{M}_{m}^{(n)}}^{(n-\bar{M}_{n})}$) and the ones due to the diffuse component ($\bar{M}_{m}^{(n)}$). By resorting to \citep[Equation (6)]{Fav_09}, we have
\begin{align*}
&\E [K_m^{(n)}| X_1, \ldots , X_n , \bar{M}_{n}] \\
& \qquad=  \E [\E [\tilde{K}_{m-\bar{M}_{m}^{(n)}}^{(n-\bar{M}_{n})} +\bar{M}_{m}^{(n)}| X_1, \ldots , X_n , \bar{M}_{n}, \bar{M}_{m}^{(n)} ]|X_1, \ldots , X_n , \bar{M}_{n}]\\
& \qquad= \E [\E [\tilde{K}_{m-\bar{M}_{m}^{(n)}}^{(n-\bar{M}_{n})} | X_1, \ldots , X_n , \bar{M}_{n}, \bar{M}_{m}^{(n)} ]
+\bar{M}_{m}^{(n)}|X_1, \ldots , X_n , \bar{M}_{n}]\\
& \qquad =  \E  \Big[ (k-\bar{M}_{n}+\vartheta/\sigma)  \Big(\frac{(\vartheta+n -\bar{M}_{n}+\sigma)_{m-\bar{M}_{m}^{(n)}}}{(\vartheta+n-\bar{M}_{n})_{m-\bar{M}_{m}^{(n)}}}  -1\Big) +\bar{M}_{m}^{(n)} \Big| X_1, \ldots , X_n , \bar{M}_{n}\Big].
\end{align*}
The expected value in the r.h.s. of the previous expression is taken with respect to $\bar{M}_{m}^{(n)}$, having a Binomial distribution with parameters 
$m$ and probability of success $1-\beta$, therefore
\begin{equation}
\label{eq:Kmn1}
\begin{split}
&\E [K_m^{(n)}| X_1, \ldots , X_n , \bar{M}_{n}]\\
& \qquad =  \sum_{\ell=0}^m  \left[ (k-\bar{M}_{n}+\vartheta/\sigma)  \Big(\frac{(\vartheta+n -\bar{M}_{n}+\sigma)_{m-\ell}}{(\vartheta+n-\bar{M}_{n})_{m-\ell}}  -1\Big) +\ell  \right] \binom{m}{\ell} (1-\beta)^\ell \beta^{m-\ell}.
\end{split}
\end{equation}
In order to obtain the posterior expected value of $K_{m}^{(n)}$ we need to marginalize the r.h.s. of \eqref{eq:Kmn1} with respect to the posterior distribution of $\bar{M}_{n}$, which has been found in \eqref{eq:condMbar}, as a consequence we get
\begin{equation}
\label{eq:Kmn_posterior}
\begin{split}
\begin{split}
&\E [K_m^{(n)}| X_1, \ldots , X_n ] =  \sum_{\bar{m}_1=0}^{m_1} \sum_{\ell=0}^m  \left[ (k-\bar{m}_1+\vartheta/\sigma)  \Big(\frac{(\vartheta+n -\bar{m}_1+\sigma)_{m-\ell}}{(\vartheta+n-\bar{m}_1)_{m-\ell}}  -1\Big) +\ell  \right] \\
& \qquad\qquad\qquad\qquad\qquad \times\binom{m}{\ell} \beta^{m+n-\bar{m}_1-\ell}(1-\beta)^{\ell+\bar{m}_1}  \frac{\sigma^{k-\bar{m}_1}}{C} \binom{m_1}{\bar{m}_1}
\frac{\Gamma (\vartheta/\sigma+k-\bar{m}_1)}{\Gamma (\vartheta+n-\bar{m}_1)}   .
\end{split}
\end{split}
\end{equation}
We can now evaluate the expectation of the number of distinct values out of the initial sample, $K_n$, by setting $n=0$
in \eqref{eq:Kmn_posterior} and replacing $m$ with $n$, more precisely we obtain
\begin{equation}
\label{eq:EK_n}
\E [K_n] =  \sum_{\ell=0}^{n}  \left[ \frac{\vartheta}{\sigma} \Big( \frac{(\vartheta+\sigma)_{n-\ell}}{(\vartheta)_{n-\ell}} -1\Big) +\ell \right]   \binom{n}{\ell} (1-\beta)^\ell  \beta^{n-\ell}.
\end{equation}
By reasoning as in Section \ref{sec:posterior_new1}, it is also possible to rewrite the expected value in \eqref{eq:EK_n} in a simpler way
\begin{equation} \label{eq:Kn_beta}
\E [K_n] = \frac{\vartheta}{\sigma} \E  (B_1\beta +1-\beta)^n  +\frac{n \beta}{\sigma} \E [B_1  (B_1\beta+1-\beta)^{n-1}]-\frac{\vartheta}{\sigma} +n (1-\beta)
\end{equation}
where $B_1$ is a Beta with parameters $\vartheta+\sigma$ and $1-\sigma$.  

\section{Strip-and-solid generalized P\'olya urn}\label{sup:urn}

The predictive distribution arising from the contaminated Pitman-Yor process may be described in terms of a generalization of the urn scheme by \cite{Zab97}. The predictive distribution of the contaminated Pitman-Yor process, conditionally on the observations and the latent variables, is given in  \eqref{eq:predictive}:
\begin{equation}
\label{eq:pred_PY_supp}
\begin{split}
&\P (X_{n+1} \in \D x | X_{1:n}, J_{1:m_1} ) =  (1 - \beta)P_0(\D x) + \beta \frac{\vartheta + (k - \bar{M}_{m_1}) \sigma}{\vartheta + n - \bar{M}_{m_1}} Q_0(\D x)\\
&\qquad\qquad + \sum_{i=1}^{m_1}J_i \beta\frac{1 - \sigma}{\vartheta + n- \bar{M}_{m_1}} \delta_{X_i^*}(\D x)+ \sum_{i=m_1 + 1}^{k}\beta\frac{n_i - \sigma}{\vartheta + n - \bar{M}_{m_1}} \delta_{X_i^*}(\D x).
\end{split}
\end{equation}
We now assume that the prior distribution for the parameter $\beta$ is a beta with parameters $\vartheta$ and $\alpha$. Thus, the distribution of $\beta$, conditionally on $X_{1:n}, J_{1:m_1}$ is  again a beta with parameters 
$(n-\bar{M}_{m_1}+\vartheta, \alpha +\bar{M}_{m_1})$, as one can realize from the augmented version of the EPPF with the inclusion of the latent elements. Thus, by integrating \eqref{eq:pred_PY_supp} with respect to the conditional distribution of $\beta$, we obtain
\begin{equation}
\label{eq:pred_PY_urn}
\begin{split}
&\P (X_{n+1} \in \D x | X_{1:n}, J_{1:m_1} ) =  \frac{\bar{M}_{m_1}+\alpha}{\alpha+\vartheta +n } P_0(\D x) +  \frac{\vartheta + (k - \bar{M}_{m_1}) \sigma}{\alpha+\vartheta+n} Q_0(\D x)\\
&\qquad\qquad + \sum_{i=1}^{m_1}J_i \frac{1 - \sigma}{\vartheta+\alpha +n} \delta_{X_i^*}(\D x)+ \sum_{i=m_1 + 1}^{k}\frac{n_i - \sigma}{\vartheta +\alpha+n} \delta_{X_i^*}(\D x).
\end{split}
\end{equation}
The predictive distribution \eqref{eq:pred_PY_urn} may be described through an urn scheme.  The main difference from usual urn schemes is that here we assume the urn composed by two types of balls: strip and solid balls, where the strip balls correspond to elements associated with the contaminant measure while the solid balls can be interpreted as elements associated with the discrete term of the model. Initially the urn is composed by a weight $\alpha$ of strip colored balls and a weight $\vartheta$ of black solid balls. We want to sample an exchangeable sequence from the urn in such a way that the updating rule is \eqref{eq:pred_PY_urn}, and the balls are sampled proportionally to their weight. At the first sampling step, if a strip colored ball is drawn from the urn, then we return the ball in the urn with an additional strip colored ball of a new color. On the other side if we draw a black solid ball, then we return a black ball in the urn with an additional weight $\sigma$ and a solid ball of a new color with weight $1-\sigma$. At the generic $i$th step, one can sample a strip ball of an arbitrary color, a black solid ball or a colored solid ball.
Thus, the updating mechanism of the urn works as follows: i) if we sample a strip ball, we return the strip ball in the urn with another strip ball of a new color having unitary weight; ii) if we sample a black solid ball, we return the solid ball in the urn with a new black solid ball of weight $\sigma$ and a solid ball of a new color having  weight $1 - \sigma$; iii) if we sample a colored solid ball, we return the ball in the urn with an additional  new solid ball of the same color having weight $1$.
Some comments are in order. 
\begin{enumerate}
	\item[a)] The overlying urn scheme describing the distinction among strip and solid balls is fully described by a P\'olya-Eggenberger urn scheme \citep{Pol23}. By introducing a suitable sequence of random variables $J_1, J_2, \dots$, where the generic $J_i = 1$ if the $i$th sampled ball is solid, and $J_i = 0$ otherwise, as in the  standard theory of P\'olya urn schemes we have 
	\begin{equation}\label{eq:lim_urn}
	\lim_{n\to \infty} \frac{1}{n}\sum_{i=1}^n J_i = Z \qquad\text{with}\qquad Z \sim \mathrm{Beta}(\vartheta, \alpha).
	\end{equation}
	See e.g \citet{Mah08}.
	\item[b)] The reinforcement mechanism of the solid balls matches the urn characterization of \cite{Zab97}. In fact, by ignoring the strip balls, we have an urn scheme where if we sample a black ball, we replace the ball in the urn with another black ball of weight $\sigma$ and a ball of a new color with weight $1 - \sigma$, while if we sample a colored ball, we replace the ball in the urn plus another ball of the same color. Such a mechanism is describing the reinforcement of a sequence sampled from a Pitman-Yor process. 
	\item[c)] If we initialize the urn without strip balls, we recover the urn of \cite{Zab97}, and the distribution in \eqref{eq:lim_urn} is degenerating to a mass point at $1$. On the other side, if we initialize the urn without solid ball, the urn is producing a sequence of strip balls of different colors, and the distribution in \eqref{eq:lim_urn} is degenerating to a mass point at $0$.
\end{enumerate}    

\section{Numerical illustrations}\label{sup:new_sample}

\subsection{Comparison on the predictive probability of sampling a new value}\label{sup:prob_new}

A crucial effect of bounding a discrete probability measure with a contaminant measure is the impact of the contamination on sampled sequences. In particular, it is relevant to study how such contaminant measure acts on the probability of sampling a new value at the $(n+1)$th step, conditionally on an already observed sample $(X_1, \ldots , X_n)$. Here we consider as illustrative example a contaminated Pitman-Yor process, and we compare the results with the standard Pitman-Yor process. We show four distinct scenarios: (i) a first less diffuse scenario with $\vartheta = 0.1$ and $\sigma = 0.2$, (ii) a second more diffuse scenario with $\vartheta = 0.1$ and $\sigma = 0.5$, (iii) a third scenario with $\vartheta = 1$ and $\sigma = 0$, (iv) and a fourth diffuse scenario with $\vartheta = 10$ and $\sigma = 0$, were the last two specifications stand for the Dirichlet process case. We then consider the probability of sampling a new species at the $n+1$ step, with $n = 50$ previous observations divided into $k = 30$ distinct values. 

\begin{figure}[!h]
	\centering
	\includegraphics[width = 0.8 \textwidth]{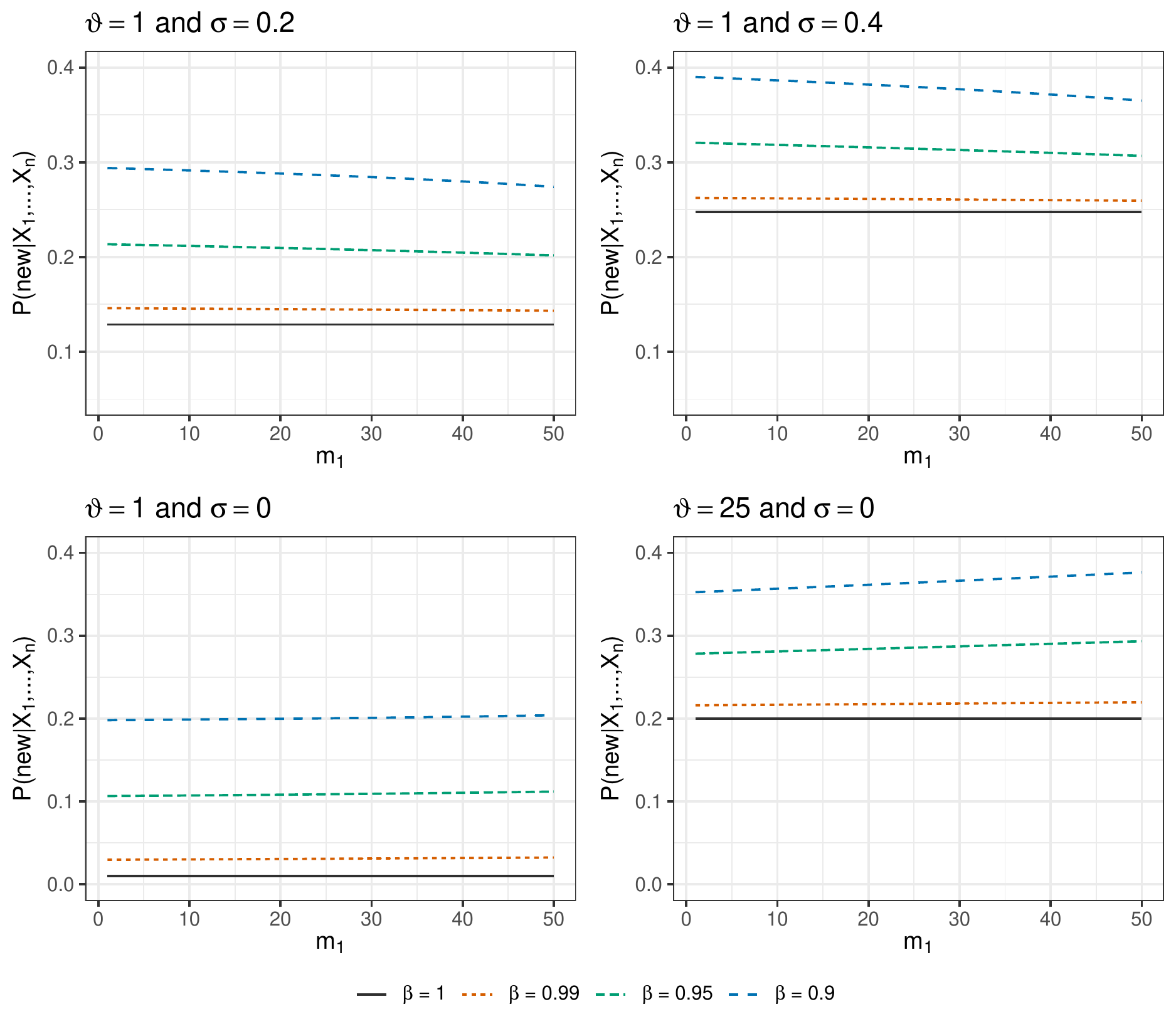}
	\caption{Probability of sampling a new value at the $(n+1)$th step, conditionally on a sample of size $n = 50$. Top-left panel: less diffuse specification, with $\vartheta = 0.1$ and $\sigma = 0.2$. Top-right panel: high diffuse specification, with $\vartheta = 0.1$ and $\sigma = 0.5$. Bottom-left panel: less diffuse specification, with $\vartheta = 1$ and $\sigma = 0$. Bottom-right panel: high diffuse specification with $\vartheta = 10$ and $\sigma = 0$. Different colors and line types correspond to different weights of the contaminant measure. The black horizontal lines corresponds to the models without contamination.}\label{fig:prob_sampling_new}
\end{figure}

One of the peculiarity of the contaminated Gibbs-type priors is that the predictive probability of sampling a new value depends on the sample sizes $n$, the number of already observed species $k$, and the number of observations with frequency one $m_1$ out of the initial sample, while in the Gibbs-type prior it does not depend on $m_1$. Such behavior is appreciable in Figure \ref{fig:prob_sampling_new}, which shows the predictive probabilities that the $(n+1)$th observation is new, i.e., it does not belong to the initial sample, as function of $m_1$, for both the contaminated Pitman-Yor process and the Pitman-Yor process, varying the weight of the contaminant measure and the specification of the discrete term of the model. The probability of sampling a new value, in the Pitman-Yor process, is a constant  function of $m_1$. As far as the weight of the discrete component is increasing, i.e. $\beta \to 1$, the probability is collapsing on the Pitman-Yor case. We further notice that the probability of sampling a new value is a decreasing function of $m_1$ when $\sigma > 0$ (contaminated Pitman-Yor process case), but it is an increasing function of $m_1$ when $\sigma = 0$ (contaminated Dirichlet process case).

\subsection{Out-of-sample prediction}

We perform a simulation study to investigate the difference between the contaminated and non-contaminated Pitman-Yor model in terms of prediction. More precisely we generate a sample $X_1, \dots, X_n$ of size $n=9\,000$ from the contaminated Pitman-Yor model with $\sigma=0.2, \vartheta=50$ and $\beta=0.9$, and we first use the sample to estimate the parameters of the two models. 
Then, we predict the posterior expected values $\E[N_{m,1}^{(n)}\mid X_1,\dots,X_n]$, $\E[N_{m,2}^{(n)}\mid X_1,\dots,X_n]$, and $\E[K_{m}^{(n)}\mid X_1,\dots,X_n]$ for different  additional sample sizes $m=1, \ldots , 1000$ for the two models, using \eqref{eq:Nmr1}, \eqref{eq:N_posterior} and \eqref{eq:Kmn_posterior} with the estimated parameters. The predicted quantities are compared with the \textit{oracle} curves  $n_{m,1}^{(n)}$, $n_{m,2}^{(n)}$ and $k_{m}^{(n)}$, which are obtained averaging over $1000$ trajectories of $N_{m,r}^{(n)}$, as $r=1,2$, and $K_{m}^{(n)}$ from the generating contaminated Pitman-Yor process. 
\begin{figure}[!h]
	\centering
	\includegraphics[width = 0.99 \textwidth]{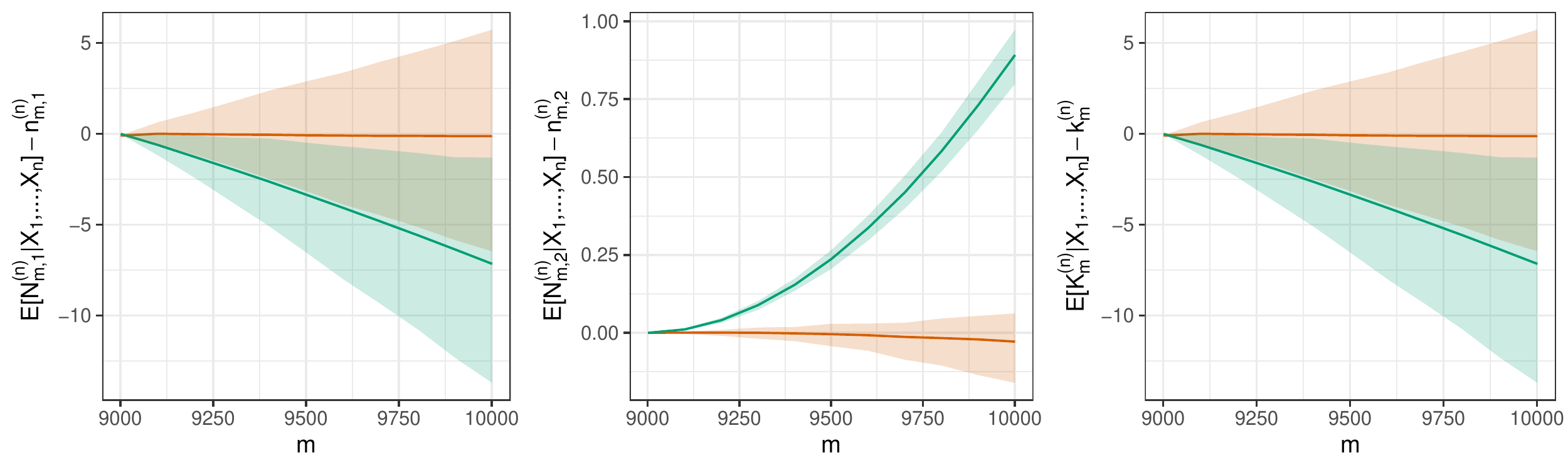}
	\caption{Simulation results of the difference between the number of new species with frequency one, the new number of species with frequency two and the number of new species in an additional sample of size $m$, discarded by their corresponding estimated true values. Two distinct models: the contaminated Pitman-Yor model (green) and the Pitman-Yor model (orange).}\label{fig:new_samp}
\end{figure}

Each panel of Figure \ref{fig:new_samp} refers to a different statistic ($N_{m,1}^{(n)}$, $N_{m,2}^{(n)}$, $K_{m}^{(n)}$).
Each panel contains two curves for different values of the additional sample size $m$: the green curve shows the difference between the  oracle curve and the predicted value under the contaminated prior; the orange curve represents the difference between the oracle and the corresponding prediction under the Pitman-Yor model.
All the experiments are averaged over $100$ iterations.
From Figure \ref{fig:new_samp}, it is appreciable how the contaminated Pitman-Yor process has an error which is averagely stable around zero for all the quantity considered, while the Pitman-Yor process, in presence of contamination, tends to underestimate the number of new elements with frequency one and the number of new distinct elements, and  it is overestimates the number of new elements with frequency two. 

\section{Algorithms}
\subsection{Discrete case}\label{sup:disc}

We can exploit the representation of the EPPF provided in Section \ref{sec:inference}  to sample realizations from the posterior distribution of the main quantities of interest. Let $\Law(\sigma)$, $\Law(\vartheta)$ and $\Law(\beta)$ denote the prior distributions  of $\sigma$, $\vartheta$ and $\beta$, respectively. Here we provide the algorithm to perform posterior inference in the case of contaminated Pitman-Yor model, given a set of data $X_{1:n} := (X_1, \dots, X_n)$, by sampling $R$ realizations from an MCMC scheme.

\begin{algo}\label{algo:discrete}
	Sampling scheme for contaminated Pitman-Yor model.
	%	\hline
	%\vspace*{-12pt}
	\begin{itemize}
		\item[(i)] Set initial values for $\vartheta$, $\bar{m}_1$; 
		\item[(ii)] For $r=1, \dots, R$:
		\begin{itemize}
			\item[(a)] Update $\sigma$ from 
			\[ 
			\begin{split}
			\qquad\Law(\sigma\mid X_{1:n}, \vartheta, \bar m_{1})  \propto \Law(\sigma)\sigma^{k - \bar m_{1}}\prod_{i=m_1+1}^k(1-\sigma)_{n_i - 1} \frac{\Gamma(\vartheta/\sigma + k - \bar m_{1})}{\Gamma(\vartheta / \sigma)};
			\end{split}
			\]
			\item[(b)] Update $\vartheta$ from 
			\[ 
			\qquad\Law(\vartheta\mid X_{1:n},\sigma, \bar m_{1}) \propto \Law(\vartheta) \frac{\Gamma(\vartheta)\Gamma(\vartheta/\sigma + k - \bar m_{1})}{\Gamma(\vartheta / \sigma)\Gamma(\vartheta + n - \bar m_{1})};
			\]
			\item[(c)] Update $\beta$ from 
			\[ 
			\qquad\Law(\beta\mid X_{1:n},\bar m_{1}) \propto \Law(\beta)\beta^{n - \bar m_{1}} (1 - \beta)^{\bar m_{1}};
			\]
			\item[(d)] Update $\bar m_{1}$ from 
			\[ 
			\begin{split}
			\Law(\bar m_{1} \mid X_{1:n}, \sigma, \vartheta,\beta) 
			\propto\binom{m_1}{\bar m_{1}} \beta^{n-\bar m_{1}} (1-\beta)^{\bar m_{1}} \sigma^{k-\bar m_{1}}  \frac{\Gamma(\vartheta/\sigma +k-\bar m_{1})}{\Gamma(\vartheta+n-\bar m_{1})}
			\end{split}
			\]
		\end{itemize}
	\end{itemize}
	%	\hline
\end{algo}

We further notice that $\Law(\beta) \stackrel{d}{=} Beta(a,b)$ is a conjugate prior distribution for $\beta$. We assume $\Law(\sigma) \stackrel{d}{=} Beta(a_\sigma, b_\sigma)$ and $\Law(\vartheta) \stackrel{d}{=} Gamma(a_\vartheta, b_\vartheta)$. We made the steps (a) and (b) via Metropolis-Hastings with Gaussian proposal on a transformed scale $(\sigma, \vartheta)\to (\psi, \lambda)$, where $\psi = \log\left(\frac{\sigma}{1- \sigma}\right)$ and $\lambda = \log(\vartheta)$. The variances of the Gaussian proposals can be tuned to reach an optimal acceptance ratio for the Metropolis-Hastings steps \citep[see][]{Rob97}. We further initialize $\bar m_{1}$ by sampling uniformly on the integers from $0$ to $m_1$.

\subsection{Mixture case}\label{sup:mix}
Hereby we describe a sampling strategy to perform posterior inference with a contaminate Pitman-Yor mixture model. Let us denote by $\kernel(\cdot ; \xi)$ a kernel function with support $\R^d$, where $\xi \in \Xi$ denotes a generic set of parameters indexing the distribution of the kernel $\kernel :  \R^d \times \Xi \to \R^+$. We can then use the predictive distribution to specify a marginal sampling scheme, in the spirit of \citet{Esc88} and \citet{Esc95}. Let $Y_1, \dots, Y_n$ be a set of $\R^d$-valued random variables. We denote by $S_1, \dots, S_n$ the variables describing the latent group allocations in the mixture, more precisely we have $S_i = j$ if the $i$th observation belongs to the $j$th group of the mixture, with the proviso $S_i = 0$ if the $i$th observation comes from the diffuse component. For the sake of notational simplicity we define the vectors $Y_{1:n}:= (Y_1, \ldots , Y_n)$ and 
$S_{1:n} := (S_1, \ldots , S_n)$, moreover, for a generic vector $V_{1:n}$, we denote by $V_{(i)}$ the vector $V_{1:n}$ with the $i$th element removed.
Here we provide the algorithm to face posterior inference with the contaminated Pitman-Yor mixture model, by sampling $R$ realizations from an MCMC scheme.
\begin{algo}\label{algo:mixture}
	Sampling scheme for contaminated Pitman-Yor mixture model.
	%	\hline
	%\vspace*{-12pt}
	\begin{itemize}
		\item[(i)] Set initial values for $S_{1:n}$, $\sigma$, $\vartheta$, $\xi_{1:n}$, $\beta$;
		\item[(ii)] For $r=1, \dots, R$: 
		\begin{itemize}
			\item[(a)] For $i = 1, \dots, n$:
			\begin{itemize}
				\item[--] Update the cluster allocation of the $i$th element, where 
				\[ 
				\begin{split}
				P(S_i^{(r)} = j &\mid Y_i,  S_{(i)}, \xi_{(i)})
				\propto \begin{cases}
				\vspace{8pt} (1 - \beta) \int_\Xi \kernel(Y_i; \theta)P_0(\D \xi) &\quad \text{if } j = 0\\
				\vspace{8pt}
				\beta \frac{n_{j(i) } - \sigma}{\vartheta + n - \bar{m}_{1 (i)} - 1}\kernel(Y_i; \xi_{j(i)}^*) &\quad \text{if } j=1,\dots,k_{(i)}\\
				\beta \frac{\vartheta + (k_{(i)} - \bar{m}_{1 (i)})\sigma}{\vartheta + n - \bar{m}_{1 (i)} - 1} \int_\Xi\kernel(Y_i; \theta)Q_0(\D \theta) &\quad \text{if }  j = k_{(i)} + 1
				\end{cases}
				\end{split}
				\]
				where $\xi_{1(i)}^*, \dots, \xi_{k(i)}^*$ denote the unique values in $\xi_{(i)}$ with frequencies $n_{1(i) }, \dots, n_{k (i)}$,  $k_{(i)}$ represents the number of  distinct unique elements and $\bar{m}_{1,(i)} = \sum_{\ell=1, \ell \neq i}^n \indic_{[S_\ell = 0]}$;
				\item[--] Discard the empty clusters;
			\end{itemize}
		\end{itemize}
	\end{itemize}
	Let $\bar{m}_{1} =\sum_{i=1}^n \indic_{[S_i = 0]}$, $m_1 = \bar m_{1} + \sum_{j=1}^k\indic_{[\#\{i:S_i = S_j^*\} = 1]}$, and $n_1, \dots, n_k$ denote the frequencies of the unique values $\xi_1^*,\dots, \xi_k^*$ out of $\xi_1, \dots, \xi_n$, with $k = \max(S_1, \dots, S_n)$.
	\vspace{6pt}
	\begin{itemize}
		\item[(b)] For $j = 1, \dots, k$:
		\begin{itemize}
			\item[--] Update the $j$th unique value $\xi_j^*$ from 
			\[
			\Law(\xi_j^*\mid Y_{1:n}, S_{1:n}) \propto Q_0(\xi_j^*)\prod_{\{i: S_i = j\}}\kernel(Y_i; \xi_j^*);
			\]
		\end{itemize}
		\item[(c)] Update the parameters of the discrete component from 
		\[ 
		\begin{split}
		\qquad\Law(\sigma\mid S_{1:n}, \vartheta, \bar{m}_{1})  \propto \Law(\sigma)\sigma^{k - \bar{m}_{1}}\prod_{i=m_1}^k(1-\sigma)_{n_i - 1} \frac{\Gamma(\vartheta/\sigma + k - \bar{m}_{1})}{\Gamma(\vartheta / \sigma)},
		\end{split}
		\]
		and 
		\[ 
		\qquad\Law(\vartheta\mid S_{1:n},\sigma, \bar{m}_{1}) \propto \Law(\vartheta) \frac{\Gamma(\vartheta)\Gamma(\vartheta/\sigma + k - \bar{m}_{1})}{\Gamma(\vartheta / \sigma)\Gamma(\vartheta + n - \bar{m}_{1})};
		\]
		\item[(d)] Update $\beta$ from 
		\[ 
		\qquad\Law(\beta\mid S_{1:n},\bar{m}_{1}) \propto \Law(\beta)\beta^{n - \bar{m}_{1}} (1 - \beta)^{\bar{m}_{1}};
		\]
		%		\item[(e)] Update $\gamma$ from 
		%		\[ 
		%		\qquad\Law(\gamma\mid \xi_{1}^*, \dots, \xi_k^*) \propto \Law(\gamma)\prod_{j=1}^k Q_0(\xi_j^*;\gamma);
		%		\]
	\end{itemize}
	%	\hline
\end{algo}

The acceleration step (b) is not mandatory, but it improves the mixing performances of the algorithm. The integral in the predictive distribution of step (a) can be easily solved for suitable choices of the kernel function and the measures $P_0$ and $Q_0$, leading to a closed form expression for the predictive distribution. Otherwise such integral can be approximated via Monte Carlo methods, in the spirit of Algorithm 8 in \citet{Nea00}. In the simulation study and application we consider both $P_0$ and $Q_0$ distributed as Normal-Inverse-Wishart distributions, and a Gaussian kernel function. Such choices lead to a closed form of the predictive distribution, which corresponds to a multivariate Student's $t$ distribution.

% \begin{algo}
% A simple algorithm.
% %\vspace*{-12pt}
% \begin{tabbing}
%   \qquad \enspace Set $s=0$\\
%   \qquad \enspace For $i=1$ to $i=n$ \\
%   \qquad \qquad Set $t=0$\\
%   \qquad \qquad For $j=1$ to $j=i$ \\\
%   \qquad \qquad\qquad  $t \leftarrow t + x_{ij}$ \\
% \qquad \qquad $s \leftarrow s + t$ \\
% \qquad \enspace Output $s$
% \end{tabbing}
% \end{algo}

%\begin{algorithm}[!h]
%\caption{A simple algorithm.} \label{al1}
%\end{algorithm}

\clearpage
\section{Simulation studies}  \label{app_simulazioni}

\subsection{Simulation studies: discrete data}\label{sup:sim_disc}

We simulated a set of $n$ observations $X_1, \ldots , X_n$ generated from a possibly contaminated model, i.e., $X_i | \tilde{p} \simiid \tilde{p}$, where $\tilde{p}$ can be the contaminated Pitman-Yor process of Example \ref{ex:cPY1} or the Pitman-Yor process.
%a set of $n$ realizations by simulating $n \times \beta$ observations from the predictive distribution of a $\mathrm{PY}(\vartheta, \sigma, P_0)$, and then we augment the sample with $n \times (1 - \beta)$ singletons. 
We considered different specifications of the data generating process' parameters, by choosing: $n \in \{ 10\,000, 50\,000\}$, $\vartheta \in \{25, 100\}$, $\sigma \in \{0.2, 0.4, 0.6\}$ and $\beta \in \{ 0.9, 0.95, 1 \}$, where $\beta = 1$ corresponds to the Pitman-Yor process case.
We then estimated the parameters of the model when $\tilde{p}$ is a contaminated Pitman-Yor process and when $\tilde{p}$ is a Pitman-Yor process.  Posterior inference for the contaminated model is carried out through Algorithm \ref{algo:discrete}, and a similar algorithm has been used in the Pitman-Yor case. We selected the following prior distributions for the parameters, corresponding to vague specifications: $\vartheta \sim \mathrm{Gamma}(2, 0.02)$, $\sigma \sim \mathrm{Unif}(0,1)$ and $\beta \sim \mathrm{Unif}(0,1)$.
We tune the variances of the Metropolis-Hastings steps for the update of $\vartheta$ and $\sigma$ in order to achieve optimal acceptance rates \citep{Rob97}. The MCMC method is based on $15\,000$ iterations, including a burn-in period of $5\,000$ iterations, moreover the values are thinned each $10$ iterations, obtaining a final sample of size $1\,000$. The convergence of the chains was assessed by randomly checking some replications, without any evidence against it.

%We compare the contaminated Pitman-Yor process model with the usual Pitman-Yor model. Posterior inference for the contaminated model is made by using algorithm \ref{algo:discrete} described in Section \ref{sup:disc} of the Appendix, and similarly for the Pitman-Yor model. We further assume a priori $\vartheta \sim \mathrm{Gamma}(2, 0.02)$, $\sigma \sim \mathrm{Unif}(0,1)$ and $\beta \sim \mathrm{Unif}(0,1)$. We tune the variances of the Metropolis-Hastings steps for the update of $\vartheta$ and $\sigma$ in order to attain optimal acceptance rates \citep{Rob97}. We sampled each chain for $15\,000$ iterations, of which $5\,000$ burn-in iterations, thinning the values each $10$ iteration, obtaining a final sample of size $1\,000$. The convergence of the chains was assessed by randomly checking some of the replications, without any evidence against it.

\begin{table}[!h]
	\centering
	\begin{tabular}{lllll rrr}
		&      &       &       & $\sigma$ & \multicolumn{1}{c}{0.2} & \multicolumn{1}{c}{0.4} & \multicolumn{1}{c}{0.6} \\ 
		& $\beta$ & $n$     & $\theta$ &       & \multicolumn{1}{l}{     } & \multicolumn{1}{l}{     } & \multicolumn{1}{l}{     } \\ 
		CPY&0.9&10000&25 & & $0.14 \,(0.08)$ & $0.37 \,(0.05)$ & $0.58 \,(0.04)$\\
		&&&100 & & $0.17 \,(0.06)$ & $0.39 \,(0.04)$ & $0.59 \,(0.04)$\\
		&&50000&25 & & $0.18 \,(0.04)$ & $0.39 \,(0.02)$ & $0.6 \,(0.02)$\\
		&&&100 & & $0.19 \,(0.03)$ & $0.4 \,(0.02)$ & $0.6 \,(0.02)$\\
		&0.95&10000&25 & & $0.14 \,(0.08)$ & $0.38 \,(0.05)$ & $0.59 \,(0.04)$\\
		&&&100 & & $0.17 \,(0.07)$ & $0.39 \,(0.05)$ & $0.59 \,(0.03)$\\
		&&50000&25 & & $0.19 \,(0.03)$ & $0.39 \,(0.03)$ & $0.6 \,(0.02)$\\
		&&&100 & & $0.19 \,(0.03)$ & $0.4 \,(0.02)$ & $0.6 \,(0.02)$\\
		&1&10000&25 & & $0.1 \,(0.07)$ & $0.35 \,(0.03)$ & $0.56 \,(0.02)$\\
		&&&100 & & $0.13 \,(0.05)$ & $0.36 \,(0.03)$ & $0.56 \,(0.02)$\\
		&&50000&25 & & $0.16 \,(0.03)$ & $0.37 \,(0.02)$ & $0.58 \,(0.01)$\\
		&&&100 & & $0.18 \,(0.02)$ & $0.39 \,(0.01)$ & $0.59 \,(0.01)$\\
		PY&0.9&10000&25 & & $0.81 \,(0.02)$ & $0.76 \,(0.01)$ & $0.76 \,(0.01)$\\
		&&&100 & & $0.67 \,(0.01)$ & $0.67 \,(0.01)$ & $0.72 \,(0.01)$\\
		&&50000&25 & & $0.93 \,(0.01)$ & $0.87 \,(0.01)$ & $0.82 \,(0.01)$\\
		&&&100 & & $0.83 \,(0.01)$ & $0.78 \,(0.01)$ & $0.77 \,(0.01)$\\
		&0.95&10000&25 & & $0.68 \,(0.02)$ & $0.65 \,(0.02)$ & $0.7 \,(0.01)$\\
		&&&100 & & $0.53 \,(0.01)$ & $0.57 \,(0.01)$ & $0.67 \,(0.01)$\\
		&&50000&25 & & $0.86 \,(0.01)$ & $0.78 \,(0.01)$ & $0.75 \,(0.01)$\\
		&&&100 & & $0.71 \,(0.01)$ & $0.67 \,(0.01)$ & $0.71 \,(0.01)$\\
		&1&10000&25 & & $0.18 \,(0.05)$ & $0.39 \,(0.02)$ & $0.6 \,(0.01)$\\
		&&&100 & & $0.19 \,(0.03)$ & $0.4 \,(0.02)$ & $0.6 \,(0.01)$\\
		&&50000&25 & & $0.19 \,(0.02)$ & $0.4 \,(0.01)$ & $0.6 \,(0.01)$\\
		&&&100 & & $0.2 \,(0.02)$ & $0.4 \,(0.01)$ & $0.6 \,(0.01)$\\
	\end{tabular}
	\caption{Posterior mean (and standard deviation) estimates of the parameter $\sigma$ for the contaminated Pitman-Yor process (CPY) and the Pitman-Yor process (PY) in different simulated scenarios. All the estimates are averaged over $100$ replications.}
	\label{tab:sigma}
	%	\begin{tabnote}
	%		CPY, contaminated Pitman-Yor process; PY, Pitman-Yor process.
	%	\end{tabnote}
\end{table}

\begin{table}[ht]
	\centering
	\begin{tabular}{llll rrr}
		&       &       & $\sigma$ & \multicolumn{1}{c}{    0.2} & \multicolumn{1}{c}{    0.4} & \multicolumn{1}{c}{    0.6} \\ 
		$\beta$ & $n$     & $\theta$ &       & \multicolumn{1}{l}{       } & \multicolumn{1}{l}{       } & \multicolumn{1}{l}{       } \\ 
		0.9&10000&25 & & $1008.41 \,(19.7)$ & $1012.78 \,(39.6)$ & $1021.15 \,(145.12)$\\
		&&100 & & $1009.15 \,(39.88)$ & $1003.96 \,(74.53)$ & $1003.74 \,(207.73)$\\
		&50000&25 & & $5005.76 \,(18.36)$ & $5007.45 \,(53.24)$ & $4981.76 \,(192.49)$\\
		&&100 & & $5003.09 \,(35.41)$ & $5002.53 \,(97.65)$ & $5011.61 \,(279.41)$\\
		0.95&10000&25 & & $509.3 \,(21.48)$ & $503.69 \,(45.43)$ & $515.07 \,(132.24)$\\
		&&100 & & $515.68 \,(44.18)$ & $499.72 \,(85.18)$ & $537.85 \,(165.13)$\\
		&50000&25 & & $2503.8 \,(18.79)$ & $2512.59 \,(49.75)$ & $2494.11 \,(184.61)$\\
		&&100 & & $2506.53 \,(34.71)$ & $2494.05 \,(87.37)$ & $2448.66 \,(304.99)$\\
		1&10000&25 & & $20.42 \,(9.91)$ & $37.42 \,(15.72)$ & $127.84 \,(55.58)$\\
		&&100 & & $37.98 \,(20.48)$ & $70.5 \,(33.63)$ & $181.82 \,(80.47)$\\
		&50000&25 & & $18.59 \,(9.13)$ & $51.24 \,(24.23)$ & $173.53 \,(89.17)$\\
		&&100 & & $30.69 \,(13.81)$ & $74.85 \,(37.6)$ & $226.81 \,(108.3)$\\
	\end{tabular}
	\caption{Summaries of the posterior mean (and standard deviation) estimates of $\bar m_1$, for different scenarios, averaged over $100$ replications, for the contaminated Pitman-Yor process. The true value of $\bar m_1$ equals $n (1 - \beta)$.}
	\label{tab:k1}
\end{table}

\begin{table}[h!]
	\centering
	\begin{tabular}{lllll rrr}
		&      &       &       & $\sigma$ & \multicolumn{1}{c}{    0.2} & \multicolumn{1}{c}{    0.4} & \multicolumn{1}{c}{    0.6} \\ 
		& $\beta$ & $n$     & $\theta$ &       & \multicolumn{1}{l}{       } & \multicolumn{1}{l}{       } & \multicolumn{1}{l}{       } \\ 
		CPY&0.9&10000&25 & & $32.25 \,(9.34)$ & $31.21 \,(7.85)$ & $32.78 \,(10.87)$\\
		&&&100 & & $108.68 \,(21.27)$ & $103.46 \,(17.3)$ & $106.82 \,(22.68)$\\
		&&50000&25 & & $28.97 \,(6.79)$ & $27.77 \,(6.33)$ & $26.96 \,(7.47)$\\
		&&&100 & & $104.03 \,(13.6)$ & $99.93 \,(14.91)$ & $104.51 \,(16.9)$\\
		&0.95&10000&25 & & $32.7 \,(9.91)$ & $29.49 \,(7.63)$ & $30.54 \,(9.84)$\\
		&&&100 & & $110.05 \,(21.7)$ & $105.91 \,(18.67)$ & $103.98 \,(20.47)$\\
		&&50000&25 & & $27.72 \,(4.99)$ & $28.58 \,(5.87)$ & $28.44 \,(7.88)$\\
		&&&100 & & $104.48 \,(12.74)$ & $99.79 \,(13.41)$ & $98.51 \,(17.77)$\\
		&1&10000&25 & & $35.04 \,(8.27)$ & $31.53 \,(6.85)$ & $35.02 \,(9.12)$\\
		&&&100 & & $118.04 \,(19.56)$ & $111.85 \,(15.02)$ & $112.76 \,(18.86)$\\
		&&50000&25 & & $30.76 \,(5.13)$ & $30.22 \,(5.41)$ & $30.24 \,(8.06)$\\
		&&&100 & & $107.15 \,(10.76)$ & $108.57 \,(14.75)$ & $110.23 \,(16.39)$\\
		PY&0.9&10000&25 & & $2.31 \,(0.36)$ & $5.01 \,(1.12)$ & $13.49 \,(4.59)$\\
		&&&100 & & $18.05 \,(3.47)$ & $36.71 \,(7.11)$ & $65.57 \,(14.23)$\\
		&&50000&25 & & $1.18 \,(0.12)$ & $1.79 \,(0.21)$ & $4.89 \,(1.2)$\\
		&&&100 & & $2.56 \,(0.27)$ & $7.03 \,(1.29)$ & $29.52 \,(6)$\\
		&0.95&10000&25 & & $3.54 \,(0.83)$ & $8.45 \,(2.22)$ & $17.38 \,(5.33)$\\
		&&&100 & & $35.7 \,(5.28)$ & $58.51 \,(9.35)$ & $77.34 \,(14.46)$\\
		&&50000&25 & & $1.12 \,(0.13)$ & $2.38 \,(0.41)$ & $8.52 \,(2.58)$\\
		&&&100 & & $5.01 \,(0.81)$ & $18.17 \,(3.16)$ & $48.4 \,(9.29)$\\
		&1&10000&25 & & $28.46 \,(6.41)$ & $27.48 \,(6.14)$ & $29.46 \,(7.54)$\\
		&&&100 & & $103.97 \,(13.86)$ & $100.86 \,(13.1)$ & $100.41 \,(16.63)$\\
		&&50000&25 & & $27.67 \,(4.4)$ & $26.95 \,(4.85)$ & $26.72 \,(6.91)$\\
		&&&100 & & $101.03 \,(10.35)$ & $101.85 \,(13.31)$ & $102.23 \,(14.61)$\\
	\end{tabular}
	\caption{Summaries of the posterior mean (and standard deviation) estimates of $\vartheta$, averaged for $100$ of replications, for the contaminated Pitman-Yor process (CPY) and the Pitman-Yor process (PY).}
	\label{tab:theta}
\end{table}

\begin{table}[ht]
	\centering
	\begin{tabular}{llll rrr}
		&       &       & $\sigma$ & \multicolumn{1}{c}{  0.2} & \multicolumn{1}{c}{  0.4} & \multicolumn{1}{c}{  0.6} \\ 
		$\beta$ & $n$     & $\theta$ &       & \multicolumn{1}{l}{     } & \multicolumn{1}{l}{     } & \multicolumn{1}{l}{     } \\ 
		0.9&10000&25 & & $0.899 \,(0.002)$ & $0.899 \,(0.004)$ & $0.898 \,(0.015)$\\
		&&100 & & $0.899 \,(0.004)$ & $0.9 \,(0.007)$ & $0.9 \,(0.021)$\\
		&50000&25 & & $0.9 \,(0)$ & $0.9 \,(0.001)$ & $0.9 \,(0.004)$\\
		&&100 & & $0.9 \,(0.001)$ & $0.9 \,(0.002)$ & $0.9 \,(0.006)$\\
		0.95&10000&25 & & $0.949 \,(0.002)$ & $0.95 \,(0.005)$ & $0.948 \,(0.013)$\\
		&&100 & & $0.948 \,(0.004)$ & $0.95 \,(0.009)$ & $0.946 \,(0.017)$\\
		&50000&25 & & $0.95 \,(0)$ & $0.95 \,(0.001)$ & $0.95 \,(0.004)$\\
		&&100 & & $0.95 \,(0.001)$ & $0.95 \,(0.002)$ & $0.951 \,(0.006)$\\
		1&10000&25 & & $0.998 \,(0.001)$ & $0.996 \,(0.002)$ & $0.987 \,(0.006)$\\
		&&100 & & $0.996 \,(0.002)$ & $0.993 \,(0.003)$ & $0.982 \,(0.008)$\\
		&50000&25 & & $1 \,(0)$ & $0.999 \,(0)$ & $0.997 \,(0.002)$\\
		&&100 & & $0.999 \,(0)$ & $0.998 \,(0.001)$ & $0.995 \,(0.002)$\\
	\end{tabular}
	\caption{Summaries of the posterior mean (and standard deviation) estimates of $\beta$, for different scenarios, averaged over $100$ replications, for the contaminated Pitman-Yor process.}
	\label{tab:beta}
\end{table}

Table \ref{tab:sigma} shows the summaries of the posterior inference on $\sigma$ for the different scenarios, corresponding to different parametrizations of the data generating process.  All the estimates are averaged over $100$ replications. %choices of the parameters to generate the data. All the experiments are averaged over $100$ replications. 
When the simulated data are contaminated, the Pitman-Yor process leads to a heavy misleading posterior inference for the discount parameter, due to the presence of a large number of singletons. On the counterpart, the contaminated model is slightly underestimating the discount parameter in absence of contaminant observations. We further remark that when both  $\vartheta$ and $\sigma$ are small, the value of $\sigma$ is slightly underestimated in the case of contaminated Pitman-Yor model. %Overall the estimation shows reasonable outcome.

%Additional results have been reported in Section \ref{sup:sim_disc} of the Appendix, showing the posterior summaries for $\vartheta$, $\beta$ and the number of structural singletons. We remark that the latter quantity is crucial in several applied fields, such as disclosure risk assessment and language modeling. See Section \ref{sec:discussion} for further details on these applications. We further investigated via simulation study the use of contaminated priors in the mixture model case (see Section \ref{sec:cont_simu}).
Similarly, Table \ref{tab:theta} shows the posterior means of $\vartheta$ for both the models, while Table \ref{tab:beta} and Table \ref{tab:k1} shows the posterior means of $\beta$ and the posterior estimates of the number of structural singletons respectively, for the contaminated Pitman-Yor model. We remark that estimating the number of structural singletons is crucial in several applied fields, such as disclosure risk assessment and language modeling. See Section \ref{sec:discussion} for further details on these applications.

\subsection{Simulation study with continuous data}\label{sec:cont_simu}

We now move to the mixture scenario, by considering a set of continuous data taking values in $\R^d$. We simulate a set of data from a mixture of two Gaussian distributions with density $f(y) = 0.5 \phi_d(y;-3, \mathrm{diag}_d(1)) + 0.5 \phi_d(y;3,\mathrm{diag}_d(1))$, where $\phi_d(\cdot; a, B)$ denotes the density of a $d$-dimensional Gaussian random variable with mean vector $a$ and covariance matrix $B$, and $\mathrm{diag}_d(b)$ denotes a diagonal matrix of dimension $d \times d$ with diagonal elements equal to $b$. After sampling $m$ observations $Y_1, \dots, Y_m$, we then augment the sample with additional $s$ outliers $cY_{m+1}, \dots, cY_{m+s}$ from an over-disperse truncated Gaussian distribution 
\[
\phi(y; 0, \mathrm{diag}_d(3^2))\mathbb \indic_{[(-\infty, -3\sqrt{\chi_d^2(0.9)})\cup(3\sqrt{\chi_d^2(0.9)}, \infty)]}\left(\lvert\lvert y \rvert\rvert_2^2\right)
\]
i.e. a multivariate Gaussian distribution with support $\R^d$ minus the $d$-dimensional sphere of radius $3\sqrt{\chi_{d}^2(0.9)}$ centered at the origin, where $\chi_d^2(0.9)$ denotes the quantile of order $0.9$ of a Chi-square distribution with $d$ degrees of freedom, obtaining a sample of size $n=m+s$; the parameter  $c$  has the role to shrink or expand the nuisance observations towards the origin. We consider different simulated scenarios by selecting: $s = 10$, $m \in \{90, 240\}$, $d \in \{2, 4\}$ and $c \in \{1, 1.25\}$. Figure \ref{fig:ex_mix_simu} shows an example of simulate dataset for different values of the scaling constant $c$. 

\begin{figure}[!h]
	\centering
	\includegraphics[width = 0.9 \textwidth]{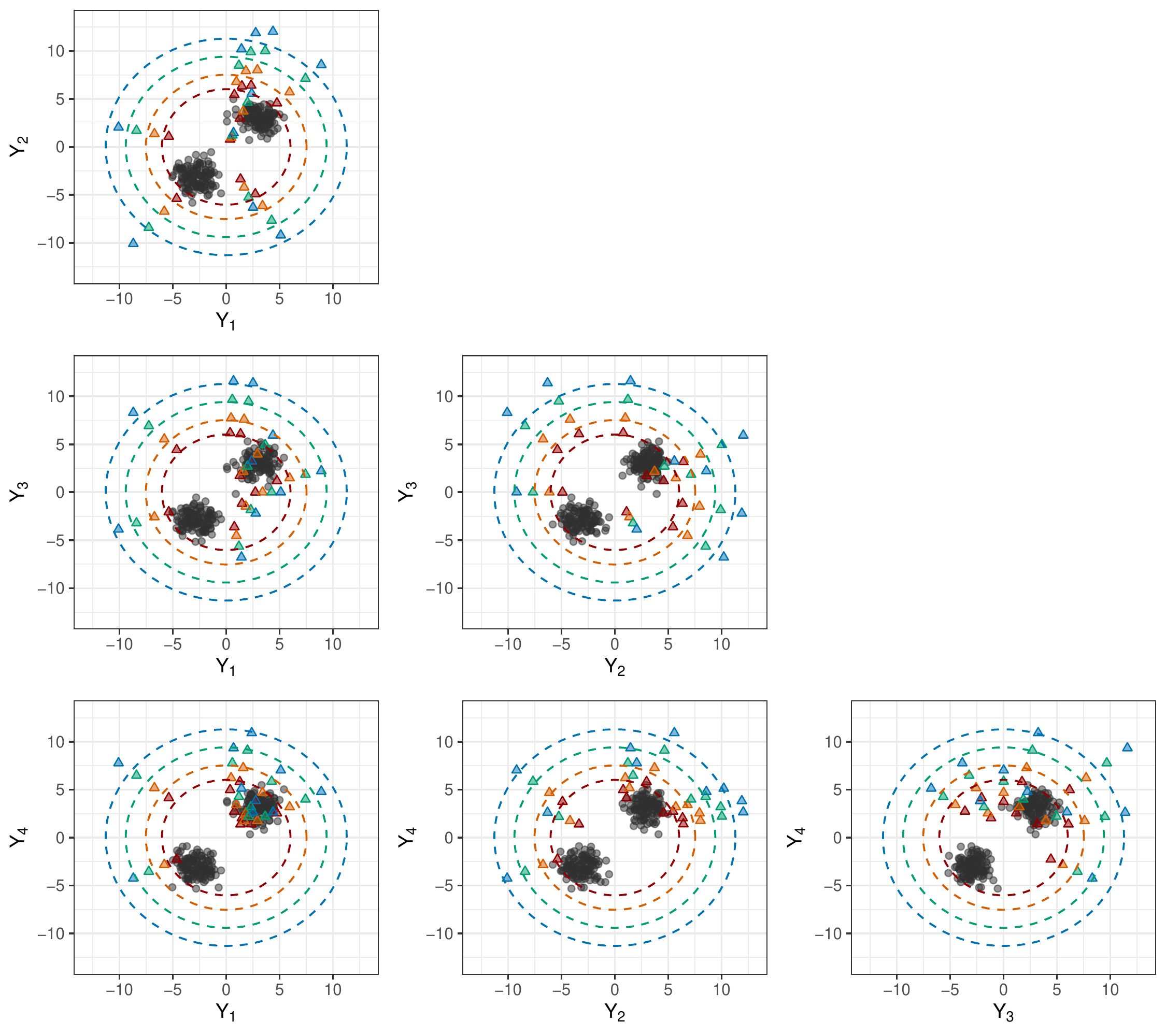}
	\caption{An example of simulated data set with dimension $d = 4$. Different colors correspond to different values of the scaling constant $c$, the dashed lines correspond to the hypersphere of radius equal to the square root of $3^2\times\chi_4^2(0.9)$, the $0.9$ quantile of a Chi-square distribution with $4$ degrees of freedom, multiplied for different scaling constants.}\label{fig:ex_mix_simu}
\end{figure}

We face posterior inference with the mixture model described in Section \ref{sec:mixture}, letting $\tilde{p}$ distributed as a contaminated Pitman-Yor process. We assume $P_0$ and $Q_0$ distributed as Normal-Inverse-Wishart distributions, with $Q_0 \sim \mathrm{NIW}(\mu_0, \kappa_0, \nu_0, S_0)$ and $P_0 \sim \mathrm{NIW}(\mu_1, \kappa_1, \nu_1, S_1)$. We set $\mu_0 = 0$, $\nu_0 = d + 3$, $\mu_1 = 0$, $\nu_1 = d + 3$, and $S_0 = S_1$ matching the diagonal of the sample variance of the data across all the scenarios. We consider two distinct specifications of the contaminated Pitman-Yor mixture model: i) CPY$_1$, here  the base measure is different from the contaminant measure, and we set $\kappa_0 = 1$, $\kappa_1 = 0.25$, so that we allow an over-disperse contaminant measure $P_0$; ii) CPY$_2$, a scenario where the base measure and the diffuse component coincide, we assume $\kappa_0 = \kappa_1 = 0.5$.  We further consider the Pitman-Yor mixture model (PY) with base measure $Q_0 \sim \mathrm{NIW}(\mu_0, \kappa_0, \nu_0, S_0)$  and $\mu_0 = 0$, $\kappa_0 = 1$, $\nu_0 = d+3$, and $S_0$ equals the diagonal of the sample variance of the data. We finally consider the following vague prior specifications for the parameters: $\vartheta \sim \mathrm{Gamma}(2, 0.02)$, $\sigma \sim \mathrm{Unif}(0,1)$, and for the contaminated models  $\beta \sim \mathrm{Unif}(0,1)$. 

We applied the algorithm described in Section \ref{sup:mix} to provide posterior estimates of the number of outliers out of the sample. In order to do this we have  identified  the optimal posterior point estimate of the latent partition of the data, which has been obtained minimizing the variation of information loss function  among the sampled partitions \citep[see][]{Wad18,Ras18}. 
%We tuned the variances of the Metropolis-Hastings steps for the update of $\vartheta$ and $\sigma$ in order to achieve optimal acceptance rates.
The MCMC procedure is based on $15\,000$ iterations, including a burn-in period of  $5\,000$ iterations. The convergence of the chains was assessed by randomly checking some of the replications, without any evidence against it. Table \ref{tab:sing} reports the number of singletons in the posterior point estimate of the latent partition of the data, while the true number of outliers equals $10$.

\begin{table}[!h]
	\centering
	\begin{tabular}{lllrrr}
		&          &      & \multicolumn{1}{c}{ $c =  1$} & \multicolumn{1}{c}{ $c =  1.25$} & \multicolumn{1}{c}{ $c =  1.5$} \\ 
		$d$ & $n$     & model &  \multicolumn{1}{l}{         } & \multicolumn{1}{l}{       } & \multicolumn{1}{l}{          } \\ 
		$2$  & $100$ & CPY$_1$ & $9.10\,(1.63)$ & $9.78\,(1.34)$ & $9.63\,(1.57)$  \\ 
		&       & CPY$_2$ & $4.58\,(3.54)$ & $3.73\,(3.81)$ & $2.50\,(3.21)$ \\ 
		&       & PY      & $1.10\,(9.99)$ & $0.09\,(0.35)$ & $0.05\,(0.29)$ \\ 
		& $250$ & CPY$_1$  & $9.67\,(6.53)$ & $8.87\,(2.59)$ & $8.36\,(3.19)$ \\ 
		&       & CPY$_2$ & $8.16\,(10.94)$ & $6.88\,(10.29)$ & $10.32\,(17.19)$ \\ 
		&       & PY      & $0.06\,(0.23)$ & $0.03\,(0.17)$ & $0.02\,(0.14)$ \\ 
		$4$  & $100$ & CPY$_1$ & $10.68\,(4.62)$ & $10.41\,(2.12)$ & $10.89\,(3.31)$ \\ 
		&       & CPY$_2$ & $8.46\,(3.40)$ & $9.50\,(4.99)$ & $9.09\,(4.41)$ \\ 
		&       & PY      & $0.66\,(1.87)$ & $0.54\,(1.67)$ & $0.62\,(1.67)$ \\ 
		& $250$ & CPY$_1$ & $14.26\,(9.15)$ & $13.15\,(5.71)$ & $13.07\,(7.91)$ \\ 
		&       & CPY$_2$ & $17.36\,(14.24)$ & $15.55\,(12.93)$ & $16.65\,(14.25)$ \\ 
		&       & PY      & $0.36\,(1.31)$ & $0.55\,(1.59)$ & $0.70\,(1.94)$
	\end{tabular}
	\caption{Posterior mean (and standard deviation) of the number of singletons, averaged over $100$ replications, detected with the two specifications of the contaminated Pitman-Yor mixture model and the Pitman-Yor mixture model.}
	\label{tab:sing}
	%	\begin{tabnote}
	%		CPY$_1$, contaminated Pitman--Yor mixture model with $P_0 \neq Q_0$ and no hyperpriors; CPY$_2$, contaminated Pitman--Yor mixture model with $P_0 \neq Q_0$ and hyperpriors; CPY$_3$, contaminated Pitman--Yor mixture model with $P_0 = Q_0$; PY, Pitman--Yor mixture model.
	%	\end{tabnote}
\end{table}

We can appreciate how the Pitman-Yor mixture model lack in flexibility to estimate the number of outliers in a set of data. Among the contaminated models, the model with $P_0 \neq Q_0$, denote by CPY$_1$, shows overall estimates closer to the true number of contaminants compared to the case with $P_0 = Q_0$, with also a smaller uncertainty over different replications.

\clearpage
\section{North America Ranidae dataset}\label{sup:disc_app}

To analyze the North America Ranidae dataset, we applied  Algorithm \ref{algo:discrete} described in Section \ref{sup:disc}. We have tuned the variances of the Metropolis-Hastings steps in order to attain optimal acceptance rates. We produced a raw chain of $35\,000$ iterations, which includes $10\,000$ burn-in iterations, and we thinned the chain every $10$ realizations, obtaining a final sample of size $2\,500$. We report here the posterior summaries and the traceplots for the main parameters of the contaminated Pitman-Yor model and the Pitman-Yor model estimated with North America Ranidae data. 

\begin{figure}[h]
	\centering
	\includegraphics[width = 0.79 \textwidth]{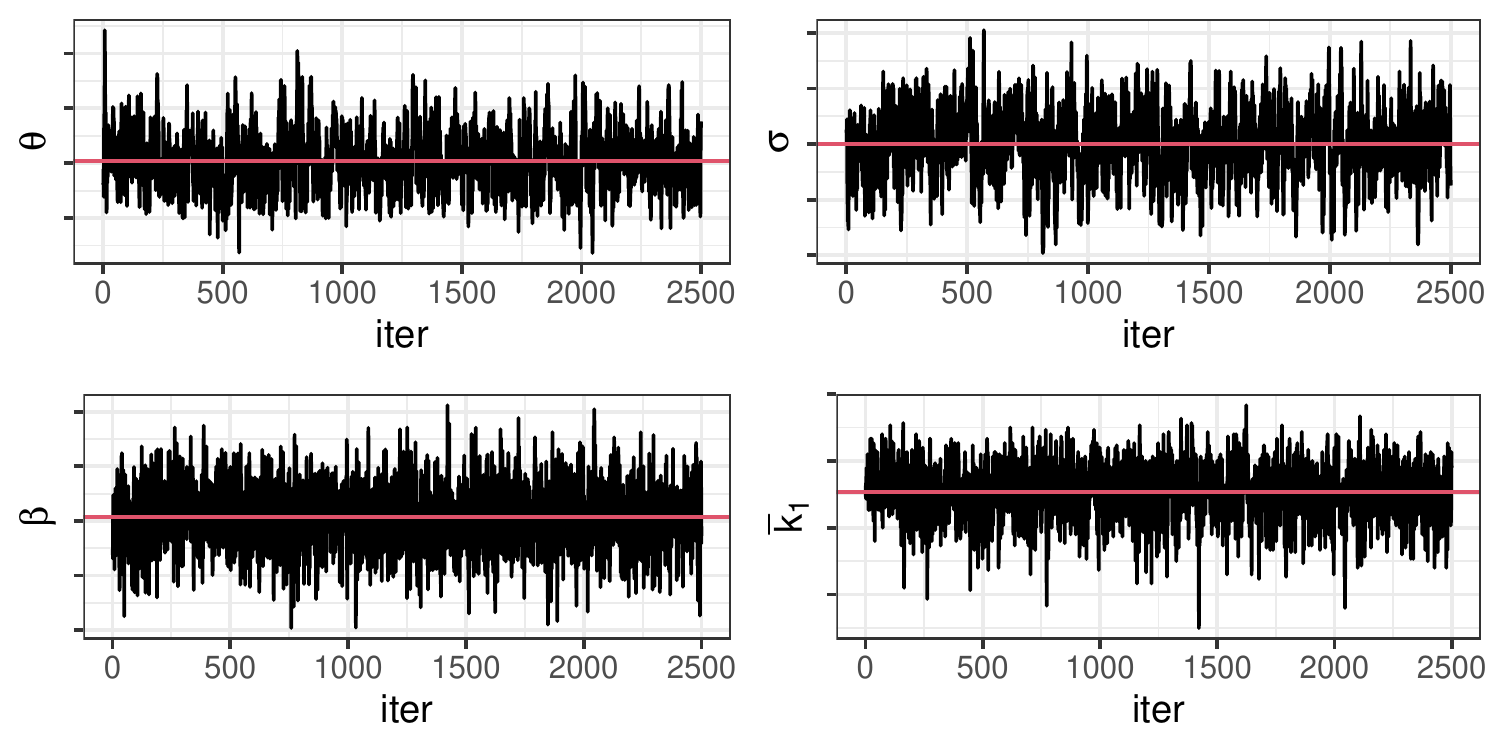}
	\caption{Traceplots for the contaminated Pitman-Yor model. Top-left panel: traceplot of $\vartheta$. Top-right panel: traceplot of $\sigma$. Bottom-left panel: traceplot of $\beta$. Bottom-right panel: traceplot of $\bar m_1$.}\label{fig:trace_PY}
\end{figure}
\begin{figure}[h]
	\centering
	\includegraphics[width = 0.79 \textwidth]{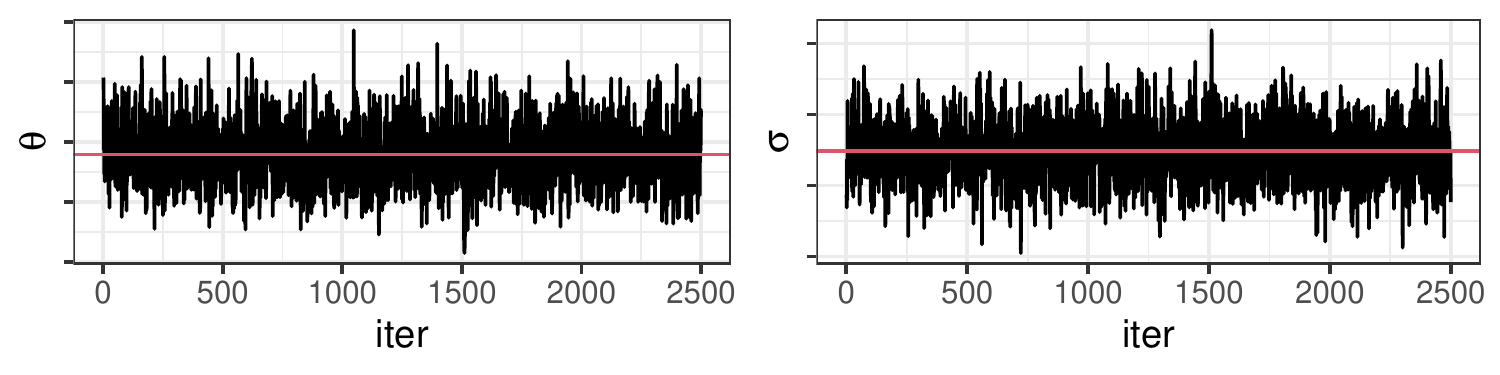}
	\caption{Traceplots for the Pitman-Yor model. Left panel: traceplot of $\vartheta$. Right panel: traceplot of $\sigma$.}\label{fig:trace_cPY}
\end{figure}

\begin{table}[h]
	\centering
	\begin{tabular}{lllll}
		\\
		& mean & SD & ESS      & Geweke diagnostic \\
		$\vartheta$ & 30.466         & 5.172          & 229.066  & 0.642             \\
		$\sigma$    & 0.099          & 0.031          & 254.958  & -0.612            \\
		$\beta$     & 0.998          & 0.00015        & 1186.791 & -0.520            \\
		$\bar m_1$  & 226.196        & 12.373         & 524.699  & 1.115            
	\end{tabular}
	\caption{Posterior summaries of the chain sampled from the posterior distribution of the contaminated Pitman-Yor model with the north America ranidae dataset.}
	\label{tab:post_cPY}
\end{table}

\begin{table}[h]
	\centering
	\begin{tabular}{lllll}
		& mean & SD & ESS      & Geweke diagnostic \\
		$\vartheta$ & 9.112          & 2.561          & 534.603  & 0.328             \\
		$\sigma$    & 0.339          & 0.0175         & 1141.642 & 0.872            
	\end{tabular}
	\caption{Posterior summaries of the chain sampled from the posterior distribution of the Pitman-Yor model with the north America ranidae dataset.}
	\label{tab:post_PY}
\end{table}
%
%\begin{table}[!h]
%	\centering
%	\caption{Posterior summaries of the chain sampled for the contaminated Pitman-Yor model.}\label{tab:post_cPY}
%	\begin{tabular}{lllll}
%		& mean & standard deviation & ESS      & Geweke diagnostic \\
%		$\vartheta$ & 30.466         & 5.172          & 229.066  & 0.642             \\
%		$\sigma$    & 0.099          & 0.031          & 254.958  & -0.612            \\
%		$\beta$     & 0.998          & 0.00015        & 1186.791 & -0.520            \\
%		$\bar m_1$  & 226.196        & 12.373         & 524.699  & 1.115            
%	\end{tabular}
%	\caption{Posterior summaries of the chain sampled for the Pitman-Yor model.}\label{tab:post_PY}
%	\begin{tabular}{lllll}
%		& posterior mean & posterior s.d. & ESS      & Geweke diagnostic \\
%		$\vartheta$ & 9.112          & 2.561          & 534.603  & 0.328             \\
%		$\sigma$    & 0.339          & 0.0175         & 1141.642 & 0.872            
%	\end{tabular}
%\end{table}

The traceplots reported in figure \ref{fig:trace_cPY} and \ref{fig:trace_PY} show a good mixing of the chains, with random spikes and without systematic behavior. Tables \ref{tab:post_cPY} and \ref{tab:post_PY} shows posterior summaries of the produced chains.

\begin{figure}[!h]
	\centering
	\includegraphics[width = 0.4 \textwidth]{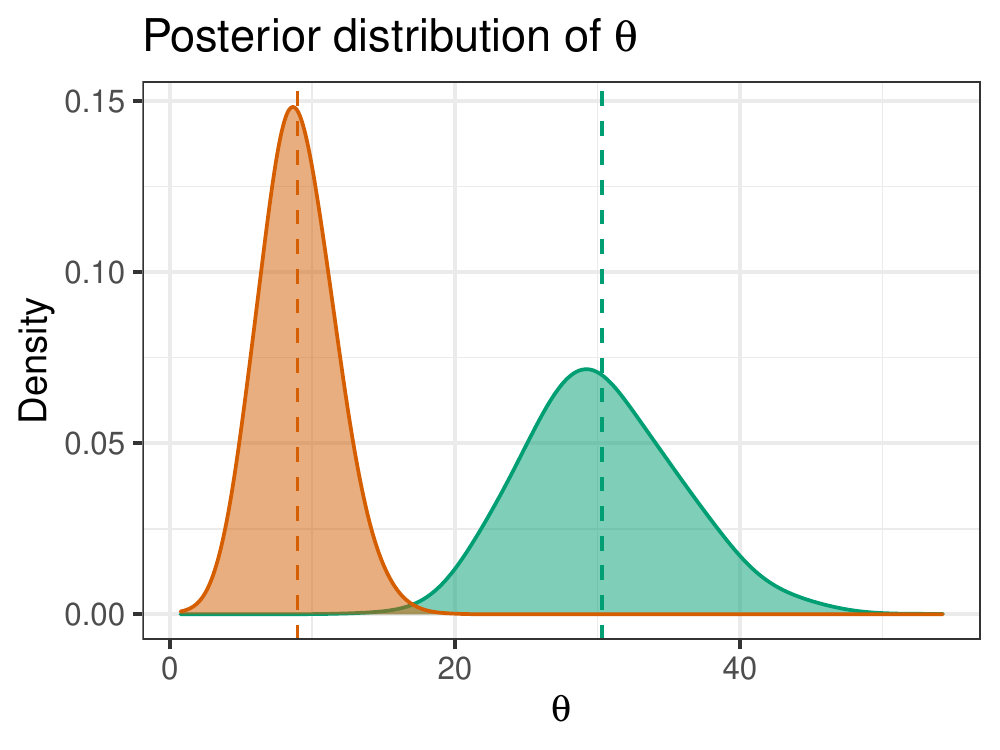}
	\caption{Posterior distributions of the parameter $\vartheta$ for the contaminated Pitman-Yor model (green) and the Pitman-Yor model (orange).}\label{fig:post_plot3}
\end{figure}

\begin{figure}[!h]
	\centering
	\includegraphics[width = 0.8 \textwidth]{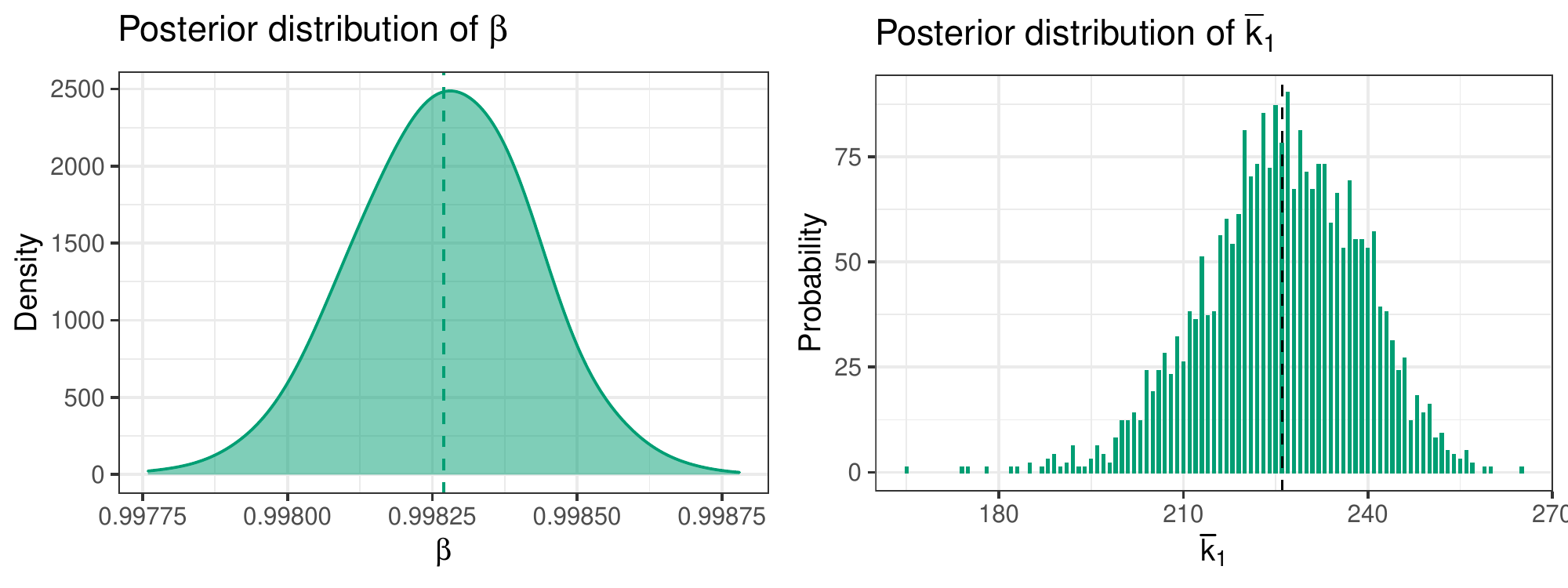}
	\caption{Traceplots for the Pitman-Yor model. Left panel: traceplot of $\vartheta$. Right panel: traceplot of $\sigma$.}\label{fig:post_plot4}
\end{figure}

\subsection{Inference on an additional sample}\label{sec:inf_new_samp}

We validated the capability of the model of capturing an inflation of the number of species with frequency one in an additional sample. To this aim, we estimate the distribution of $\E [N_{m,1}^{(n-m)}| X_1, \ldots , X_{n - m}]$, $\E [N_{m,2}^{(n-m)}| X_1, \ldots ,\allowbreak X_{n - m}]$ and $\E [K_{m}^{(n-m)}| X_1, \ldots , X_{n - m}]$ via cross-validation, by sampling without replacement the $80\%$ of the observations  $X_1^{\mathrm{train}},\dots,X_{n-m}^{\mathrm{train}}$ to estimate the model, and then predicting the number of new species with frequency one in an additional sample, composed by the remaining $20\%$ of the observations, $X_{n-m+1}^{\mathrm{test}},\dots,X_n^{\mathrm{test}}$. We estimate both the contaminated Pitman-Yor model and the Pitman-Yor model with $X_1^{\mathrm{train}},\dots,X_{n-m}^{\mathrm{train}}$, assuming the same model specification of Section \ref{sec:discrete_app}, and by running the algorithm described in Section \ref{sup:disc} for $10\,000$ iterations, of which $5\,000$ burn-in iterations, thinning the sampled chains every $5$ realizations. We then evaluate $\E [N_{m,1}^{(n-m)}| X_1, \ldots , X_{n - m}]$, $\E [N_{m,2}^{(n-m)}| X_1, \ldots , X_{n - m}]$ and $\E [K_{m}^{(n-m)}| X_1, \ldots , X_{n - m}]$ for both the models. We further compute, for each sampled train and test data set, the observed number of new species with frequency equal to one, the observed number of new species with frequency equal to two, and the observed number of new species in the additional sample. We replicated the cross-validation for $1\,000$ times. 

\begin{figure}[h]
	\centering
	\includegraphics[width = 0.49 \textwidth]{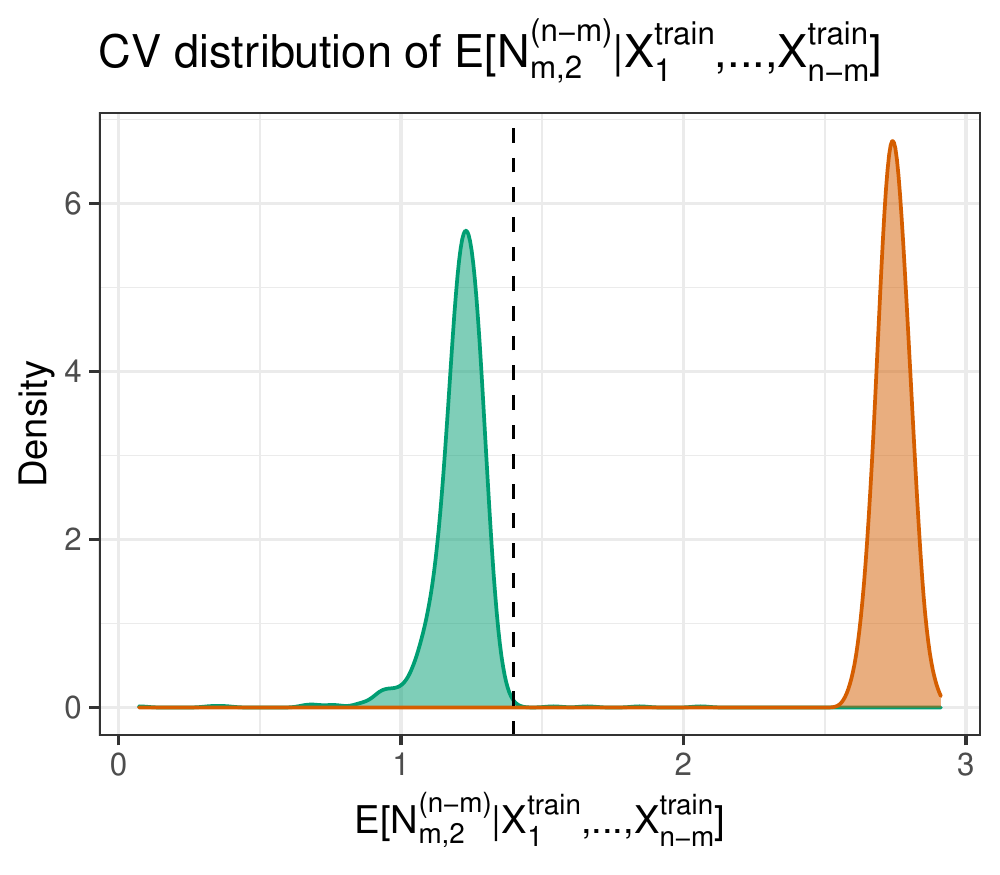}
	\caption{Cross-validated distribution of $\E [N_{m,2}^{(n-m)}| X_1^{\mathrm{train}}, \ldots , X_{n-m}^{\mathrm{train}}]$, for the contaminated Pitman-Yor model (green) and the Pitman-Yor model (orange). The black dashed line denotes the average of the observed number of species with frequency equal to two in the additional sample.}\label{fig:add_sample2}
\end{figure}

Figure \ref{fig:add_sample} of the article shows the resulting distributions of $\E [N_{m,1}^{(n-m)}| X_1^{\mathrm{train}}, \ldots , X_{n-m}^{\mathrm{train}}]$ and $\E [K_{m}^{(n-m)}| X_1^{\mathrm{train}},\allowbreak \ldots , X_{n-m}^{\mathrm{train}}]$, while Figure \ref{fig:add_sample2} show the resulting distributions $\E [N_{m,2}^{(n-m)}| X_1^{\mathrm{train}}, \ldots , X_{n-m}^{\mathrm{train}}]$, for both the contaminated Pitman-Yor model and the Pitman-Yor model. We can appreciate that the distributions for the contaminated Pitman-Yor model shrink toward the observed values, while the Pitman-Yor model has a distortion in the produced estimates. 

\clearpage
\section{NGC 2419 data}\label{sup:mixt_app}

In order to analyze the NGC 2419 dataset we applied the sampling strategy of Section \ref{sup:mix}. We have tuned the variances of the Metropolis-Hastings steps to update $\vartheta$ and $\sigma$ to achieve optimal acceptance rates. We ran the model for $35\,000$ iterations, which include $10\,000$ burn-in iterations, thinning the produced chain every $10$ realizations, and obtaining a final sample  of size $2\,500$.
We report here the posterior summaries and the traceplots for the main parameters of the contaminated Pitman-Yor mixture model estimated with the NGC 2419 data.

\begin{table}[ht]
	\centering
	\begin{tabular}{lllll}
		\\
		& mean & SD & ESS      & Geweke diagnostic \\
		$\vartheta$ & 2.904         & 1.710          & 654.302  & -0.702            \\
		$\sigma$    & 0.116         & 0.032          & 302.821   & 1.373            \\
		$\beta$     & 0.925         & 0.034        	 & 1059.736   & -0.819            \\
		$k$   		& 20.491        & 4.429          & 636.313	  & 1.961            
	\end{tabular}
	\caption{Posterior summaries of the chain sampled for the contaminated Pitman-Yor mixture model. In evidence the main parameters of the mixing measure ($\vartheta$, $\sigma$ and $\beta$), and the number of different clusters $k$.}
	\label{tab:post_cPY_mix}
\end{table}

\begin{figure}[!h]
	\centering
	\includegraphics[width = 0.8 \textwidth]{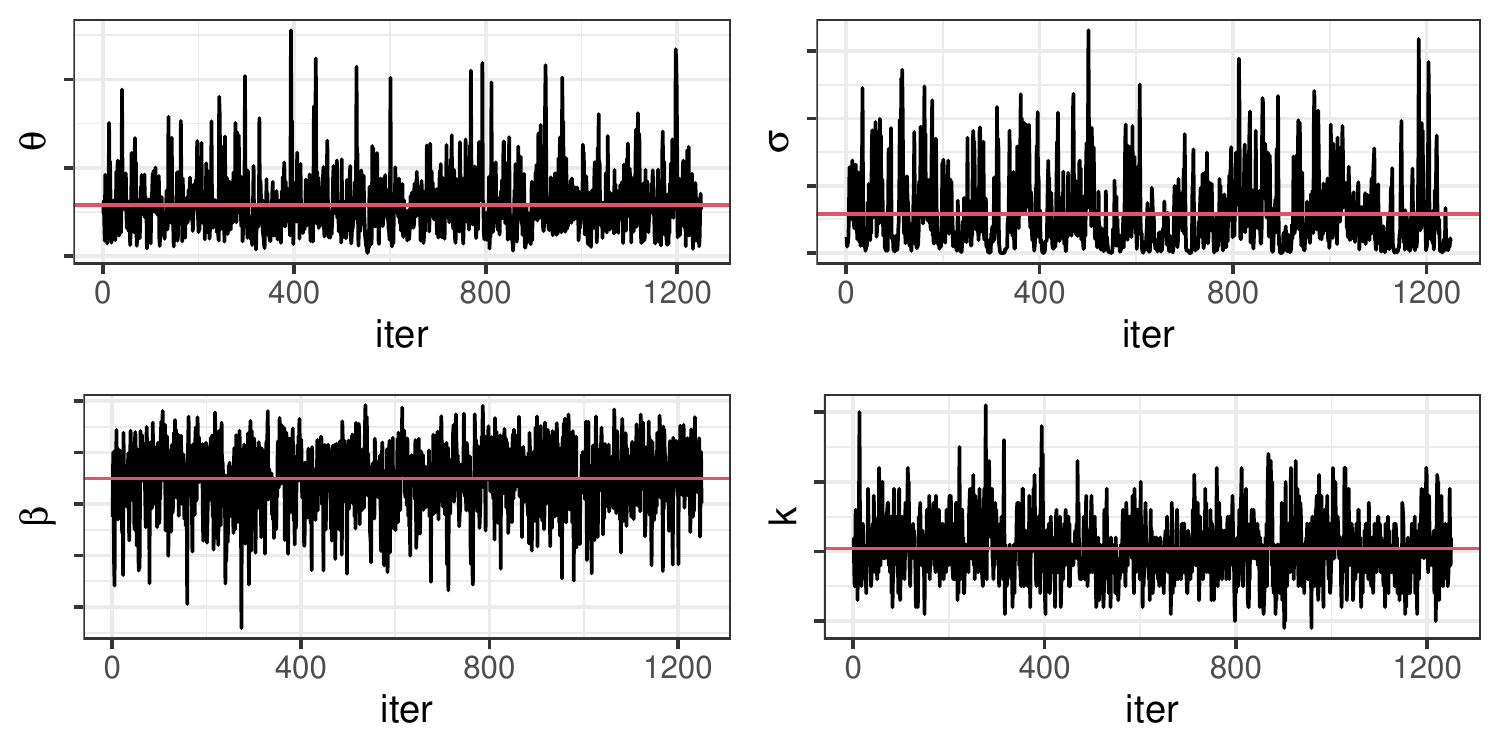}
	\caption{Traceplots for the contaminated Pitman-Yor mixture model. Top-left panel: traceplot of $\vartheta$. Top-right panel: traceplot of $\sigma$. Bottom-left panel: traceplot of $\beta$. Bottom-right panel: traceplot of the number of distinct clusters $k$.}\label{fig:trace_PY_mix}
\end{figure}

\begin{figure}[!h]
	\centering
	\includegraphics[width = 0.99 \textwidth]{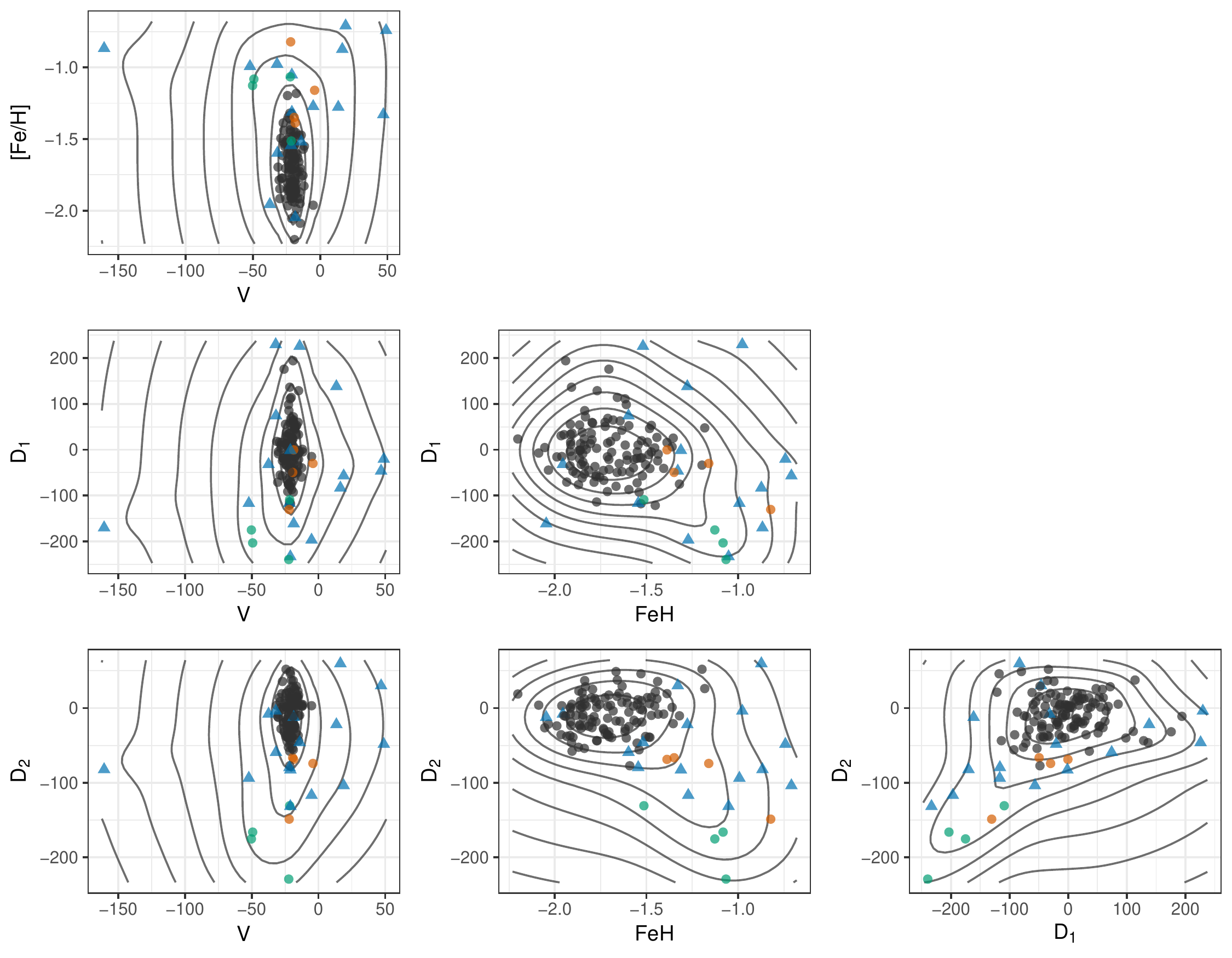}
	\caption{Estimated density and optimal partition for the NGC 2419 dataset. The blue triangles denote observations which are singletons in the posterior optimal partition, different colors denote different clusters. The contour lines denote the expectation of the estimated posterior random density.}\label{fig:application_plot}
\end{figure}

\subsection{Comparison with Pitman-Yor mixture model}\label{sec:mixt_comparison}
We compare the posterior inference faced for the contaminated Pitman-Yor mixture model, described in Section \ref{sec:cont_app}, with the standard Pitman-Yor mixture model case. We specify the Pitman-Yor mixture model preserving at most the commonalities with the contaminated model of Section \ref{sec:cont_app}, by considering the same multivariate Gaussian kernel function $\kernel(\cdot;(\mu,\Sigma))$, with expectation $\mu$ and covariance matrix $\Sigma$. We further maintain the same specification of the base measure, i.e. $Q_0 \sim NIW(\mu_0, \kappa_0, \nu_0, S_0)$ is a Normal-Inverse-Wishart distribution. We specify the parameters of the base measure by setting $\mu_0$ equals to the sample mean of the data, $\kappa_0 = 1$, $\nu_0 = d + 3 = 7$ and $S_0$ equals to the diagonal of the sample variance of the data. We complete the model specification by choosing vague priors for the parameters of the mixing measure, $\vartheta \sim \mathrm{Gamma}(2,0.02)$, $\sigma \sim \mathrm{Unif}(0,1)$ and $\beta \sim \mathrm{Unif}(0,1)$. Posterior inference is carried out using a marginal sampling scheme \citep{Esc88, Esc95}, in the spirit of the algorithm described in section \ref{sup:mix}, but without the contaminant measure. We tune variances of the Metropolis-Hastings steps to update $\vartheta$ and $\sigma$ to achieve optimal acceptance rates. We ran the model for $35\,000$ iterations, which include $10\,000$ burn-in iterations. The produced chain is thinned every $10$ realizations, thus obtaining a final sample  of size $2\,500$.

\begin{table}[h]
	\centering
	\begin{tabular}{llc*{5}{c}}
		\\
		&       && \multicolumn{5}{c}{CPY partition} \\ 
		&       && \textit{Singletons}    & \textit{A}   & \textit{B}    & \textit{C}  \vspace{5pt}\\
		&&\multicolumn{1}{c}{\textit{total}}&\textit{16}&\textit{115}&\textit{4}&\textit{4}\\ 
		\multirow{4}{*}{PY partition} & \textit{Singletons$_{\mathrm{PY}}$}   &\multicolumn{1}{c}{\textit{7}}& 5  & 2 & 0 &0   \\
		& \textit{A$_{\mathrm{PY}}$} &\multicolumn{1}{c}{\textit{113}}& 3 & 109 & 0 &1 \\
		& \textit{B$_{\mathrm{PY}}$} & \multicolumn{1}{c}{\textit{15}}& 7 &1 &4&3 \\
		& \textit{C$_{\mathrm{PY}}$} & \multicolumn{1}{c}{\textit{2}}& 0 &2 &0&0      \\
		& \textit{D$_{\mathrm{PY}}$} & \multicolumn{1}{c}{\textit{2}}& 1 &1 &0&0      \\
	\end{tabular}
	\caption{Comparison between the partition estimated using a contaminated Pitman-Yor mixture model and the optimal partition estimated using a Pitman-Yor mixture model.}
	\label{tab:vsPY}
\end{table}

Table \ref{tab:vsPY} shows a comparison between the optimal partition estimated with a contaminated Pitman-Yor mixture model and the optimal partition estimated using a Pitman-Yor mixture model. Only $5$ of the $16$ stars identified as singletons in the contaminated mixture model are singletons also in the standard mixture model, while the optimal partition for the standard Pitman-Yor mixture model shows a total of $7$ singletons. Out of $115$ stars belonging to the main cluster in the contaminated model, $109$ belong to the main cluster also in the standard model. The standard Pitman-Yor mixture model, with respect to the contaminated mixture model, is producing overall a partition with fewer singletons, and with the remaining stars less concentrated in the main cluster. 

\begin{figure}[!h]
	\centering
	\includegraphics[width = 0.99 \textwidth]{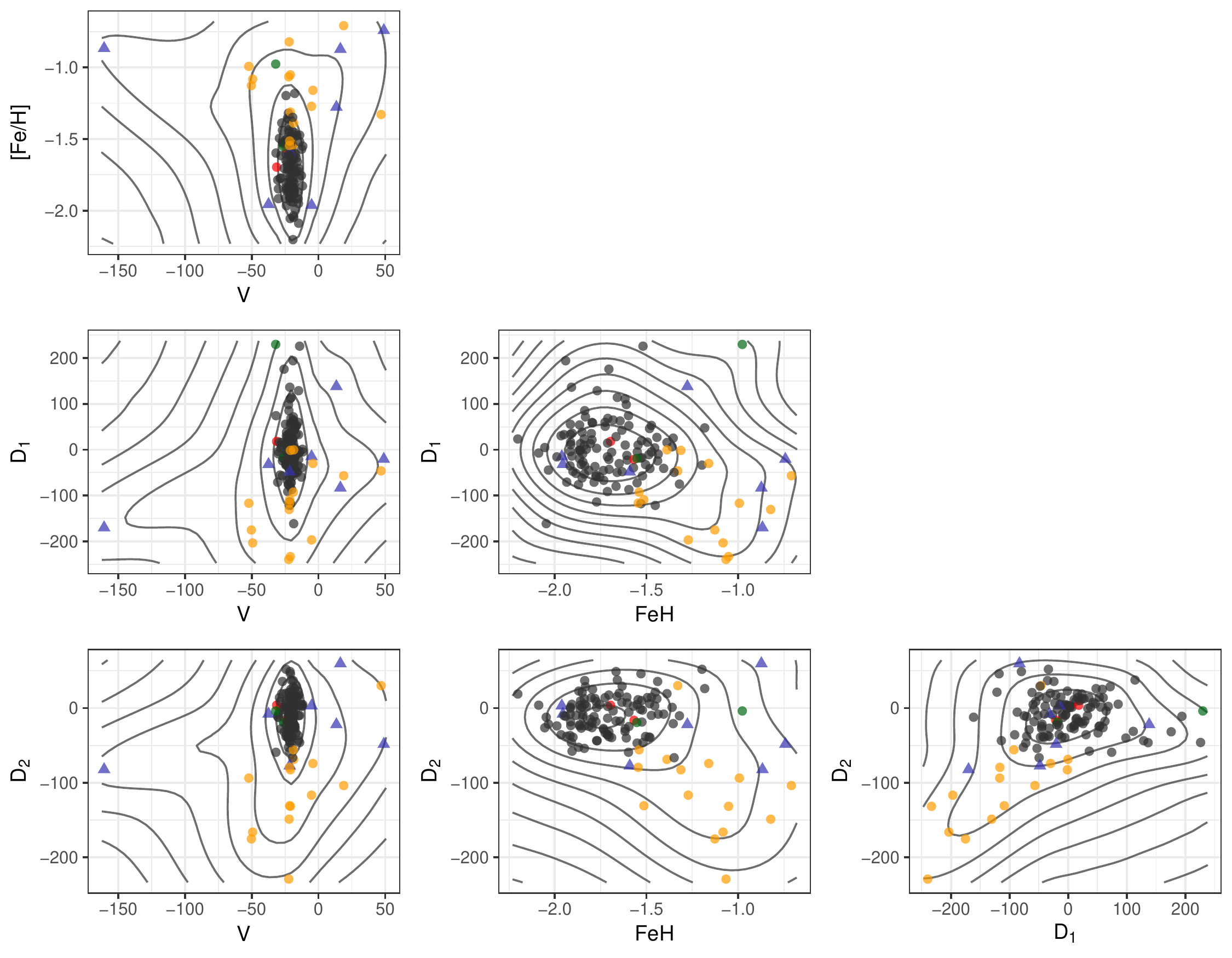}
	\caption{Estimated density and optimal partition for the NGC 2419 dataset with the Pitman-Yor mixture model. The blue triangles denote observations which are singletons in the posterior optimal partition, different colors denote different clusters. The contour lines denote the expectation of the estimated posterior random density.}\label{fig:application_plot_PY}
\end{figure}

\end{document}